\documentclass[11pt]{article}

\usepackage[
    a4paper,
    left=2.5cm,
    right=2.5cm,
    top=2.8cm,
    bottom=2.8cm
]{geometry}

\usepackage[T1]{fontenc}
\usepackage[utf8]{inputenc}
\usepackage{lmodern}
\usepackage{microtype}
\usepackage{xcolor}
\usepackage{amsmath}
\numberwithin{equation}{section}
\usepackage{amsthm}
\usepackage{amssymb}
\usepackage{empheq}
\usepackage{physics}
\usepackage{mathtools}
\usepackage{enumitem}
\usepackage{xspace}
\usepackage{dsfont}
\usepackage[colorlinks=true, linkcolor=black, urlcolor=black, citecolor=blue]{hyperref}
\usepackage{booktabs}
\usepackage{tikz}
\usepackage{placeins}
\usepackage{etoolbox}

\newcommand{\C}{\mathbb{C}}

\newcommand{\N}{\mathbb{N}}

\newcommand{\T}{\mathbb{T}}
\newcommand{\Z}{\mathbb{Z}}
\newcommand{\R}{\mathbb{R}}
\newcommand{\1}{\mathds{1}}

\newcommand{\cA}{\mathcal{A}}
\newcommand{\cB}{\mathcal{B}}
\newcommand{\cD}{\mathcal{D}}
\newcommand{\cE}{\mathcal{E}}
\newcommand{\cF}{\mathcal{F}}
\newcommand{\cG}{\mathcal{G}}

\newcommand{\cK}{\mathcal{K}}
\newcommand{\cL}{\mathcal{L}}

\newcommand{\cC}{\mathcal{C}}

\newcommand{\cO}{\mathcal{O}}

\newcommand{\cR}{\mathcal{R}}
\newcommand{\cS}{\mathcal{S}}
\newcommand{\cT}{\mathcal{T}}

\newcommand{\cY}{\mathcal{Y}}

\newcommand{\mi}{\mathrm{i}}
\newcommand{\me}{\mathrm{e}}
\newcommand{\md}{\mathrm{d}}

\newtheorem{theorem}{Theorem}[section]

\newtheorem{proposition}[theorem]{Proposition}
\newtheorem{lemma}[theorem]{Lemma}

\theoremstyle{definition}

\newtheorem{example}[theorem]{Example}

\DeclareMathOperator{\dom}{Dom}

\DeclareMathOperator{\ran}{Ran}
\DeclareMathOperator{\supp}{supp}
\DeclareMathOperator{\diag}{diag}

\title{Trace-class spectra of irreducible\\
Gaussian quantum Markov semigroups}

\author{
Franco Fagnola$^{1}$ \qquad Zheng Li$^{2,*}$
\\[1.2ex]
\small $^{1}$Dipartimento di Matematica, Politecnico di Milano\\
\small Piazza Leonardo da Vinci 32, 20133 Milano, Italy
\\[0.6ex]
\small $^{2}$School of Mathematics and Statistics, Central South University\\
\small Changsha 410083, China
\\[0.8ex]
\small \href{mailto:franco.fagnola@polimi.it}{franco.fagnola@polimi.it}
\qquad
\href{mailto:zheng.li@csu.edu.cn}{zheng.li@csu.edu.cn}
\\[0.4ex]
\small $^{*}$Corresponding author
}

\date{}

\makeatletter
\AtBeginDocument{
    \patchcmd{\@citex}{,\penalty\@m\ }{;\penalty\@m\ }
        {}{\PackageError{main}{Unable to set citation delimiter}{}}
}
\patchcmd{\thebibliography}{\advance\leftmargin\labelsep}
    {\setlength{\labelsep}{1em}\advance\leftmargin\labelsep}
    {}{\PackageError{main}{Unable to set bibliography spacing}{}}
\makeatother

\allowdisplaybreaks

\begin{document}

    \maketitle

    \begin{abstract}
        We determine the spectrum of the predual generator of an irreducible Gaussian quantum Markov semigroup on trace-class operators, the natural space for density operators. The drift matrix governs the linear part of the Gaussian evolution and organizes the spectral classification. When the linear flow generated by this matrix is not periodic, the spectrum is the closed left half-plane. If all drift eigenvalues have negative real part, every interior point is an eigenvalue. There are no trace-class eigenvectors if some drift eigenvalue has positive real part or all have zero real part. For periodic flow, the spectrum consists of unions of horizontal half-lines or filled parabolic regions. Examples show that irreducibility cannot in general be omitted. We apply the results to oscillators, free particles, and boson chains.
    \end{abstract}

    {\small
        \par
    }

    \vskip1cm

    {\small
        \tableofcontents
    }

    \section{Introduction} \label{section-introduction}

    Gaussian quantum Markov semigroups (QMSs) describe open bosonic systems with Hamiltonians that are at most quadratic in creation and annihilation operators and with linear noise operators; their definition is recalled in Section~\ref{subsection-gaussian-semigroups}. Their Schr\"odinger-picture generators are commonly called \emph{quadratic bosonic Liouvillians}, and the dynamics is also described as \emph{quasi-free} \cite{ProsenSeligman2010BosonOperatorSpaces,BarthelZhang2022QuasiFree}. The relation to classical Ornstein--Uhlenbeck semigroups can be seen by restricting the quantum evolution to commuting observables. More precisely, every finite-dimensional classical Ornstein--Uhlenbeck semigroup with nondegenerate diffusion and real linear drift can be realized as the restriction of a Gaussian QMS to the algebra generated by commuting position observables \cite[Section~3, Propositions~8--9]{fagnola2025spectral}.

    This restriction describes only part of the quantum dynamics. The action on Weyl operators, defined in Section~\ref{subsection-fock-weyl-gaussian}, is determined by the finite-dimensional drift, diffusion and displacement parameters of Theorem~\ref{theo-explicit-action-on-weyl-operators} \cite{agredo2022decoherence,fagnola2025spectral,girotti2025gaussian}. On the full quantum algebra, the canonical commutation relations enter both the complete-positivity condition \eqref{eq-quantum-diffusion-matrix} and the uncertainty relation for Gaussian states. In addition, a quadratic Hamiltonian may contribute to the phase-space drift and affect spectral quantities such as the $ L^2 $ relaxation rate \cite{fagnola2025spectral}. These parameters connect the quantum evolution with classical Ornstein--Uhlenbeck dynamics, but a spectral description of the quantum generator also requires a choice of operator space.

    Let $ \mathsf{h}=\Gamma(\C^d) $, with $ d\ge1 $, be the $ d $-mode bosonic Fock space, and let $ (\cT_t)_{t\ge0} $ be a Gaussian QMS on $ \cB(\mathsf{h}) $. Its predual $ (\cT_{*t})_{t\ge0} $ is a strongly continuous contraction semigroup on the trace-class space $ \cB_1(\mathsf{h}) $. We study its closed generator $ \cL_* $ (defined in Section~\ref{subsection-gaussian-semigroups}), whose spectrum we call the \emph{trace-class spectrum}, or $ L^1 $-spectrum. This is the natural Banach-space setting for arbitrary normal states and their trace-class perturbations. The physical interest in generator spectra is also reflected in finite-dimensional studies of random Lindblad operators \cite{Denisov2019UniversalSpectra} and constraints on relaxation rates \cite{Chruscinski2021RelaxationRates}. In infinite dimensions, specifying the operator space is essential. When a Gaussian invariant density exists, a central question is whether the \emph{polynomial eigenvalues}, namely the nonnegative integer combinations of drift eigenvalues displayed in \eqref{eq-drift-eigenvalue-combinations}, exhaust this spectrum. Their algebraic construction gives explicit eigenoperators, described in Section~\ref{subsection-polynomial-eigenvalues}, but does not by itself exclude additional eigenvalues, continuous spectrum, or residual spectrum. This distinction is relevant to resolvent estimates and uniform convergence in trace norm.

    The dependence on the ambient space already occurs classically. For an Ornstein--Uhlenbeck semigroup with stable drift (see Section~\ref{subsection-main-results} for the drift classification) and a nondegenerate invariant Gaussian measure, Metafune, Pallara, and Priola proved that its spectrum on $ L^p $ of that measure, for $ 1<p<\infty $, consists of nonnegative integer combinations of the drift eigenvalues. On $ L^1 $ of the invariant measure, the spectrum is instead the closed left half-plane, and every interior point is an eigenvalue \cite[Theorems~3.1 and~5.1]{MetafunePallaraPriola2002}. Metafune also studied the spectrum on unweighted $ L^p $ spaces \cite{Metafune2001LpSpectrumOU}, and Fornaro, Metafune, Pallara, and Schnaubelt extended this analysis to degenerate hypoelliptic diffusion \cite[Theorems~1.1 and~2.2]{FMPS21}.

    In the quantum setting, discrete $ L^2 $-spectral descriptions are known for several realizations. Cipriani, Fagnola, and Lindsay studied symmetric quantum Ornstein--Uhlenbeck semigroups \cite{CiprianiFagnolaLindsay2000}. For a faithful normal invariant density $ \rho $, the one-mode spectra associated with the embeddings
    \begin{equation*}
        X\longmapsto\rho^{s/2}X\rho^{(1-s)/2},\qquad 0\le s\le1,
    \end{equation*}
    were studied in \cite{fagnola_li_2025spectral}. Spectral gaps in the Gelfand--Naimark--Segal (GNS) and Kubo--Martin--Schwinger (KMS) embeddings were computed in \cite{fagnola2025spectral}; see also \cite{Li2025KMS}. Wirth proved that, for QMSs with a faithful normal invariant state, the GNS spectral gap does not exceed the KMS spectral gap, with both gaps understood as optimal exponential contraction rates with prefactor one on centered observables \cite[Corollary~3.7]{Wirth2026KMSGNS}. These results concern induced generators on the Hilbert--Schmidt space $ \cB_2(\mathsf{h}) $ and do not determine the spectrum of $ \cL_* $ on $ \cB_1(\mathsf{h}) $.

    Relating these classical spectral results and quantum $ L^2 $ realizations to the trace-class generator requires further arguments. The Fourier--Weyl transform, defined in Section~\ref{subsection-fock-weyl-gaussian}, gives an $ L^2 $ correspondence for Hilbert--Schmidt operators \cite[Section~2.3]{KeylKiukasWerner2016}, but does not identify $ \cB_1(\mathsf{h}) $ with classical $ L^1 $. A solution of the transformed eigenvalue equation need not be the characteristic function of a trace-class operator. Moreover, quantum irreducibility, recalled in Section~\ref{subsection-irreducibility-invariant}, is characterized by a condition on complex invariant subspaces involving both drift and diffusion, whereas classical hypoellipticity uses real invariant subspaces \cite[Theorems~8 and~12]{FagnolaGirotti2026Irreducibility}. A trace-class spectral analysis must therefore establish trace-class membership, control trace norms, and justify the relevant generator domains.

    For an irreducible Gaussian QMS with stable drift, we also give a direct link with the classical invariant-measure result. Integration against displaced Gaussian states defines a bounded injective map from $ L^1 $ of a nondegenerate Gaussian measure into $ \cB_1(\mathsf{h}) $, intertwining a classical Ornstein--Uhlenbeck semigroup with the quantum predual semigroup; see Appendix~\ref{subsection-stable-classical-comparison}. This construction transfers the eigenvalues in the open left half-plane from \cite[Theorem~5.1]{MetafunePallaraPriola2002}. Our direct quantum proof, based on characteristic functions and an auxiliary Gaussian diffusion semigroup, also provides estimates needed for mixed stable--critical drift (Section~\ref{subsection-mixed-stable-critical}), where the invariant Gaussian measure used in this argument is unavailable. The present work thus gives a trace-class spectral classification for irreducible finite-mode Gaussian QMSs in all drift regimes, including nonstationary dynamics, and determines the finer spectral types specified below.

    \subsection{Main results and scope} \label{subsection-main-results}

    Our standing assumption is \emph{irreducibility}: the only orthogonal projections $ p $ satisfying $ \cT_t(p)\ge p $ for every $ t\ge0 $ are $ 0 $ and $ \1 $. Fagnola and Girotti characterize this property by a condition on complex invariant subspaces of the drift and quantum diffusion matrices \cite[Theorem~12]{FagnolaGirotti2026Irreducibility}; we recall it in Proposition \ref{proposition-irreducibility}. For stable drift, irreducibility is equivalent to the existence of a unique faithful normal invariant state \cite[Proposition~15]{girotti2025gaussian}.

    Let $ \mathbf{Z} $, $ \mathbf{C} $, and $ \boldsymbol{\zeta} $ denote the real drift matrix, diffusion matrix, and displacement vector introduced in Section~\ref{subsection-gaussian-semigroups}, with the real-coordinate convention \eqref{eq-definition-iota}. We classify the drift by its spectral abscissa
    \begin{equation*}
        s(\mathbf{Z}) = \max\{ \Re\lambda : \lambda\in\sigma(\mathbf{Z}) \}.
    \end{equation*}
    We call the drift \emph{stable} if $ s(\mathbf{Z})<0 $, \emph{strictly unstable} if $ s(\mathbf{Z})>0 $, and \emph{critical} if $ s(\mathbf{Z})=0 $. Thus strict instability requires only one eigenvalue with positive real part and allows other eigenvalues in the closed left half-plane.

    The following theorem summarizes the spectrum in all drift regimes. We call the drift \emph{periodic} if $ \me^{t\mathbf{Z}}=\1 $ for some $ t>0 $, and \emph{nonperiodic} otherwise. The accumulated diffusion $ \mathbf{S}_t $ and displacement $ \mathbf{b}_t $ are defined in \eqref{eq-definition-boldface-S-t} and \eqref{eq-definition-b-t}.

    \begin{theorem}[Trace-class spectrum] \label{theorem-global-trace-class-spectrum}
        Suppose that $ (\cT_t)_{t\ge0} $ is an irreducible Gaussian QMS.
        \begin{enumerate}
            \item If $ \mathbf{Z} $ is nonperiodic, then
            \begin{equation*}
                \sigma(\cL_*)=\{\lambda\in\C:\Re\lambda\le0\}.
            \end{equation*}
            \item If $ \mathbf{Z}=0 $, then
            \begin{equation*}
                \sigma(\cL_*)=\left\{-\frac12\mathbf{z}^T\mathbf{C}\mathbf{z}+\mi\boldsymbol{\zeta}^T\mathbf{z}:\mathbf{z}\in\R^{2d}\right\}.
            \end{equation*}
            \item If $ \mathbf{Z}\neq0 $ is periodic with least positive period $ \tau $, then
            \begin{equation*}
                \sigma(\cL_*)=\bigcup_{n\in\Z}\left\{-\frac{1}{2\tau}\mathbf{z}^T\mathbf{S}_\tau\mathbf{z}+\frac{\mi}{\tau}\bigl(\mathbf{b}_\tau^T\mathbf{z}+2\pi n\bigr):\mathbf{z}\in\R^{2d}\right\}.
            \end{equation*}
        \end{enumerate}
    \end{theorem}

    The theorem follows from the results in Sections \ref{eq-stable-drift}--\ref{section-critical-drift-aperiodic}. The periodic spectra are described geometrically in Table \ref{tab-periodic-spectrum}. These formulas determine the spectrum in every case; the spectral types are distinguished as follows.

    For stable and strictly unstable drift, the spectral types differ as shown in Table \ref{tab-stable-unstable-spectrum}. In the stable case, every point in the open left half-plane is an eigenvalue, and zero is the only eigenvalue on the imaginary axis. Our stable-drift theorem extends the one-mode thermal result in \cite[Proposition~3.1]{BeckerRouzeSalzmann2026}. In the strictly unstable case, there are no trace-class eigenvectors; the open left half-plane and zero belong to the residual spectrum. The whole imaginary axis belongs to the approximate point spectrum in both regimes. These assertions are proved in Theorems \ref{theorem-stable-spectrum}, \ref{theorem-whole-spectrum-strictly-unstable-drift}, and \ref{theorem-unstable-spectral-types}.

    \begin{table}[htbp]
        \centering
        \small
        \renewcommand{\arraystretch}{1.35}
        \begin{tabular}{@{}p{3.1cm}p{5.0cm}p{5.1cm}@{}}
            \toprule
            Spectral region & Stable drift & Strictly unstable drift \\
            \midrule
            $ \{\eta : \Re\eta < 0\} $ & Point spectrum & Residual spectrum \\
            $ \{0\} $ & Point spectrum & Residual and approximate point spectrum \\
            $ \mi\R\setminus\{0\} $ & Approximate point spectrum; no eigenvalues & Approximate point spectrum; no eigenvalues \\
            \bottomrule
        \end{tabular}
        \caption{Spectral properties of $ \cL_* $ under irreducibility. In both regimes, $ \sigma(\cL_*) = \{\eta \in \C : \Re\eta \le 0\} $. Approximate point spectrum and residual spectrum need not be disjoint; see Section \ref{section-spectral-notation}.}
        \label{tab-stable-unstable-spectrum}
    \end{table}

    For nonperiodic critical drift, the approximate point spectrum is the whole closed left half-plane. If the drift is purely critical, the point spectrum is empty, whether or not $ \mathbf{Z} $ is semisimple. These statements are proved in Theorems \ref{theorem-spectrum-mixed-stable-critical-drift}, \ref{theorem-spectrum-purely-critical-semisimple-drift}, and \ref{theorem-spectrum-purely-critical-nonsemisimple-drift}. Periodic drift also has empty point spectrum by Theorems \ref{theorem-zero-drift-spectral-types} and \ref{theorem-spectrum-predual-generator-periodic-critical-drift}; at zero drift, the entire spectrum is residual. The point spectrum in the general mixed stable--critical case, and the residual or continuous character of every spectral value for nonzero periodic drift, are not determined here.

    Irreducibility cannot be omitted in general. Section \ref{section-reducible-semigroups} compares a reducible noiseless semigroup with an irreducible noisy one having the same drift and displacement. At zero drift, the identity semigroup has spectrum $ \{0\} $, whereas the quantum heat semigroup has spectrum $ (-\infty,0] $. For the nonzero nilpotent drift considered there, the unitary evolution generated by the free-particle Hamiltonian gives spectrum $ \mi\R $, whereas adding a strictly positive diffusion matrix gives the closed left half-plane. Thus the formulas obtained under irreducibility can fail when this assumption is dropped, even with the drift fixed.

    The stable result also clarifies the meaning of relaxation rates. Let $ \rho $ be the invariant density in the stable case. For every real $ \eta<0 $ there is a nonzero $ Y_\eta\in\dom\cL_* $ such that $ \cL_*(Y_\eta)=\eta Y_\eta $ and $ \Tr Y_\eta=0 $. Since such eigenvalues approach zero, no constants $ M<\infty $ and $ g>0 $ can satisfy
    \begin{equation*}
        \norm{\cT_{*t}(Y)-(\Tr Y)\rho}_1
        \le M\me^{-gt}\norm{Y}_1,
        \qquad Y\in\cB_1(\mathsf{h}), \quad t\ge0,
    \end{equation*}
    on the whole trace-class space. This does not preclude exponential rates on weighted Hilbert-space realizations, for finite-degree observables, or for classes of states satisfying additional bounds.

    If $ \mathbf{Z}\ne0 $, Theorem~\ref{theorem-global-trace-class-spectrum} also implies that the predual semigroup is not eventually norm continuous, and hence is not eventually differentiable on $ \cB_1(\mathsf{h}) $. Indeed, its generator spectrum contains an unbounded subset of the imaginary axis, which is incompatible with eventual norm continuity by \cite[Theorem~II.4.18]{EngelNagel2000}.

    Section~\ref{section-applications} illustrates these spectral results with three physical models. For the thermal oscillator, the full trace-class spectrum extends beyond the polynomial eigenvalues. For a free particle with position noise, the spectrum is the closed left half-plane and there are no trace-class eigenoperators. The third application considers a driven boson chain with noise only at its endpoints, based on the diffusive reservoirs of \cite{Nicacio2015}. Its two- and three-mode versions have periodic drift. The two-mode spectrum consists of horizontal half-lines, whereas a nonzero coherent drive produces a repeated family of filled parabolic regions in the three-mode case. This illustrates how constant driving can change the trace-class spectrum while leaving the drift matrix unchanged.

    The proofs work directly with characteristic functions and trace-norm estimates. For stable drift, division by the characteristic function of the invariant Gaussian state reduces the eigenvalue equation to transport, while an auxiliary Gaussian diffusion semigroup realizes the resulting functions as trace-class operators. For strictly unstable drift, bounded eigenoperators of the dual semigroup determine residual spectral values, and diffusion along expanding trajectories excludes trace-class eigenvectors. Periodic critical drift is handled through the Gaussian Weyl-covariant channel obtained at a return time. For mixed stable--critical drift, a discounted Lyapunov integral and time averaging yield trace-class approximate eigenvectors. In the semisimple purely critical case, a compact rotation flow and fixed-support estimates for the inverse Fourier--Weyl transform provide the approximate spectrum and exclude eigenvalues. The nonsemisimple purely critical case uses a Jordan chain and rescaled Laplace integrals to construct approximate eigenvectors; the details are given in Appendix \ref{appendix-purely-critical-nonsemisimple}.

    \newpage
    \subsection{A critical discussion of spectral claims in physics}
    \label{subsection-polynomial-eigenvalues}

    Damping bases and third quantization give an algebraic description of eigenoperators for quadratic bosonic generators admitting a Gaussian invariant state \cite{ProsenSeligman2010BosonOperatorSpaces,HondaNakazatoYoshida2010SpectralResolution,BarthelZhang2022QuasiFree}. In this invariant-state setting, let $ \lambda_1,\ldots,\lambda_{2d} $ be the eigenvalues of the real drift matrix $ \mathbf{Z} $ in \eqref{eq-gaussian-weyl-action}, counted with algebraic multiplicity. We call the numbers in the family
    \begin{equation} \label{eq-drift-eigenvalue-combinations}
        \left\{\sum_{j=1}^{2d}n_j\lambda_j:n_1,\ldots,n_{2d}\in\N\right\},
        \qquad \N=\{0,1,2,\ldots\},
    \end{equation}
    the \emph{polynomial eigenvalues}. The name refers to the form of their associated eigenoperators: their characteristic functions are polynomials in the phase-space coordinates multiplied by the characteristic function of a Gaussian invariant state. At the operator level, they are finite linear combinations of products obtained by multiplying the Gaussian density on either side by polynomials in creation and annihilation operators.

    A Gaussian QMS preserves Gaussian states, but its trace-class state space also contains non-Gaussian densities. Constructing the eigenoperators just described does not prove that the polynomial eigenvalues exhaust the spectrum. Nor does a formal change of variables prove similarity of closed operators on a specified Banach space. The following discussion examines how these constructions are interpreted in physics papers on damped oscillators, quadratic bosonic generators, and Gaussian channels. Several of the spectral and relaxation claims in these papers require qualification when applied to the full trace-class space, as we explain below.

    Bern\'ad and Torres \cite[Appendix~A.3]{BernadTorres2015Partly} identify the damped-oscillator spectrum with a discrete family of damping eigenvalues. This description does not give the full trace-class spectrum of the positive-temperature stable oscillator: by Theorem~\ref{theorem-stable-spectrum}, that spectrum is the closed left half-plane and every interior point is an eigenvalue. Similarly, the statement in \cite[Section~III.B, following Eq.~(34)]{McDonaldClerk2023ThirdQuantization} that observable decay rates are confined to integer multiples of the basic damping rate does not extend to arbitrary normal initial states. The calculations for finite-degree observables need to be distinguished from such an unrestricted relaxation statement.

    The choice of operator space explains why algebraic calculations can lead to different spectral descriptions. Prosen and Seligman use an $ \ell^2 $ topology on coefficient sequences \cite[Sections~2--3]{ProsenSeligman2010BosonOperatorSpaces}; Honda, Nakazato, and Yoshida use a Fr\'echet test space whose dual accommodates unbounded left eigenoperators \cite[Section~IV, Eq.~(23)]{HondaNakazatoYoshida2010SpectralResolution}. Such left eigenoperators need not define continuous functionals on $ \cB_1(\mathsf{h}) $. Barthel and Zhang obtain a triangular bosonic representation with diagonal entries \eqref{eq-drift-eigenvalue-combinations} \cite[Section~IV.5, Proposition~11]{BarthelZhang2022QuasiFree}. Its discrete spectral description and positive dissipative gap do not describe the full trace-class generator, even for stable irreducible drift. Triangularity or density of the constructed eigenvectors does not replace control of domains and resolvents.

    The choice of operator space also matters for transformations of the generator. A bounded similarity on $ \cB_1(\mathsf{h}) $ requires both the transformation and its inverse to be bounded; see \cite[Section~II.C and Notes~18, 22]{McDonaldClerk2023ThirdQuantization}. Pointwise inversion of a Gaussian multiplier alone does not establish an isomorphism of the relevant spaces \cite[Appendix~A.1--A.2]{QuintelaRodriguez2026Exceptional}. The resulting discrete-spectrum and exceptional-point descriptions therefore do not by themselves give a classification of the closed trace-class generator. These distinctions concern the interpretation of the spectra and relaxation rates on the full trace-class space; the algebraic eigenoperators and results established on explicitly specified smaller spaces retain their value.

    The paper is organized as follows: Section \ref{section-gaussian-framework} introduces the Gaussian framework and spectral notation. Sections \ref{eq-stable-drift}--\ref{section-critical-drift-aperiodic} contain the irreducible spectral classification, and Section \ref{section-reducible-semigroups} gives the reducible counterexamples. Section~\ref{section-applications} presents the three physical applications. Appendix \ref{section-real-linear-operators} records the real-linear conventions. The remaining appendices contain the proofs for the drift regimes and the reducible examples.

    \section{Gaussian framework and irreducibility} \label{section-gaussian-framework}

    We use the Fock-space and Weyl-operator conventions of \cite{fagnola2025spectral,fagnola_li_2025spectral,Li2025KMS}. We describe Gaussian QMSs through their action on Weyl operators and recall the corresponding generalized GKSL form of their generators.

    \subsection{Fock space, Weyl operators, and Gaussian states} \label{subsection-fock-weyl-gaussian}

    Let $ \mathsf{h} = \Gamma(\C^d) \simeq \Gamma(\C)^{\otimes d} $, with canonical orthonormal basis
    \begin{equation*}
        \{e(n_1,\ldots,n_d) : n_1,\ldots,n_d \in \N\}.
    \end{equation*}
    Let $ \cD $ be the finite linear span of this basis, called the finite-particle domain. The annihilation and creation operators act on $ \cD $ by
    \begin{align*}
        a_j e(n_1,\ldots,n_d) &= \sqrt{n_j}\,e(n_1,\ldots,n_j-1,\ldots,n_d), \\
        a_j^\dagger e(n_1,\ldots,n_d) &= \sqrt{n_j+1}\,e(n_1,\ldots,n_j+1,\ldots,n_d).
    \end{align*}
    The first expression is zero when $ n_j = 0 $. On this common invariant domain, $ [a_j,a_k^\dagger] = \delta_{jk}\1 $. We denote the total number operator by $ N $, so that
    \begin{equation*}
        N=\sum_{j=1}^d a_j^\dagger a_j, \qquad
        N e(n_1,\ldots,n_d)=(n_1+\cdots+n_d)e(n_1,\ldots,n_d).
    \end{equation*}
    Scalar products are linear in the second variable. For $ z \in \C^d $, write
    \begin{equation*}
        a(z) := \sum_{j=1}^d \overline{z_j}a_j, \qquad
        a^\dagger(z) := \sum_{j=1}^d z_j a_j^\dagger.
    \end{equation*}
    The field operator
    \begin{equation*}
        p(z) := \frac{\mi}{\sqrt{2}}\bigl(a^\dagger(z)-a(z)\bigr)
    \end{equation*}
    is essentially self-adjoint on the finite-particle domain. We use the same notation for its closure and define $ W(z) := \me^{-\mi\sqrt{2}p(z)} $. The Weyl relations are
    \begin{equation*}
        W(z)W(w) = \me^{-\mi\Im\langle z,w\rangle}W(z+w), \qquad W(z)^* = W(-z).
    \end{equation*}
    The map $ z \mapsto W(z) $ is strongly continuous, and the linear span of the Weyl operators is $ \sigma $-weakly dense in $ \cB(\mathsf{h}) $.

    We identify $ \C^d $, regarded as a real Hilbert space, with $ \R^{2d} $ by
    \begin{equation} \label{eq-definition-iota}
        \iota(z) = \mathbf{z} := \begin{bmatrix}\Re z\\ \Im z\end{bmatrix}.
    \end{equation}
    The real matrix representing a real-linear operator $ S $ is denoted by $ \mathbf{S} $, and $ S^\sharp $ is its adjoint with respect to $ \Re\langle\cdot,\cdot\rangle $. Thus $ S^\sharp $ is represented by $ \mathbf{S}^T $; explicit formulas are given in Appendix \ref{section-real-linear-operators}. We write
    \begin{equation*}
        \mathbf{J} := \begin{bmatrix}0&\1_d\\-\1_d&0\end{bmatrix},
        \qquad \Im\langle z,w\rangle = \mathbf{z}^T\mathbf{J}\mathbf{w}.
    \end{equation*}
    The associated real-linear operator is
    \begin{equation*}
        Jz = -\mi z, \qquad J^\sharp = J^{-1} = -J.
    \end{equation*}
    For finite-dimensional real or complex vectors, we write $ \abs{\cdot} $ for the Euclidean norm; in particular, $ \abs{z}=\abs{\mathbf{z}} $. A function on $ \C^d $ will also be viewed as a function on $ \R^{2d} $ through $ \iota $, without a change of notation.

    We denote by $ \cB(\mathsf{h}) $, $ \cB_1(\mathsf{h}) $, and $ \cB_2(\mathsf{h}) $ the spaces of bounded, trace-class, and Hilbert--Schmidt operators, respectively. Their norms are the operator norm $ \norm{\cdot} $, the trace norm $ \norm{Y}_1=\Tr\abs{Y} $, and the Hilbert--Schmidt norm $ \norm{Y}_2=(\Tr(Y^*Y))^{1/2} $. For $ Y \in \cB_1(\mathsf{h}) $, its \emph{quantum characteristic function}, also called its \emph{Fourier--Weyl transform}, is
    \begin{equation*}
        \chi_Y(z) := \Tr(YW(z)), \qquad z \in \C^d.
    \end{equation*}
    We use the shorter term \emph{characteristic function} below. It is continuous and satisfies $ \abs{\chi_Y(z)} \le \norm{Y}_1 $. Continuity follows from finite-rank approximation and strong continuity of the Weyl operators. The $ \sigma $-weak density of the Weyl operators also shows that the Fourier--Weyl transform is injective: $ \chi_Y = 0 $ entails $ Y = 0 $.

    A normal state is identified with its positive trace-class density of trace one; it is \emph{faithful} if that density has trivial kernel. A normal state $ \rho $ is \emph{Gaussian} if its characteristic function has the form
    \begin{equation} \label{eq-characteristic-function-rho}
        \chi_\rho(z) = \exp{2\mi\Im\langle\omega,z\rangle - \frac{1}{2}\Re\langle z,Sz\rangle},
        \qquad z \in \C^d,
    \end{equation}
    where $ \omega \in \C^d $ is its mean vector and $ S $ is its covariance operator. We write $ \boldsymbol{\omega}:=\iota(\omega) $. The real covariance matrix $ \mathbf{S} $ is symmetric and satisfies $ \mathbf{S}+\mi\mathbf{J} \ge 0 $. Conversely, this condition characterizes Gaussian covariances; the state is faithful precisely when $ \mathbf{S}+\mi\mathbf{J} > 0 $ \cite{parthasarathy2010gaussian,holevo2011probabilistic}. Since $ \mathbf{S} $ and $ \mathbf{J} $ are real, either condition is equivalent to the corresponding condition with $ -\mi\mathbf{J} $. Every Gaussian characteristic function is nowhere zero.

    We shall also use exponential vectors
    \begin{equation*}
        e(\alpha) := \bigoplus_{n=0}^{\infty}\frac{\alpha^{\otimes n}}{\sqrt{n!}},
        \qquad \norm{e(\alpha)}^2 = \me^{\abs{\alpha}^2}, \quad \alpha \in \C^d.
    \end{equation*}
    The vacuum vector is $ e(0) $. The normalized coherent state $ \me^{-\abs{\alpha}^2}\ketbra{e(\alpha)}{e(\alpha)} $ has characteristic function
    \begin{equation} \label{eq-coherent-characteristic-function}
        \me^{-\abs{\alpha}^2}\langle e(\alpha),W(z)e(\alpha)\rangle
        = \exp{-\frac{1}{2}\abs{z}^2+2\mi\Im\langle\alpha,z\rangle}.
    \end{equation}

    \subsection{Gaussian semigroups} \label{subsection-gaussian-semigroups}

    A \emph{quantum Markov semigroup} $ (\cT_t)_{t\ge0} $ on $ \cB(\mathsf{h}) $ is a $ \sigma $-weakly continuous semigroup of normal, completely positive, identity-preserving maps. Its predual semigroup $ (\cT_{*t})_{t\ge0} $ on $ \cB_1(\mathsf{h}) $ is characterized by
    \begin{equation} \label{eq-predual-pairing}
        \Tr(\cT_{*t}(Y)X) = \Tr(Y\cT_t(X)),
        \quad Y \in \cB_1(\mathsf{h}), \quad X \in \cB(\mathsf{h}).
    \end{equation}
    The predual is a strongly continuous contraction semigroup and preserves the trace.
    Its generator $ \cL_* $ is defined by
    \begin{gather*}
        \dom\cL_* := \left\{Y\in\cB_1(\mathsf{h}) :
        \lim_{t\downarrow0}\frac{\cT_{*t}(Y)-Y}{t}\text{ exists in trace norm}\right\},\\
        \cL_*(Y) := \lim_{t\downarrow0}\frac{\cT_{*t}(Y)-Y}{t},\qquad Y\in\dom\cL_*.
    \end{gather*}
    It is closed and densely defined. Its Banach adjoint, under the bilinear trace pairing \eqref{eq-predual-pairing}, is denoted by $ \cL $; it is the $ \sigma $-weak generator of $ (\cT_t)_{t\ge0} $. In particular,
    \begin{equation*}
        \Tr(\cL_*(Y)X) = \Tr(Y\cL(X)), \quad Y\in\dom\cL_*, \quad X\in\dom\cL.
    \end{equation*}
    A QMS is \emph{Gaussian} if its predual preserves Gaussian states. This property is equivalent to the Weyl-operator description in Theorem~\ref{theo-explicit-action-on-weyl-operators} \cite{Poletti2022Characterization}; the action formula is proved in \cite[Theorem~2.4]{agredo2022decoherence}.

    Such a semigroup admits a generalized Gorini--Kossakowski--Sudarshan--Lindblad (GKSL) representation with a quadratic Hamiltonian and linear noise operators. We use the following parametrization, where the bar denotes entrywise complex conjugation, $ T $ transposition, and $ * $ the adjoint. The coefficients satisfy
    \begin{equation*}
        \Omega=\Omega^*\in M_d(\C), \qquad
        \kappa=\kappa^T\in M_d(\C), \qquad
        \zeta\in\C^d, \qquad U,V\in M_{m\times d}(\C),
    \end{equation*}
    where $ m\ge1 $ is finite. Write $ \Omega=(\Omega_{jk}) $, $ \kappa=(\kappa_{jk}) $, $ \zeta=(\zeta_j) $, $ U=(u_{\ell j}) $, and $ V=(v_{\ell j}) $. Following \cite[Section~2, Eqs.~(5)--(6)]{fagnola2025spectral}, define the Hamiltonian and noise operators on the finite-particle domain $ \cD $ of Section~\ref{subsection-fock-weyl-gaussian} by
    \begin{align*}
        H &= \sum_{j,k=1}^{d}\left(
            \Omega_{jk}a_j^\dagger a_k
            +\frac{\kappa_{jk}}{2}a_j^\dagger a_k^\dagger
            +\frac{\overline{\kappa_{jk}}}{2}a_ja_k\right)
            +\frac{1}{2}\sum_{j=1}^{d}
            \left(\zeta_j a_j^\dagger+\overline{\zeta_j}a_j\right), \\
        L_\ell &= \sum_{j=1}^{d}
            \left(\overline{v_{\ell j}}a_j+u_{\ell j}a_j^\dagger\right),
            \qquad 1\le\ell\le m.
    \end{align*}
    We use the same notation for their closures; $ H $ is self-adjoint. The matrices $ \Omega $ and $ \kappa $ specify the quadratic part of the Hamiltonian, $ \zeta $ its linear part, and $ U,V $ the noise operators. With these operators, the Heisenberg generator has the formal expression
    \begin{equation*}
        \cL(X) = \mi[H,X]
        +\sum_{\ell=1}^{m}\left(L_\ell^*XL_\ell
        -\frac{1}{2}L_\ell^*L_\ell X
        -\frac{1}{2}XL_\ell^*L_\ell\right).
    \end{equation*}
    Since $ H $ and $ L_\ell $ are generally unbounded, the right-hand side is interpreted as a sesquilinear form on $ \cD\times\cD $ for $ X\in\cB(\mathsf{h}) $; it represents $ \cL(X) $ when $ X\in\dom\cL $. For the construction of the associated QMS, see \cite[Section~2, Theorem~2]{fagnola2025spectral} and \cite[Appendix~A]{agredo2022decoherence}.

    The predual generator has the corresponding expression
    \begin{equation*}
        \cL_*(Y) = -\mi[H,Y]
        +\sum_{\ell=1}^{m}\left(L_\ell YL_\ell^*
        -\frac{1}{2}L_\ell^*L_\ell Y
        -\frac{1}{2}YL_\ell^*L_\ell\right)
    \end{equation*}
    for every finite linear combination of $ \ketbra{\xi}{\eta} $ with $ \xi,\eta\in\cD $. These operators belong to $ \dom\cL_* $, and the equality holds in trace class; see \cite[Theorem~2]{fagnola2025spectral}.

    The same coefficients determine the drift and diffusion in the Weyl evolution. The following formula connects the GKSL representation with the phase-space description used throughout the paper \cite[Theorem~3, Eqs.~(9)--(11)]{fagnola2025spectral}.

    \begin{theorem} \label{theo-explicit-action-on-weyl-operators}
        With the parameters above, define the real-linear operators $ Z,C:\C^d\to\C^d $ by
        \begin{align*}
            Zz &= \left[\frac{1}{2}\left(U^T\overline{U}-V^T\overline{V}\right)+\mi\Omega\right]z
                +\left[\frac{1}{2}\left(U^TV-V^TU\right)+\mi\kappa\right]\overline{z}, \\
            Cz &= \left(U^T\overline{U}+V^T\overline{V}\right)z
                +\left(U^TV+V^TU\right)\overline{z}.
        \end{align*}
        Then
        \begin{equation} \label{eq-gaussian-weyl-action}
            \cT_t(W(z)) = \phi_t(z)W(\me^{tZ}z), \quad t \ge 0, \quad z \in \C^d,
        \end{equation}
        where
        \begin{equation} \label{eq-definition-phi-t}
            \phi_t(z) := \exp{-\frac{1}{2}\int_0^t \Re\langle\me^{sZ}z,C\me^{sZ}z\rangle\,\md s
            +\mi\int_0^t\Re\langle\zeta,\me^{sZ}z\rangle\,\md s}.
        \end{equation}
        The diffusion operator $ C $ is self-adjoint and positive semidefinite. The real matrices $ \mathbf{Z} $ and $ \mathbf{C} $ satisfy
        \begin{equation} \label{eq-quantum-diffusion-matrix}
            \mathbf{C}_{\mathbf{Z}} := \mathbf{C}+\mi(\mathbf{Z}^T\mathbf{J}+\mathbf{J}\mathbf{Z}) \ge 0.
        \end{equation}
    \end{theorem}

    We call $ Z $, $ C $, and $ \zeta $ the drift, diffusion, and displacement parameters, respectively. Their real representations are $ \mathbf{Z} $, $ \mathbf{C} $, and $ \boldsymbol{\zeta}:=\iota(\zeta) $; see \eqref{eq-definition-iota} and \eqref{eq-identification-real-linear-operators}. The matrix $ \mathbf{C}_{\mathbf{Z}} $ in \eqref{eq-quantum-diffusion-matrix} is the \emph{quantum diffusion matrix}. Its positivity follows from complete positivity of the semigroup \cite[Proposition~6]{fagnola2025spectral}. Since Weyl operators have $ \sigma $-weakly dense linear span, \eqref{eq-gaussian-weyl-action} determines the semigroup uniquely.

    In real coordinates, define the accumulated diffusion and displacement by
    \begin{align}
        \mathbf{S}_t &:= \int_0^t \me^{s\mathbf{Z}^T}\mathbf{C}\me^{s\mathbf{Z}}\,\md s, \label{eq-definition-boldface-S-t} \\
        \mathbf{b}_t &:= \int_0^t \me^{s\mathbf{Z}^T}\boldsymbol{\zeta}\,\md s. \label{eq-definition-b-t}
    \end{align}
    Then $ \phi_t(z) = \exp{-\mathbf{z}^T\mathbf{S}_t\mathbf{z}/2+\mi\mathbf{b}_t^T\mathbf{z}} $. The matrix $ \mathbf{S}_t $ is the diffusion contribution to the covariance of an evolved Gaussian state; the full covariance also includes the transported initial covariance. Finally, the trace duality \eqref{eq-predual-pairing} gives the evolution of characteristic functions:
    \begin{equation} \label{eq-evolution-predual-semigroup-characteristic-function}
        \chi_{\cT_{*t}(Y)}(z) = \phi_t(z)\chi_Y(\me^{tZ}z),
        \qquad Y\in\cB_1(\mathsf{h}).
    \end{equation}
    This identity, together with injectivity of the Fourier--Weyl transform, is the main link between the semigroup and the phase-space calculations below.

    \subsection{Irreducibility and invariant states} \label{subsection-irreducibility-invariant}

    An orthogonal projection $ p \in \cB(\mathsf{h}) $ is \emph{subharmonic} if $ \cT_t(p) \ge p $ for every $ t \ge 0 $. The QMS is \emph{irreducible} when its only subharmonic projections are $ 0 $ and $ \1 $. The following characterization combines \cite[Theorem 12]{FagnolaGirotti2026Irreducibility} with the standard finite-dimensional Gramian criterion; compare \cite[Appendix A, Lemma A.1]{FarkasLorenzi2009}.

    \begin{proposition} \label{proposition-irreducibility}
        The following conditions are equivalent:
        \begin{enumerate}
            \item The Gaussian QMS $ (\cT_t)_{t\ge0} $ is irreducible.
            \item The only complex $ \mathbf{Z} $-invariant subspace of $ \C^{2d} $ contained in $ \ker\mathbf{C}_{\mathbf{Z}} $ is $ \{0\} $.
            \item For some, equivalently every, $ t > 0 $,
            \begin{equation*}
                \int_0^t \me^{s\mathbf{Z}^T}\mathbf{C}_{\mathbf{Z}}\me^{s\mathbf{Z}}\,\md s > 0.
            \end{equation*}
            \item $ \displaystyle\bigcap_{j=0}^{2d-1}\ker\bigl(\mathbf{C}_{\mathbf{Z}}\mathbf{Z}^j\bigr) = \{0\} $.
        \end{enumerate}
        In the last three conditions, the real matrix $ \mathbf{Z} $ acts by complexification on $ \C^{2d} $.
    \end{proposition}

    \emph{Except in Section \ref{section-reducible-semigroups} and Appendix \ref{appendix-reducible-semigroups}, the given semigroup $ (\cT_t)_{t\ge0} $ is assumed to be irreducible.} This implies
    \begin{equation} \label{eq-positive-diffusion-gramian}
        \mathbf{S}_t > 0, \qquad t > 0.
    \end{equation}
    Indeed, on a real vector $ \mathbf{z} $, the quadratic form of the integral in item 3 of Proposition \ref{proposition-irreducibility} equals $ \mathbf{z}^T\mathbf{S}_t\mathbf{z} $. Replacing $ \mathbf{C}_{\mathbf{Z}} $ by $ \mathbf{C} $ in the finite-dimensional criteria gives the usual diffusion nondegeneracy condition, which is weaker than irreducibility. We therefore keep these notions distinct. The difference between irreducibility and regularisation for Gaussian QMSs is discussed in \cite{FagnolaGirotti2026Irreducibility}.

    For stable drift, the invariant state supplies the normalization used in Section \ref{eq-stable-drift}. The next proposition records its parameters and the role of faithfulness; see \cite[Theorem 7]{fagnola2025spectral} and \cite[Theorem 13, Corollary 14, and Proposition 15]{girotti2025gaussian}.

    \begin{proposition} \label{proposition-stable-invariant-state}
        Suppose that $ s(\mathbf{Z}) < 0 $. Then the Gaussian QMS has a unique normal invariant state $ \rho $, which is Gaussian, with covariance and mean
        \begin{equation*}
            \mathbf{S} = \int_0^\infty \me^{s\mathbf{Z}^T}\mathbf{C}\me^{s\mathbf{Z}}\,\md s,
            \qquad \omega = -\frac{1}{2}J(Z^\sharp)^{-1}\zeta.
        \end{equation*}
        The following are equivalent: the semigroup is irreducible; $ \rho $ is faithful; and $ \mathbf{S}+\mi\mathbf{J} > 0 $. In particular, for stable drift, irreducibility is equivalent to the existence of a unique faithful normal invariant state.
    \end{proposition}

    The distinction between $ \mathbf{C} $ and $ \mathbf{C}_{\mathbf{Z}} $ is visible in the identity
    \begin{equation*}
        \int_0^\infty \me^{s\mathbf{Z}^T}\mathbf{C}_{\mathbf{Z}}\me^{s\mathbf{Z}}\,\md s
        = \mathbf{S}-\mi\mathbf{J},
    \end{equation*}
    obtained from \eqref{eq-quantum-diffusion-matrix} and stability. Strict positivity on complex phase space gives precisely the faithfulness condition. No invariant state is assumed in the strictly unstable or critical regimes.

    \subsection{Spectral and functional notation} \label{section-spectral-notation}

    For a closed densely defined operator $ A $, we write $ \sigma(A) $ for its spectrum and $ \sigma_{\mathrm{p}}(A) $ for its point spectrum. The \emph{residual spectrum} $ \sigma_{\mathrm{r}}(A) $ consists of those $ \lambda $ for which $ \lambda-A $ is injective and has nondense range. The \emph{approximate point spectrum} $ \sigma_{\mathrm{ap}}(A) $ consists of those $ \lambda $ for which there are $ x_n\in\dom A $ with $ \norm{x_n}=1 $ and $ \norm{(\lambda-A)x_n}\to0 $. These definitions imply $ \sigma_{\mathrm{p}}(A)\subset\sigma_{\mathrm{ap}}(A) $; the approximate point and residual spectra may overlap. We use the standard facts that $ \partial\sigma(A)\subset\sigma_{\mathrm{ap}}(A) $ and that a contraction semigroup generator has spectrum in the closed left half-plane \cite{EngelNagel2000}.

    For the ordinary Euclidean Fourier transform used in the appendices, we adopt the convention
    \begin{equation*}
        \widehat{f}(\xi) := \int_{\R^k}f(x)\me^{-\mi\xi^Tx}\,\md x,
        \qquad f\in L^1(\R^k).
    \end{equation*}
    For an open set $ \cO\subset\R^k $, we denote by $ \cC(\cO) $ and $ \cC^\infty(\cO) $ the spaces of continuous and smooth complex-valued functions, respectively. The subspaces of functions whose support is compact in $ \cO $ are denoted by $ \cC_c(\cO) $ and $ \cC_c^\infty(\cO) $. For a compact space $ K $, $ \cC(K) $ is equipped with the uniform norm. We write $ \cS(\R^k) $ for the Schwartz space of smooth functions whose derivatives of every order decay faster than any inverse power of $ \abs{x} $. The spaces $ L^p(\cO) $, with $ 1\le p\le\infty $, are taken with respect to Lebesgue measure unless another measure is specified; $ \norm{\cdot}_p $ denotes the usual norm, with the essential supremum norm when $ p=\infty $. We write $ \ell^2 $ for the space of square-summable complex sequences. For a Banach space $ E $ and an open interval $ I $, $ C^\infty(I;E) $ denotes the space of maps with norm-continuous derivatives of every order.

    \section{Spectrum for stable drift} \label{eq-stable-drift}

    Assume that $ s(\mathbf{Z}) < 0 $. By Proposition \ref{proposition-stable-invariant-state}, there is a unique faithful normal invariant state $ \rho $. We denote its mean by $ \omega $ and its covariance operator by $ S $.

    \begin{theorem} \label{theorem-strict-negative-implies-point-spectrum}
        If $ s(\mathbf{Z}) < 0 $, then every $ \eta \in \C $ with $ \Re\eta < 0 $ belongs to $ \sigma_{\mathrm{p}}(\cL_*) $.
    \end{theorem}

    The remaining spectral assertions follow from contractivity, closedness of the spectrum, and stability of the drift. In particular, stability excludes every nonzero imaginary eigenvalue.

    \begin{theorem} \label{theorem-stable-spectrum}
        If $ s(\mathbf{Z}) < 0 $, then
        \begin{equation} \label{eq-full-spectrum-cL-star}
            \sigma (\cL_*) = \sigma_{\mathrm{ap}} (\cL_*)  = \{ \eta: \Re \eta \le 0 \}.
        \end{equation}
        In particular,
        \begin{equation} \label{eq-point-spectrum-cL-star}
            \sigma_{\mathrm{p}} (\cL_*) = \{ \eta : \Re \eta < 0 \} \cup \{ 0 \}.
        \end{equation}

    \end{theorem}

    No diagonalizability assumption on $ \mathbf{Z} $ is needed. The proof uses the characteristic function of the invariant Gaussian state to remove the Gaussian multiplier from the evolution. The eigenoperator equation then becomes a scaling relation along the drift. The essential analytical step is to realize the resulting functions as characteristic functions of trace-class operators and verify membership in the generator domain. The construction and proofs are given in Appendix~\ref{appendix-stable-proofs}.

    Section~\ref{subsection-application-thermal} illustrates this result for the thermal oscillator and gives an explicit trace-class eigenoperator whose eigenvalue lies outside the family \eqref{eq-drift-eigenvalue-combinations}.

    \section{Spectrum for strictly unstable drift} \label{section-strictly-unstable-drift}

    Assume $ s(\mathbf{Z})>0 $. The drift may also have eigenvalues with zero or negative real part. An eigenvalue with positive real part determines a real invariant subspace of dimension one or two. Integrating Weyl operators over this subspace yields the following bounded eigenoperators of the dual semigroup; the construction and its integrability estimates are given in Appendix \ref{appendix-unstable-proofs}.

    \begin{proposition} \label{proposition-eigenvector-Gaussian-generator}
        If $ s(\mathbf{Z})>0 $, then, for every $ \eta\in\C $ with $ \Re\eta<0 $, there is a nonzero operator $ X_\eta\in\cB(\mathsf{h}) $ such that
        \begin{equation} \label{eq-eigenvector-Gaussian-generator}
            \cT_t (X_\eta) = \me^{\eta t} X_\eta, \quad \forall t \ge 0.
        \end{equation}
    \end{proposition}

    The dual eigenoperators determine spectral values of $ \cL_* $ through the trace pairing. More precisely, $ Y\mapsto\Tr(YX_\eta) $ is a nonzero continuous functional that annihilates $ \ran(\eta-\cL_*) $. Thus every $ \eta $ with negative real part lies in the spectrum, and contractivity determines the opposite inclusion.

    \begin{theorem} \label{theorem-whole-spectrum-strictly-unstable-drift}
        If $ s(\mathbf{Z})>0 $, then
        \begin{equation} \label{eq-spectrum-predual-generator-expanding-case}
            \sigma (\cL_*) = \{ \eta \in \C : \Re \eta \le 0 \}.
        \end{equation}
    \end{theorem}

    Although the full spectrum agrees with that of the stable case, the point spectrum changes completely. Along a dense set of phase-space trajectories, the accumulated diffusion grows faster than any linear function of time. The characteristic-function identity then rules out every trace-class eigenvector.

    \begin{theorem} \label{theorem-unstable-spectral-types}
        If $ s(\mathbf{Z})>0 $, then $ \mi\R\subset\sigma_{\mathrm{ap}}(\cL_*) $,
        \begin{equation} \label{eq-expanding-case-point-spectrum}
            \sigma_{\mathrm{p}} (\cL_*) = \varnothing,
        \end{equation}
        and
        \begin{equation} \label{eq-expanding-case-residual-spectrum}
            \{ \eta : \Re \eta < 0 \} \cup \{ 0 \} \subset \sigma_{\mathrm{r}} (\cL_*).
        \end{equation}
    \end{theorem}

    Appendix \ref{appendix-unstable-proofs} contains the proofs. The assertion at zero uses trace preservation: $ \ran\cL_* $ is contained in the closed subspace of trace-zero operators. Table \ref{tab-stable-unstable-spectrum} records exactly these conclusions; it does not assign a further residual or continuous classification to the nonzero imaginary axis.

    \section{Spectrum for periodic critical drift} \label{section-critical-drift}

    Assume $ s(\mathbf{Z})=0 $. We call the drift periodic when $ \me^{r\mathbf{Z}}=\1 $ for some $ r>0 $. This is a condition on the deterministic linear flow, not on the full QMS, which still contains diffusion and displacement. A periodic matrix has no eigenvalues with negative real part and no nontrivial Jordan blocks. The zero matrix requires separate treatment because it has no least positive period.

    At a return time, the Weyl evolution is multiplication by $ \phi_r $. Appendix \ref{appendix-periodic-proofs} proves that $ \sigma(\cT_{*r})=\overline{\phi_r(\C^d)} $ by representing $ \cT_{*r} $ as a Gaussian average of Weyl conjugations. This common step applies both to zero drift and to a nonzero periodic drift.

    \subsection{Zero drift}

    Suppose $ \mathbf{Z}=0 $. Irreducibility implies
    \begin{equation*}
        \mathbf{C}>0,
    \end{equation*}
    since $ \mathbf{S}_t=t\mathbf{C} $. Define
    \begin{equation} \label{eq-definition-psi-z}
        \psi(z):=-\frac{1}{2}\mathbf{z}^T\mathbf{C}\mathbf{z}+\mi\boldsymbol{\zeta}^T\mathbf{z}.
    \end{equation}
    Then
    \begin{equation} \label{eq-evolution-Wz-zero-drift-case}
        \cT_t(W(z))=\me^{t\psi(z)}W(z).
    \end{equation}
    The scalar exponent determines the entire generator spectrum.

    \begin{theorem} \label{theorem-zero-drift-case-spectrum}
        If $ \mathbf{Z}=0 $, then
        \begin{equation*}
            \sigma (\cL_*) = \psi (\C^d).
        \end{equation*}
    \end{theorem}

    The geometry of the range of $ \psi $ can be made explicit. The displacement produces a parabolic region; in its absence the range is the negative real half-line.

    \begin{proposition}[Shape of the spectrum] \label{proposition-zero-drift-shape}
        Assume $ \mathbf{Z}=0 $. If $ \zeta=0 $, then
        \begin{equation*}
            \sigma(\cL_*) = \psi(\C^d) = (-\infty, 0].
        \end{equation*}
        When $ \zeta \neq 0 $, we have 
        \begin{equation*}
            \sigma(\cL_*) = \psi(\C^d) = \left\{ \lambda \in \C : \Re \lambda \le - \frac{(\Im \lambda)^2}{2 \boldsymbol{\zeta}^T \mathbf{C}^{-1} \boldsymbol{\zeta} } \right\}.
        \end{equation*}
    \end{proposition}

    In this case each spectral value is detected by a Weyl operator in the Banach dual. Continuity of trace-class characteristic functions, on the other hand, excludes eigenvectors of $ \cL_* $.

    \begin{theorem} \label{theorem-zero-drift-spectral-types}
        If $ \mathbf{Z}=0 $, then
        \begin{equation*}
            \sigma_{\mathrm{p}} (\cL_*) = \varnothing, \quad \sigma_{\mathrm{r}} (\cL_*) = \sigma (\cL_*) = \psi(\C^d).
        \end{equation*}
    \end{theorem}

    Proofs of all three statements are given in Appendix \ref{appendix-periodic-proofs}.

    \subsection{Nonzero periodic drift}

    Suppose $ \mathbf{Z}\neq0 $ and the flow is periodic. Its return times form a proper nonzero closed subgroup of $ \R $, so the least positive period exists. We denote it by
    \begin{equation*}
        \tau:=\min\{t>0:\me^{t\mathbf{Z}}=\1\}.
    \end{equation*}
    The generator spectrum is obtained by allowing every logarithmic branch of the multiplier at time $ \tau $.

    \begin{theorem} \label{theorem-spectrum-predual-generator-periodic-critical-drift}
        Assume that $ \mathbf{Z}\neq0 $ and $ \me^{\tau\mathbf{Z}}=\1 $, where $ \tau $ is the least positive period. Let $ \mathbf{S}_\tau $ be the accumulated diffusion defined in \eqref{eq-definition-boldface-S-t}, and recall from \eqref{eq-definition-b-t} that
        \begin{equation*}
            \mathbf{b}_\tau=\int_0^\tau\me^{s\mathbf{Z}^T}\boldsymbol{\zeta}\,\md s
        \end{equation*}
        is the accumulated displacement over one period. Then
        \begin{equation*}
            \sigma(\cL_*) = \bigcup_{n \in \Z} \left\{ -\frac{1}{2\tau} \mathbf{z}^T \mathbf{S}_\tau \mathbf{z} + \frac{\mi}{\tau} ( \mathbf{b}_\tau^T \mathbf{z} + 2 \pi n ) : \mathbf{z} \in \R^{2d} \right\}.
        \end{equation*}
        Moreover, $ \sigma_{\mathrm{p}}(\cL_*)=\varnothing $.
    \end{theorem}

    Spectral inclusion gives one direction. For the other, averaging $ \cT_s(W(z)) $ against a suitable exponential over $ [0,\tau] $ produces bounded dual eigenoperators. They are nonzero whenever the orbit of $ z $ has least period $ \tau $. Such points form a dense subset of phase space, and closedness of the spectrum gives the full formula. The details are in Appendix \ref{appendix-periodic-proofs}.

    Table \ref{tab-periodic-spectrum} expresses these results geometrically. For $ \mathbf{b}_\tau\neq0 $, its final formula follows from the same quadratic minimization as Proposition \ref{proposition-zero-drift-shape}, with $ \mathbf{C} $ replaced by $ \mathbf{S}_\tau/\tau $ and $ \boldsymbol{\zeta} $ by $ \mathbf{b}_\tau/\tau $, followed by translation by $ 2\pi\mi n/\tau $. If $ \mathbf{b}_\tau=0 $, the positive quadratic form ranges over $ [0,\infty) $. In particular, $ \mathbf{b}_\tau=0 $ whenever $ \mathbf{Z} $ is invertible, since $ \int_0^\tau\me^{s\mathbf{Z}^T}\,\md s=(\mathbf{Z}^T)^{-1}(\me^{\tau\mathbf{Z}^T}-\1)=0 $.

    \begin{table}[htbp]
        \centering
        \small
        \renewcommand{\arraystretch}{1.8}
        \begin{tabular}{@{}p{3.7cm}p{9.7cm}@{}}
            \toprule
            Drift and displacement & Spectrum of $ \cL_* $ \\
            \midrule
            $ \mathbf{Z}=0 $, $ \zeta=0 $ & $ (-\infty,0] $ \\
            $ \mathbf{Z}=0 $, $ \zeta\neq0 $ & $ \displaystyle\left\{\lambda:\Re\lambda\le-\frac{(\Im\lambda)^2}{2\boldsymbol{\zeta}^T\mathbf{C}^{-1}\boldsymbol{\zeta}}\right\} $ \\
            $ \mathbf{Z}\neq0 $, $ \mathbf{b}_\tau=0 $ & $ \displaystyle\bigcup_{n\in\Z}\left(( -\infty,0]+\frac{2\pi\mi n}{\tau}\right) $ \\
            $ \mathbf{Z}\neq0 $, $ \mathbf{b}_\tau\neq0 $ & $ \displaystyle\bigcup_{n\in\Z}\left\{\lambda:\Re\lambda\le-\frac{(\tau\Im\lambda-2\pi n)^2}{2\tau\mathbf{b}_\tau^T\mathbf{S}_\tau^{-1}\mathbf{b}_\tau}\right\} $ \\
            \bottomrule
        \end{tabular}
        \caption{The spectrum for periodic critical drift under irreducibility. In the last two rows, $ \tau $ is the least positive period. The point spectrum is empty in all four cases. The first two spectra are entirely residual by Theorem \ref{theorem-zero-drift-spectral-types}.}
        \label{tab-periodic-spectrum}
    \end{table}

    \section{Spectrum for nonperiodic critical drift} \label{section-critical-drift-aperiodic}

    Assume that the drift is critical and nonperiodic, that is,
    \begin{equation} \label{eq-open-critical-case}
        s(\mathbf{Z})=0, \qquad \me^{t\mathbf{Z}}\neq\1 \quad\text{for every } t>0.
    \end{equation}
    We distinguish mixed stable--critical drift from purely critical drift. In the latter case, the proof depends on whether the drift matrix is semisimple.

    \subsection{Mixed stable--critical drift} \label{subsection-mixed-stable-critical}

    \begin{theorem} \label{theorem-spectrum-mixed-stable-critical-drift}
        If
        \begin{equation*}
            s(\mathbf{Z})=0, \qquad \sigma(\mathbf{Z})\cap\{\eta\in\C:\Re\eta<0\}\neq\varnothing,
        \end{equation*}
        then
        \begin{equation} \label{eq-spectrum-predual-generator-mixed-stable-critical-drift}
            \sigma_{\mathrm{ap}}(\cL_*)=\sigma(\cL_*)=\{\eta\in\C:\Re\eta\le0\}.
        \end{equation}
    \end{theorem}

    The proof is given in Appendix \ref{appendix-mixed-stable-critical}.

    \subsection{Purely critical semisimple drift} \label{subsection-purely-critical-semisimple}

    Recall that a real matrix is semisimple if it is diagonalizable over $ \C $.

    \begin{theorem} \label{theorem-spectrum-purely-critical-semisimple-drift}
        Suppose that $ \sigma(\mathbf{Z})\subset\mi\R $, that $ \mathbf{Z} $ is semisimple, and that $ \me^{t\mathbf{Z}}\neq\1 $ for every $ t>0 $. Then
        \begin{equation} \label{eq-spectrum-purely-critical-semisimple-drift}
            \sigma_{\mathrm{ap}}(\cL_*)=\sigma(\cL_*)=\{\eta\in\C:\Re\eta\le0\},
        \end{equation}
        whereas
        \begin{equation*}
            \sigma_{\mathrm{p}}(\cL_*)=\varnothing.
        \end{equation*}
    \end{theorem}

    The proof is given in Appendix \ref{appendix-purely-critical-semisimple}.

    \subsection{Purely critical nonsemisimple drift} \label{subsection-purely-critical-nonsemisimple}

    A nontrivial Jordan block on the imaginary axis produces polynomially growing orbits. In particular, a purely critical nonsemisimple drift cannot be periodic.

    \begin{theorem} \label{theorem-spectrum-purely-critical-nonsemisimple-drift}
        Suppose that $ \sigma(\mathbf{Z})\subset\mi\R $ and that $ \mathbf{Z} $ is not semisimple. Then
        \begin{equation} \label{eq-spectrum-purely-critical-nonsemisimple-drift}
            \sigma_{\mathrm{ap}}(\cL_*)=\sigma(\cL_*)=\{\eta\in\C:\Re\eta\le0\},
            \qquad \sigma_{\mathrm{p}}(\cL_*)=\varnothing.
        \end{equation}
    \end{theorem}

    The proof is given in Appendix~\ref{appendix-purely-critical-nonsemisimple}. This theorem completes the nonperiodic classification: every irreducible Gaussian QMS with nonperiodic drift has the closed left half-plane as its trace-class spectrum. Together with the periodic cases in Section~\ref{section-critical-drift}, this establishes Theorem~\ref{theorem-global-trace-class-spectrum}.

    \section{Why irreducibility matters}
    \label{section-reducible-semigroups}

    The spectral conclusions above need not survive when irreducibility is
    dropped.  In particular, nonperiodic critical drift no longer forces the
    spectrum to be the closed left half-plane.  The following one-mode examples
    keep both the drift and the displacement fixed and change only the diffusion.
    Here $ c>0 $ and $ \1_2 $ denotes the identity matrix on real phase space.

    \begin{proposition}
        \label{proposition-same-drift-counterexamples}
        The Gaussian QMSs in each row of Table
        \ref{tab-same-drift-counterexamples} have the same drift $ \mathbf{Z} $ and zero
        displacement.  The QMS with $ \mathbf{C}=0 $ is reducible, whereas the QMS with
        $ \mathbf{C}=c\1_2 $ is irreducible, and their predual generators have the
        spectra displayed in the table.
    \end{proposition}

    \begin{table}[htbp]
        \centering
        \small
        \renewcommand{\arraystretch}{1.55}
        \begin{tabular}{@{}p{2.7cm}p{3.6cm}p{3.6cm}p{3.6cm}@{}}
            \toprule
            Drift & Regime & Spectrum for $ \mathbf{C}=0 $
            & Spectrum for $ \mathbf{C}=c\1_2 $ \\
            \midrule
            $ \mathbf{Z}=0 $
            & Zero drift
            & $ \{0\} $
            & $ (-\infty,0] $ \\
            $ \displaystyle
            \mathbf{Z}=\begin{pmatrix}0&-1\\0&0\end{pmatrix} $
            & \raggedright Purely critical, nonsemisimple
            & $ \mi\R $
            & $ \displaystyle\{\lambda\in\C:\Re\lambda\le0\} $ \\
            \bottomrule
        \end{tabular}
        \caption{Same-drift examples with different trace-class spectra.  In each
        row the displacement is zero.}
        \label{tab-same-drift-counterexamples}
    \end{table}

    The first row compares the identity semigroup with quantum heat flow. In the
    second row, $ H=p^2/2 $ generates the unitary dynamics with zero diffusion.
    Here $ q=(a+a^\dagger)/\sqrt2 $ and $ p=(a-a^\dagger)/(\mi\sqrt2) $ are the
    position and momentum operators. For the same drift, the choice
    $ \mathbf{C}=c\1_2 $ yields an irreducible QMS.
    Thus diffusion can enlarge the spectrum from a point to a half-line or from
    the imaginary axis to the entire closed left half-plane, without changing
    the drift.  The proof, including the generator-domain argument for the
    free-particle group, is given in Appendix
    \ref{appendix-reducible-semigroups}.

    \section{Applications} \label{section-applications}

    We apply the spectral classification to a thermal oscillator, a free particle with position noise, and a driven boson chain with noise at its endpoints. Their drifts are, respectively, stable, purely critical and nonsemisimple, and nonzero periodic.

    \subsection{Thermal oscillator with stable drift}
    \label{subsection-application-thermal}

    Let $ \Gamma_->\Gamma_+>0 $, $ \omega\in\R $, and $ N=a^\dagger a $. Consider the predual Gaussian GKSL generator
    \begin{align} \label{eq-application-thermal-generator}
        \cL_*(Y)&=-\mi\omega[N,Y]
        +\Gamma_-\left(aYa^\dagger-\frac12\{N,Y\}\right)\notag\\
        &\quad+\Gamma_+\left(a^\dagger Ya-\frac12\{N+\1,Y\}\right).
    \end{align}
    This is the standard quantum optical master equation for a damped oscillator in a thermal field; see \cite[Eq.~(10)]{Vacchini2002Covariance} and the references therein. It also arises in the weak coupling description of an oscillator interacting with a heat bath \cite{Davies1973HeatBath}. Set $ \beta_0=\Gamma_--\Gamma_+>0 $ and $ r=\Gamma_+/\Gamma_-\in(0,1) $. The drift and diffusion are
    \begin{equation} \label{eq-application-thermal-drift}
        Zz=(-\beta_0/2+\mi\omega)z,\qquad
        \mathbf{Z}=\begin{pmatrix}-\beta_0/2&-\omega\\\omega&-\beta_0/2\end{pmatrix},
        \qquad \mathbf{C}=(\Gamma_-+\Gamma_+)\1_2.
    \end{equation}
    Thus the drift eigenvalues are $ -\beta_0/2\pm\mi\omega $, so the drift is stable. Moreover,
    \begin{equation*}
        \mathbf{C}_{\mathbf{Z}}=(\Gamma_-+\Gamma_+)\1_2-\mi\beta_0\mathbf{J}>0
    \end{equation*}
    has eigenvalues $ 2\Gamma_- $ and $ 2\Gamma_+ $. The semigroup is therefore irreducible. Its unique normal invariant state is faithful and equals
    \begin{equation*}
        \rho=(1-r)r^N,\qquad
        \chi_\rho(z)=\exp{-\frac{1+r}{2(1-r)}\abs{z}^2}.
    \end{equation*}
    Theorem~\ref{theorem-stable-spectrum} yields
    \begin{equation} \label{eq-application-thermal-spectrum}
        \sigma(\cL_*)=\sigma_{\mathrm{ap}}(\cL_*)=\{\lambda\in\C:\Re\lambda\le0\},
        \qquad \sigma_{\mathrm p}(\cL_*)=\{\lambda\in\C:\Re\lambda<0\}\cup\{0\}.
    \end{equation}
    By comparison, the polynomial eigenvalues in \eqref{eq-drift-eigenvalue-combinations} form only the family
    \begin{equation*}
        \left\{-\frac{\beta_0}{2}(j+k)+\mi\omega(j-k):j,k\in\N\right\}.
    \end{equation*}
    The following example gives an explicit eigenoperator corresponding to an additional eigenvalue.

    \begin{example} \label{example-thermal-oscillator}
        Define $ c_n $ by the Taylor expansion
        \begin{equation} \label{eq-thermal-example-coefficients}
            \frac{(1-w)^{1/3}}{(1-rw)^{4/3}}
            =\sum_{n=0}^\infty c_n w^n, \qquad \abs{w}<1,
        \end{equation}
        with both powers equal to one at $ w=0 $. Then
        \begin{equation*}
            Y:=\sum_{n=0}^\infty c_n\ketbra{e(n)}{e(n)}
            \in\dom\cL_*, \qquad \cL_*(Y)=-\frac{\beta_0}{3}Y.
        \end{equation*}
        Thus $ -\beta_0/3 $ is a trace-class eigenvalue outside the polynomial family above. The generator-domain assertion is proved in Appendix~\ref{subsection-proof-thermal-oscillator}.
    \end{example}

    \subsection{Free particle with position noise}
    \label{subsection-application-free-particle}

    Let $ q=(a+a^\dagger)/\sqrt2 $ and $ p=(a-a^\dagger)/(\mi\sqrt2) $, so that $ [q,p]=\mi\1 $. For a particle of mass $ m=2 $, consider
    \begin{equation} \label{eq-application-free-generator}
        \cL_*(Y)=-\frac{\mi}{4}[p^2,Y]-\sigma^2[q,[q,Y]],\qquad \sigma>0.
    \end{equation}
    This model occurs in the literature on quantum Brownian motion and position localization; see \cite[Eq.~(15)]{BassiIppolitiVacchini2005} and \cite[Eq.~(3)]{Branford2019Diffusion}. It does not, however, belong to the class of Gaussian QMSs of weak coupling limit type considered in \cite{BolanosFagnolaQuezada2026WeakCoupling}: the Hamiltonian $ H=p^2/4 $ has purely continuous spectrum, whereas the Hamiltonians of the generators in that class admit a complete orthonormal family of eigenvectors; see \cite[Section~5 and Appendix~A]{BolanosFagnolaQuezada2026WeakCoupling}. Here
    \begin{equation*}
        H=\frac{p^2}{4}=\frac14N-\frac18(a^{\dagger2}+a^2)+\frac18\1,
        \qquad L=\sqrt2\sigma q=\sigma(a+a^\dagger).
    \end{equation*}
    Apart from the irrelevant scalar term in $ H $, the GKSL parameters are $ \Omega=1/4 $, $ \kappa=-1/4 $, $ U=V=\sigma $, and $ \zeta=0 $. The Weyl drift and diffusion are
    \begin{equation} \label{eq-application-free-drift}
        Z(x+\mi y)=-\frac y2,\qquad
        \mathbf{Z}=\begin{pmatrix}0&-1/2\\0&0\end{pmatrix},
        \qquad \mathbf{C}=4\sigma^2\begin{pmatrix}1&0\\0&0\end{pmatrix}.
    \end{equation}
    In particular,
    \begin{equation*}
        \sigma(\mathbf{Z})=\{0\},\qquad
        \mathbf{Z}^2=0,\qquad \mathbf{Z}\ne0,\qquad
        \me^{t\mathbf{Z}}=\begin{pmatrix}1&-t/2\\0&1\end{pmatrix}.
    \end{equation*}
    The drift is purely critical, nonsemisimple and nonperiodic. Since it is Hamiltonian, $ \mathbf{Z}^T\mathbf{J}+\mathbf{J}\mathbf{Z}=0 $ and $ \mathbf{C}_{\mathbf{Z}}=\mathbf{C} $. For every $ t>0 $,
    \begin{equation} \label{eq-application-free-gramian}
        \int_0^t\me^{s\mathbf{Z}^T}\mathbf{C}_{\mathbf{Z}}\me^{s\mathbf{Z}}\,\md s
        =\mathbf{S}_t
        =4\sigma^2\begin{pmatrix}t&-t^2/4\\-t^2/4&t^3/12\end{pmatrix}>0.
    \end{equation}
    Proposition~\ref{proposition-irreducibility} therefore implies irreducibility.

    \begin{proposition} \label{proposition-application-free-spectrum}
        The predual generator in \eqref{eq-application-free-generator} satisfies
        \begin{equation*}
            \sigma(\cL_*)=\sigma_{\mathrm{ap}}(\cL_*)=\{\lambda\in\C:\Re\lambda\le0\},
            \qquad \sigma_{\mathrm p}(\cL_*)=\varnothing.
        \end{equation*}
    \end{proposition}

    \begin{proof}
        The conclusion follows from Theorem~\ref{theorem-spectrum-purely-critical-nonsemisimple-drift}, since the semigroup is irreducible and its drift is purely critical and nonsemisimple.
    \end{proof}

    Thus there is no nonzero trace-class operator $ Y $ satisfying $ \cT_{*t}(Y)=\me^{\lambda t}Y $ for every $ t\ge0 $. In particular, there is no normal invariant state.

\subsection{A driven boson chain with noise at the endpoints}
\label{subsection-application-boson-chain}

Consider a chain of $d\in\{2,3\}$ boson modes with nearest-neighbour
coherent coupling, balanced loss and incoherent pumping at the two
endpoints, and a coherent drive on the first mode. This is a small modification
of the harmonic chain with diffusive endpoint reservoirs in
\cite[Section~IV.B]{Nicacio2015}; see also the related cavity arrays in
\cite[Eqs.~(1)--(3)]{WanjuraBrunelliNunnenkamp2020} and the references therein.
We add a resonant coherent drive and work in the frame rotating at the
common onsite frequency. All spectra below refer to the autonomous
generator in this frame.

On $\mathsf h=\Gamma(\C^d)$, let
\begin{align*}
 H&=g\sum_{j=1}^{d-1}
       \left(a_j^\dagger a_{j+1}+a_{j+1}^\dagger a_j\right)
       +\mi\varepsilon(a_1^\dagger-a_1),\\
 L_1&=\sqrt\gamma\,a_1,\qquad L_2=\sqrt\gamma\,a_d,\qquad
 L_3=\sqrt\gamma\,a_1^\dagger,\qquad L_4=\sqrt\gamma\,a_d^\dagger,
\end{align*}
where $g,\gamma>0$ and $\varepsilon\geq0$. We determine the spectrum of the closed
predual generator $\cL_*$ on $\cB_1(\mathsf h)$.

Let $e_1,\ldots,e_d$ be the standard basis of $\R^d$ and set
\begin{equation*}
 \Omega=g\sum_{j=1}^{d-1}
       (e_je_{j+1}^T+e_{j+1}e_j^T),\qquad
 P_{\mathrm{end}}=e_1e_1^T+e_de_d^T.
\end{equation*}
In the notation of the Gaussian generator, $\Omega$ is the quadratic
Hamiltonian matrix, $\kappa=0$,
$\zeta=2\mi\varepsilon e_1$, and
\begin{equation*}
 \mathbf Z=\begin{pmatrix}0&-\Omega\\ \Omega&0\end{pmatrix},\qquad
 \mathbf C=2\gamma\begin{pmatrix}P_{\mathrm{end}}&0\\0&P_{\mathrm{end}}\end{pmatrix},
 \qquad
 \boldsymbol\zeta=2\varepsilon\begin{pmatrix}0\\e_1\end{pmatrix}.
\end{equation*}
For $Jz=-\mi z$, the real matrix is
$\mathbf J=\bigl(\begin{smallmatrix}0&\1_d\\-\1_d&0\end{smallmatrix}\bigr)$.
Since $\mathbf Z^T\mathbf J+\mathbf J\mathbf Z=0$, the quantum
diffusion matrix is $\mathbf C_{\mathbf Z}=\mathbf C$.

The semigroup is irreducible. Indeed, $e_1$ is cyclic for $\Omega$:
the vectors $e_1,\Omega e_1,\ldots,\Omega^{d-1}e_1$ span $\C^d$,
as follows successively from the nonzero nearest-neighbour couplings.
If $(u,v)\in\C^d\oplus\C^d$ satisfies
$\mathbf C\mathbf Z^k(u,v)=0$ for every $k\geq0$, the endpoint
coordinates give $e_1^T\Omega^ku=e_1^T\Omega^kv=0$ for every $k$.
Cyclicity implies $u=v=0$. Proposition~\ref{proposition-irreducibility}(4), together with the
Cayley--Hamilton theorem, proves the claim. For three sites,
$\mathbf C$ is singular: the Hamiltonian coupling transmits the
endpoint noise to the middle mode.

Recall the accumulated diffusion and displacement from
\eqref{eq-definition-boldface-S-t}--\eqref{eq-definition-b-t}:
\begin{equation*}
 \mathbf S_t=\int_0^t\me^{s\mathbf Z^T}\mathbf C\me^{s\mathbf Z}\,\md s,
 \qquad
 \mathbf b_t=\int_0^t\me^{s\mathbf Z^T}\boldsymbol\zeta\,\md s.
\end{equation*}
The drift is purely critical and semisimple in both cases below.
We compute its least positive period $\tau$ and apply Theorem~\ref{theorem-spectrum-predual-generator-periodic-critical-drift}.

\medskip
\noindent\textbf{Two modes.}
Here $P_{\mathrm{end}}=\1_2$ and the eigenvalues of $\Omega$
are $g$ and $-g$. Hence $\tau=2\pi/g$ and
\begin{equation*}
 \mathbf S_\tau=2\gamma\tau\1_4,\qquad
 \int_0^\tau\me^{s\mathbf Z^T}\,\md s=0,\qquad
 \mathbf b_\tau=0.
\end{equation*}
The integral vanishes because $\mathbf Z$ is invertible and
$\me^{\tau\mathbf Z}=\1$. Therefore
\begin{equation}\label{eq-application-two-mode-spectrum}
 \sigma(\cL_*)=
 \bigcup_{n\in\Z}\bigl(( -\infty,0]+\mi ng\bigr),
 \qquad \sigma_{\mathrm p}(\cL_*)=\varnothing.
\end{equation}
The spectrum consists of horizontal half-lines. It is independent
of the drive amplitude: the displacement averages to zero over
one period, even when $\varepsilon>0$.

\medskip
\noindent\textbf{Three modes.}
The eigenvalues of $\Omega$ are $0,\nu,-\nu$, where
$\nu=\sqrt2g$, and $\tau=2\pi/\nu$. Corresponding orthonormal
eigenvectors are
\begin{equation*}
 v_0=\frac{e_1-e_3}{\sqrt2},\qquad
 v_\pm=\frac{e_1\pm\sqrt2e_2+e_3}{2}.
\end{equation*}
Let $P_0=v_0v_0^T$ be the orthogonal projection onto $\ker\Omega$. Since
$v_0^TP_{\mathrm{end}}v_0=1$ and
$v_\pm^TP_{\mathrm{end}}v_\pm=1/2$, averaging over a period gives
\begin{equation*}
 \int_0^\tau\me^{-\mi s\Omega}P_{\mathrm{end}}
             \me^{\mi s\Omega}\,\md s
       =\frac{\tau}{2}(\1_3+P_0)=\frac{\tau}{4}
       \begin{pmatrix}3&0&-1\\0&2&0\\-1&0&3\end{pmatrix}.
\end{equation*}
The terms between distinct eigenspaces integrate to zero.
Consequently,
\begin{align*}
 \mathbf S_\tau
   &=\gamma\tau
       \begin{pmatrix}\1_3+P_0&0\\0&\1_3+P_0\end{pmatrix}, \qquad
 \int_0^\tau\me^{s\mathbf Z^T}\,\md s =\tau\begin{pmatrix}P_0&0\\0&P_0\end{pmatrix}.
\end{align*}
In particular, $\mathbf S_\tau>0$. The accumulated displacement is
\begin{equation*}
 \mathbf b_\tau=
   \varepsilon\tau\begin{pmatrix}0\\e_1-e_3\end{pmatrix}.
\end{equation*}
Since $(\1_3+P_0)(e_1-e_3)=2(e_1-e_3)$, we obtain
\begin{equation*}
 \mathbf b_\tau^T\mathbf S_\tau^{-1}\mathbf b_\tau
       =\frac{\varepsilon^2\tau}{\gamma}.
\end{equation*}
For $\varepsilon>0$, Theorem~\ref{theorem-spectrum-predual-generator-periodic-critical-drift} and the quadratic minimization
in Table~\ref{tab-periodic-spectrum} therefore yield
\begin{equation}\label{eq-application-cavity-spectrum}
 \sigma(\cL_*)=
 \bigcup_{n\in\Z}\left\{\lambda\in\C:
    \Re\lambda\leq-\frac{\gamma}{2\varepsilon^2}
                  (\Im\lambda-n\nu)^2\right\},
 \qquad \sigma_{\mathrm p}(\cL_*)=\varnothing.
\end{equation}
These are filled parabolic regions, with vertices at $\mi n\nu$;
Figure~\ref{fig-boson-chain-spectra} compares the two cases.
When $\varepsilon=0$, the spectrum instead consists of the half-lines
$\bigcup_{n\in\Z}(( -\infty,0]+\mi n\nu)$.

The difference between the two cases is the zero-frequency mode
$c_0=(a_1-a_3)/\sqrt2$ of the three-site chain. The component of the
drive along this mode accumulates over each period, whereas the
components along the other modes average to zero. More explicitly,
for a state with finite mean total occupation, writing
$\rho_t=\cT_{*t}(\rho)$ gives
\begin{equation*}
 \Tr(\rho_t c_0)=\Tr(\rho c_0)+\frac{\varepsilon t}{\sqrt2}.
\end{equation*}
This mode is subject to noise: the restriction of $\mathbf S_\tau$
to its two real coordinates is $2\gamma\tau\1_2$.
Thus periodicity of the drift does not make the QMS periodic.
In both cases the empty trace-class point spectrum excludes normal
invariant states. In particular, the imaginary-axis spectral points
in \eqref{eq-application-two-mode-spectrum} and
\eqref{eq-application-cavity-spectrum} are not trace-class eigenvalues.

\begin{figure}[htbp]
 \centering
 \begin{minipage}[t]{0.48\textwidth}
  \centering
  \textbf{(a)} Two modes\par\smallskip
  \begin{tikzpicture}[x=1.45cm,y=0.64cm]
   \foreach \n in {-2,-1,0,1,2} {
    \draw[<-,line width=0.7pt] (-2.65,\n)--(0,\n);
    \fill (0,\n) circle[radius=1.2pt];
   }
   \draw[->,darkgray] (-2.8,0)--(0.7,0) node[right] {$u$};
   \draw[->,darkgray] (0,-2.65)--(0,2.85) node[above] {$v$};
   \foreach \n in {-2,-1,1,2} {
    \node[right,font=\small] at (0.06,\n) {$\n$};
   }
   \node[below right,font=\small] at (0,0) {$0$};
  \end{tikzpicture}
  \par\smallskip
  $u=\Re\lambda/g$,\quad $v=\Im\lambda/g$
 \end{minipage}\hfill
 \begin{minipage}[t]{0.48\textwidth}
  \centering
  \textbf{(b)} Three modes, $\varepsilon>0$\par\smallskip
  \begin{tikzpicture}[x=1.45cm,y=0.64cm]
   \begin{scope}
    \clip (-2.65,-2.45) rectangle (0.45,2.45);
    \foreach \n in {-3,...,3} {
     \path[fill=black!10]
      (-2.65,{\n-0.5}) --
      plot[domain=-0.5:0.5,samples=45,variable=\t]
       ({-\t*\t},{\n+\t}) --(-2.65,{\n+0.5})--cycle;
    }
    \foreach \n in {-4,...,4} {
     \draw[black!40,line width=0.3pt]
      plot[domain=-1.7:1.7,samples=85,variable=\t]
       ({-\t*\t},{\n+\t});
    }
    \foreach \n in {-3,...,3} {
     \draw[black,line width=0.7pt]
      plot[domain=-0.5:0.5,samples=45,variable=\t]
       ({-\t*\t},{\n+\t});
    }
   \end{scope}
   \draw[->,darkgray] (-2.8,0)--(0.7,0) node[right] {$u$};
   \draw[->,darkgray] (0,-2.65)--(0,2.85) node[above] {$v$};
   \foreach \n in {-2,-1,1,2} {
    \draw[darkgray] (-0.025,\n)--(0.025,\n);
    \node[right,font=\small] at (0.06,\n) {$\n$};
   }
   \node[below right,font=\small] at (0,0) {$0$};
  \end{tikzpicture}
  \par\smallskip
  $u=\dfrac{2\varepsilon^2\Re\lambda}{\gamma\nu^2}$,\quad
  $v=\Im\lambda/\nu$
 \end{minipage}
 \caption{Trace-class spectra for a chain with noise only at its endpoints.
 (a) The half-lines $\bigcup_{n\in\Z}\{u\leq0,\ v=n\}$.
 (b) The shaded union $\bigcup_{n\in\Z}\{u\leq-(v-n)^2\}$.
 Thin curves show the constituent parabolas; bold arcs form the boundary
 of their union. Both generators have empty point spectrum.}
 \label{fig-boson-chain-spectra}
\end{figure}
\FloatBarrier

    \appendix

    \section{Real-linear operators} \label{section-real-linear-operators}

     Let $ S:\C^d\to\C^d $ be real-linear. There are unique complex matrices $ S_1,S_2 $ such that
    \begin{equation*}
        S z = S_1 z + S_2 \overline{z}, \quad \forall z \in \C^d.
    \end{equation*}

    Under the identification \eqref{eq-definition-iota}, its real matrix representation is
    \begin{equation} \label{eq-identification-real-linear-operators}
        \mathbf{S} := \begin{bmatrix}
            \Re S_1 + \Re S_2 & \Im S_2 - \Im S_1 \\ 
            \Im S_1 + \Im S_2 & \Re S_1 - \Re S_2
        \end{bmatrix}.
    \end{equation}
    We call $\mathbf S$ the \emph{real matrix representation} of $S$.

    Conversely, if a real $2d\times2d$ matrix is written in $d\times d$ blocks as
    \begin{equation*}
        \mathbf{S} = \begin{bmatrix}
            S_{11} & S_{12} \\ 
            S_{21} & S_{22}
        \end{bmatrix},
    \end{equation*}
    then its associated real-linear operator is
    \begin{equation*}
        S z = \left( \frac{S_{11} + S_{22}}{2} + \mi \, \frac{S_{21} - S_{12}}{2}  \right) z + \left( \frac{S_{11} - S_{22}}{2} + \mi \, \frac{S_{12} + S_{21}}{2} \right) \overline{z}.
    \end{equation*}

    The real adjoint is given by
    \begin{equation*}
        S^\sharp z = S_1^* z + S_2^T \overline{z}, \quad \forall z \in \C^d,
    \end{equation*}
    and is represented by $ \mathbf{S}^T $.

    The representation preserves the action and the real inner product.

    \begin{lemma}
        Let $S$ be real-linear and let $y,z\in\C^d$. Then
        \begin{equation*}
            \mathbf{S} \mathbf{z} = \begin{bmatrix}
                \Re S z \\ 
                \Im S z
            \end{bmatrix}, \quad \langle \mathbf{y}, \mathbf{z} \rangle = \Re \langle y, z \rangle.
        \end{equation*}
    \end{lemma}
    \begin{proof}
        Both identities follow directly from \eqref{eq-definition-iota} and \eqref{eq-identification-real-linear-operators}.
    \end{proof}

    \section{Proofs for stable drift} \label{appendix-stable-proofs}

    Throughout this appendix, the drift $ \mathbf{Z} $ is stable and $ \rho $ denotes the unique faithful normal invariant state of the semigroup. We retain the notation of Section \ref{eq-stable-drift}.

    \subsection{The eigenvalue criterion and spectral classification}

    Invariance of $ \rho $ and \eqref{eq-evolution-predual-semigroup-characteristic-function} give
    \begin{equation*}
        \phi_t (z) \chi_\rho (\me^{t Z} z) = \chi_\rho (z), \quad \forall t \ge 0.
    \end{equation*}
    The Gaussian factor can therefore be written as
    \begin{equation} \label{eq-phi-t-z-as-quotient}
        \phi_t (z) = \frac{\chi_\rho (z)}{\chi_\rho (\me^{t Z} z)},
    \end{equation} 
    which is well-defined because $ \chi_\rho $ never vanishes.

    For $ Y \in \cB_1(\mathsf{h}) $, define its \emph{relative characteristic function} by
    \begin{equation*}
        g_Y (z) := \frac{\chi_Y (z)}{\chi_\rho (z)}.
    \end{equation*}

    Equation \eqref{eq-evolution-predual-semigroup-characteristic-function} now becomes
    \begin{equation*}
        g_{\cT_{*t} (Y)} (z) = g_{Y} (\me^{t Z} z).
    \end{equation*}
    Thus division by $ \chi_\rho $ removes the Gaussian multiplier from the evolution. The following criterion makes this reduction precise.

    \begin{lemma} \label{lemma-eigenvalue-criterion}
        Let $ Y \in \cB_1(\mathsf{h}) $ and $ \lambda \in \C $. Then $ Y \in \dom\cL_* $ and
        \begin{equation} \label{eq-cL-star-Y-lambda-Y}
            \cL_* (Y) = \lambda Y
        \end{equation}
        if and only if 
        \begin{equation} \label{eq-gY-etZ-gY-z}
            g_Y (\me^{t Z} z) = \me^{\lambda t} g_Y (z), \quad t \ge 0, \quad z \in \C^d.
        \end{equation}
    \end{lemma}

    \begin{proof}
        Suppose that $ Y \in \dom\cL_* $ and \eqref{eq-cL-star-Y-lambda-Y} holds. Then $ \cT_{*t}(Y) = \me^{t\lambda}Y $. Taking characteristic functions and using \eqref{eq-evolution-predual-semigroup-characteristic-function} and \eqref{eq-phi-t-z-as-quotient}, we obtain
        \begin{equation} \label{eq-median-computations}
            \chi_Y (\me^{t Z} z) \frac{\chi_\rho (z)}{\chi_\rho (\me^{t Z} z)} =\Tr ( Y W(\me^{t Z} z) ) \phi_t (z) = \Tr( \cT_{*t} (Y) W(z) ) = \Tr ( \me^{t \lambda} Y W(z)) = \me^{t \lambda} \chi_Y (z),
        \end{equation}
        which immediately implies \eqref{eq-gY-etZ-gY-z}. 

        Conversely, suppose that \eqref{eq-gY-etZ-gY-z} holds. Reversing the calculation in \eqref{eq-median-computations} and using injectivity of the Fourier--Weyl transform gives $ \cT_{*t}(Y) = \me^{t\lambda}Y $. Hence $ Y \in \dom\cL_* $ and $ \cL_*(Y) = \lambda Y $.
    \end{proof}

    \begin{proof}[Proof of Theorem \ref{theorem-strict-negative-implies-point-spectrum}]
        Choose $ \lambda \in \sigma(\mathbf{Z}) $. By Lemma \ref{lemma-existence-Q-lambda}, there is a nonzero positive semidefinite matrix $ Q_\lambda \in M_{2d}(\R) $ such that
        \begin{equation} \label{eq-lyapunov-type-identity-Z-Q-lambda}
            \mathbf{Z}^T Q_\lambda + Q_\lambda \mathbf{Z} = 2 \Re(\lambda) Q_\lambda.
        \end{equation}
        Using \eqref{eq-lyapunov-type-identity-Z-Q-lambda}, we compute
        \begin{equation*}
            \frac{\md}{\md t} \left( \me^{- 2 \Re(\lambda) t} \me^{t \mathbf{Z}^T} Q_\lambda \me^{t \mathbf{Z}} \right) = \me^{-2 \Re (\lambda) t} \me^{t \mathbf{Z}^T} \left( - 2 \Re(\lambda) Q_\lambda + \mathbf{Z}^T Q_\lambda + Q_\lambda \mathbf{Z} \right) \me^{t \mathbf{Z}} = 0,
        \end{equation*}
        which implies that 
        \begin{equation} \label{eq-sandwiched-Q-lambda}
            \me^{t \mathbf{Z}^T} Q_\lambda \me^{t \mathbf{Z}} = \me^{2 \Re(\lambda) t} Q_\lambda, \quad \forall t \in \R.
        \end{equation}
    
        Define the nonnegative quadratic form on real phase space by
        \begin{equation} \label{eq-definition-q-lambda}
            q_{\lambda} (z) := \mathbf{z}^T Q_\lambda \mathbf{z}. 
        \end{equation}
        It follows that
        \begin{equation*}
            q_\lambda(\me^{t Z} z) = \mathbf{z}^T \me^{t \mathbf{Z}^T} Q_\lambda \me^{t \mathbf{Z}} \mathbf{z} = \me^{2 \Re(\lambda) t} \mathbf{z}^T Q_\lambda \mathbf{z} = \me^{2 \Re(\lambda) t} q_\lambda (z).
        \end{equation*}

        To obtain the scaling required by Lemma \ref{lemma-eigenvalue-criterion}, fix $ \eta \in \C $ with $ \Re\eta < 0 $ and define
        \begin{equation*}
            g_\eta (z) := ( q_\lambda (z) )^{\eta / (2 \Re \lambda)} := \begin{cases}
                \exp{ \frac{\eta}{2 \Re \lambda} \log q_\lambda (z)  }, & q_\lambda (z) \neq 0, \\
                0, & q_\lambda (z) = 0,
            \end{cases}
        \end{equation*} 
        where the logarithm is real. This defines a continuous function because $ \Re\eta/(2\Re\lambda) > 0 $. Indeed, $ q_\lambda(z) = 0 $ precisely when $ \mathbf{z}\in\ker Q_\lambda $, and $ \abs{g_\eta(z)} = q_\lambda(z)^{\Re\eta/(2\Re\lambda)} \to 0 $ as $ q_\lambda(z)\to0^+ $.

        Equation \eqref{eq-sandwiched-Q-lambda} also shows that $ \mathbf{z}\in\ker Q_\lambda $ if and only if $ \me^{t\mathbf{Z}}\mathbf{z}\in\ker Q_\lambda $. When $ q_\lambda(z)>0 $, we therefore have
        \begin{equation} \label{eq-g-eta-satisfies-scaling}
            g_\eta ( \me^{t Z} z ) = \exp{ \frac{\eta}{2 \Re \lambda} \log q_\lambda(\me^{t Z} z) } = \exp{\frac{\eta}{2 \Re \lambda} \log ( \me^{2 \Re(\lambda) t}  q_\lambda(z) ) } = \me^{\eta t} g_\eta (z). 
        \end{equation}
        If $ q_\lambda(z)=0 $, both sides of \eqref{eq-g-eta-satisfies-scaling} vanish. Thus the identity holds for every $ z\in\C^d $.

        Set $ \alpha := \eta/(2\Re\lambda) $, so that $ \Re\alpha>0 $. Lemma \ref{lemma-construction-Y-alpha} constructs a nonzero trace-class operator $ Y_\alpha $ satisfying $ \chi_{Y_\alpha}(z)=\chi_\rho(z)q_\lambda(z)^\alpha $. Hence $ g_\eta $ is its relative characteristic function. Equations \eqref{eq-g-eta-satisfies-scaling} and \eqref{eq-gY-etZ-gY-z} imply $ \cL_*(Y_\alpha)=\eta Y_\alpha $, as required.
    \end{proof}

    \begin{proof}[Proof of Theorem \ref{theorem-stable-spectrum}]
        Since $ \cL_* $ generates a strongly continuous contraction semigroup, the Hille--Yosida theorem and Theorem \ref{theorem-strict-negative-implies-point-spectrum} give
        \begin{equation*}
            \{ \eta : \Re \eta < 0 \} \subset \sigma_{\mathrm{p}} (\cL_*) \subset \sigma (\cL_*) \subset \{ \eta : \Re \eta \le 0 \}.
        \end{equation*}
        Because the spectrum is closed, we have 
        \begin{equation*}
            \sigma (\cL_*) = \{ \eta : \Re \eta \le 0  \}.
        \end{equation*}
        The invariant state satisfies $ \cL_*(\rho)=0 $, so $ 0\in\sigma_{\mathrm{p}}(\cL_*) $. To exclude other imaginary eigenvalues, suppose that $ \theta\in\R\setminus\{0\} $ and $ \cL_*(Y)=\mi\theta Y $ for a nonzero trace-class operator $ Y $. Then
        \begin{equation*}
            \cT_{*t} (Y) = \me^{\mi \theta t} Y.
        \end{equation*}
        Taking characteristic functions and using \eqref{eq-evolution-predual-semigroup-characteristic-function}, we obtain
        \begin{equation*}
            \me^{\mi \theta t} \chi_Y (z) = \phi_t (z) \chi_Y (\me^{t Z} z).
        \end{equation*}
        By Lemma \ref{lemma-eigenvalue-criterion} we have $ g_Y (\me^{t Z} z) = \me^{\mi \theta t} g_Y (z) $. Taking $ z = 0 $ shows $ g_Y (0) = \me^{\mi \theta t} g_Y (0) $ for all $ t \ge 0 $ and hence $ g_Y (0) = \chi_Y (0) / \chi_\rho (0) = \Tr (Y) = 0 $. Since $ Y \neq 0 $, injectivity of the Fourier--Weyl transform gives a $ z_0 $ such that $ g_Y (z_0) \neq 0 $. Stability of $ Z $ and continuity of $ g_Y $ yield 
        \begin{equation*}
            g_Y (\me^{t Z} z_0) \rightarrow g_Y (0) = 0,
        \end{equation*}
        whereas 
        \begin{equation*}
            \abs{\me^{\mi \theta t} g_Y (z_0)} = \abs{g_Y (z_0)} > 0, \quad \forall t \ge 0,
        \end{equation*}
        a contradiction. This proves $ \sigma_{\mathrm{p}}(\cL_*)\cap\mi\R=\{0\} $ and hence \eqref{eq-point-spectrum-cL-star}. Finally, \cite[Chapter IV, Proposition 1.10]{EngelNagel2000} gives
        \begin{equation*}
            \mi \R = \partial \sigma (\cL_*) \subset \sigma_{\mathrm{ap}} (\cL_*),
        \end{equation*}
        where $ \partial\sigma(\cL_*) $ is the topological boundary of the spectrum. The open left half-plane already belongs to the point spectrum. Therefore
        \begin{equation*}
            \sigma_{\mathrm{ap}} (\cL_*) = \sigma (\cL_*) = \{ \eta : \Re \eta  \le 0 \},
        \end{equation*}
        which is exactly \eqref{eq-full-spectrum-cL-star}.
    \end{proof}

    \subsection{Trace-class construction and Gaussian diffusion estimates}

    The following lemmas justify the finite-dimensional construction and its realization in trace norm.

    \begin{lemma} \label{lemma-existence-Q-lambda}
        Let $ n\ge2 $ and $ A\in M_n(\R) $. For every $ \lambda\in\sigma(A) $, there is a nonzero positive semidefinite matrix $ Q_\lambda\in M_n(\R) $ of rank at most two such that
        \begin{equation} \label{eq-lemma-existence-Q-lambda}
            A^T Q_\lambda + Q_\lambda A = 2 \Re(\lambda) Q_\lambda.
        \end{equation}
    \end{lemma}
    \begin{proof}
        Since $ \sigma(A)=\sigma(A^T) $, choose $ v\in\C^n\setminus\{0\} $ with $ A^Tv=\lambda v $, and set
        \begin{equation} \label{eq-definition-Q-lambda}
            Q_\lambda := \Re ( v v^* ) = \Re v (\Re v)^T + \Im v (\Im v)^T.
        \end{equation}
        Therefore,
        \begin{align*}
            A^T Q_\lambda + Q_\lambda A = \Re \left( A^T v v^* + v v^* A \right) = \Re \left( (A^T v) v^* + v (A^T v)^* \right) = 2 \Re (\lambda) Q_\lambda,
        \end{align*}
        which proves \eqref{eq-lemma-existence-Q-lambda}. Formula \eqref{eq-definition-Q-lambda} also shows that $ Q_\lambda $ is real, nonzero, positive semidefinite, and has rank at most two.
    \end{proof}

    \begin{lemma} \label{lemma-construction-Y-alpha}
        Let $ Q_\lambda $ and $ q_\lambda $ be as in \eqref{eq-definition-q-lambda}. For every $ \alpha\in\C $ with $ \Re\alpha>0 $, there is a nonzero trace-class operator $ Y_\alpha $ such that
        \begin{equation} \label{eq-characteristic-function-Y-alpha-equal-to-rho-q-lambda}
            \chi_{Y_\alpha} ( z ) = \chi_\rho (z) q_\lambda (z)^{\alpha}.
        \end{equation}
    \end{lemma}
    \begin{proof}
        Define the auxiliary Gaussian diffusion semigroup $ (\cG_t)_{t\ge0} $ by
        \begin{equation} \label{eq-definition-auxiliary-quantum-diffusion-semigroup-G-t}
            \cG_t (W(z)) := \me^{-t q_\lambda (z)} W(z) = \exp{ - t \, \mathbf{z}^T Q_\lambda \mathbf{z} } W(z).
        \end{equation}
        Its drift and displacement are zero, and its diffusion matrix is $ 2Q_\lambda $.

        We first choose a Gaussian state with slightly smaller covariance. Faithfulness of $ \rho $ gives $ \mathbf{S}-\mi\mathbf{J}=\overline{\mathbf{S}+\mi\mathbf{J}}>0 $. Choose $ \epsilon>0 $ sufficiently small that
        \begin{equation*}
            \mathbf{S} - 2 \epsilon Q_\lambda - \mi \mathbf{J} > 0.
        \end{equation*}
        Since $ 0\le Q_\lambda\le\norm{Q_\lambda}\1 $, it is enough to take
        \begin{equation*}
            0 < \epsilon < \frac{\lambda_{\min} (\mathbf{S} - \mi \mathbf{J})}{2 \norm{Q_\lambda}},
        \end{equation*}
        where $ \lambda_{\min} $ denotes the smallest eigenvalue. Let $ \sigma $ be the Gaussian state with mean $ \omega $ and covariance matrix $ \mathbf{S}-2\epsilon Q_\lambda $. By \eqref{eq-characteristic-function-rho},
        \begin{align} 
            \chi_\sigma (z) &= \exp{ 2 \mi \, \boldsymbol{\omega}^T \mathbf{J} \mathbf{z} - \frac{1}{2} \mathbf{z}^T ( \mathbf{S} - 2 \epsilon Q_\lambda ) \mathbf{z} } \nonumber \\
            &= \me^{\epsilon \, \mathbf{z}^T Q_\lambda \mathbf{z} } \cdot \exp{ 2 \mi \, \boldsymbol{\omega}^T \mathbf{J} \mathbf{z} -\frac{1}{2} \mathbf{z}^T \mathbf{S} \mathbf{z} } \nonumber \\
            &= \me^{\epsilon \, q_\lambda(z) } \cdot \chi_\rho (z). \label{eq-relation-sigma-and-rho}
        \end{align}

        The auxiliary semigroup evolves $ \sigma $ into $ \rho $ at time $ \epsilon $. Indeed, for $ t\ge0 $,
        \begin{equation} \label{eq-characteristic-function-evolution-of-sigma-under-auxiliary-QMS}
            \chi_{\cG_{* t} (\sigma) } (z) = \Tr( \cG_{* t} (\sigma) W(z) ) = \Tr( \sigma \cG_{t} (W(z)) ) =  \me^{- t \mathbf{z}^T Q_\lambda \mathbf{z} } \chi_\sigma (z) = \me^{-t q_\lambda(z)} \chi_\sigma (z).
        \end{equation}
        Setting $ t=\epsilon $ and using \eqref{eq-relation-sigma-and-rho} yields
        \begin{equation*}
            \chi_{\cG_{*\epsilon} (\sigma)} (z) = \me^{-\epsilon q_\lambda(z)} \me^{\epsilon q_\lambda (z)} \chi_\rho(z) = \chi_\rho (z),
        \end{equation*} 
        and $ \cG_{*\epsilon} (\sigma) = \rho $ follows from the injectivity of the Fourier--Weyl transform.

        To estimate the time derivatives of $ \cG_{*t} $, let $ r=\rank Q_\lambda $ and choose a real $ 2d\times r $ matrix $ R $ such that
        \begin{equation} \label{eq-definition-R}
            R R^T = \frac{1}{4} \mathbf{J}^T Q_\lambda \mathbf{J}.
        \end{equation}
        Let $ f_t $ be the heat kernel on $ \R^r $,
        \begin{equation} \label{eq-definition-f-t-heat-kernel}
            f_t (u) := (4 \pi t)^{-r/2} \exp{ - \frac{\abs{u}^2}{4 t} }, \quad u \in \R^r, \quad t > 0.
        \end{equation}
        Then $ \int_{\R^r}f_t(u)\,\md u=1 $. For $ Y\in\cB_1(\mathsf{h}) $, Lemma \ref{eq-equiv-form-action-auxiliary-QMS} gives
        \begin{equation*}
            \cG_{*t } (Y) = \int_{\R^r} f_t (u) W(\iota^{-1} (R u)) Y W( \iota^{-1} ( R u) )^* \md u.
        \end{equation*}
        Lemma \ref{lemma-L1-norm-estimate-time-derivative-heat-kernel} provides the required bounds on the time derivatives of $ f_t $. The explicit formula and dominated convergence on compact subintervals of $ (0,\infty) $ also show that $ t\mapsto f_t $ belongs to $ C^\infty((0,\infty);L^1(\R^r)) $. Thus, for $ t>0 $,
        \begin{align*}
            &\quad \norm{ \frac{\cG_{*(t+h)} (Y) - \cG_{*t} (Y)}{h} - \int_{\R^r} \frac{\md}{\md t} f_t (u) W( \iota^{-1} (R u) ) Y W( \iota^{-1} (R u) )^* \md u }_1 \\
            &\le \norm{ \frac{f_{t + h} (u) - f_t (u)}{h} - \frac{\md}{\md t} f_t (u) }_{L^1(\R^r; \md u)} \norm{Y}_1 \rightarrow 0
        \end{align*}
        as $ h \rightarrow 0 $. Iterating this argument gives 
        \begin{equation*}
            \frac{\md^m}{\md t^m} \cG_{*t} (Y) = \int_{\R^r} \frac{\md^m}{\md t^m} f_t (u) W( \iota^{-1} (R u) ) Y W( \iota^{-1} (R u) )^* \md u.
        \end{equation*}
        Consequently, we have
        \begin{align}
            \norm{ \frac{\md^m}{\md t^m} \cG_{*t} (Y) }_1 &\le \norm{ \int_{\R^r} \frac{\md^m}{\md t^m} f_t (u) \, W( \iota^{-1} (R u ) ) Y  W( \iota^{-1} (R u ) )^* \, \md u }_1 \nonumber \\ 
            &\le \int_{\R^r} \abs{ \frac{\md^m}{\md t^m} f_t (u) } \norm{ W( \iota^{-1} (R u ) ) Y  W( \iota^{-1} (R u ) )^* }_1  \md u \nonumber \\ 
            &\le \int_{\R^r} \abs{ \frac{\md^m}{\md t^m} f_t (u) } \md u \cdot \norm{Y}_1 = \norm{Y}_1 \norm{ \frac{\md^m}{\md t^m} f_t (u) }_{L^1 (\R^r; \md u)} \nonumber \\ 
            &\le c_{r, m} t^{-m} \norm{Y}_1, \label{eq-L1-norm-estimate-time-derivative-auxiliary-QMS}
        \end{align}
        In particular, all these derivatives exist in trace norm.

        Choose an integer $ m>\Re\alpha $ and define
        \begin{equation*}
            Y_\alpha := \frac{(-1)^m}{\Gamma (m-\alpha)} \int_0^\infty t^{m - \alpha - 1} \left. \frac{\md^m}{\md s^m} \cG_{*s} (\sigma) \right\vert_{s = \epsilon + t} \md t.
        \end{equation*}
        This Bochner integral converges in trace norm. Indeed, \eqref{eq-L1-norm-estimate-time-derivative-auxiliary-QMS} gives
        \begin{align}
            \norm{Y_\alpha}_1 &\le \frac{1}{\abs{ \Gamma (m - \alpha) }} \int_0^\infty \abs{ t^{m - \alpha - 1} } \norm{ \left. \frac{\md^m}{\md s^m} \cG_{*s} (\sigma) \right\vert_{s = \epsilon + t} }_1 \md t \nonumber \\ 
            &\le \frac{c_{r, m} \norm{\sigma}_1}{\abs{\Gamma(m - \alpha)}} \int_0^\infty t^{m - \Re \alpha - 1} (\epsilon + t)^{-m} \md t \label{eq-L1-norm-estimate-Y-alpha-part-1}.
        \end{align}
        Since $ \sigma $ is a state, $ \norm{\sigma}_1=1 $. The substitution $ t=\epsilon s $ and the beta integral \cite[Eqs. (5.12.1) and (5.12.3)]{Olver2010NIST} yield
        \begin{equation} \label{eq-beta-function-integral}
            \int_0^\infty t^{m - \Re \alpha - 1} (\epsilon + t)^{-m} \md t = \epsilon^{- \Re \alpha} \int_0^\infty s^{m - \Re \alpha - 1} (1 + s)^{-m} \md s = \epsilon^{-\Re \alpha} \frac{\Gamma(m - \Re \alpha) \Gamma (\Re \alpha)}{\Gamma (m)}.
        \end{equation}
        Combining \eqref{eq-L1-norm-estimate-Y-alpha-part-1} and \eqref{eq-beta-function-integral}, we obtain
        \begin{equation*}
            \norm{Y_\alpha}_1 \le c_{r, m} \, \epsilon^{-\Re \alpha} \, \frac{\Gamma (m - \Re \alpha) \Gamma (\Re \alpha)}{\abs{\Gamma(m-\alpha) } \Gamma (m)},
        \end{equation*}
        which is finite because $ 0<\Re\alpha<m $.
    
        By Lemma \ref{lemma-m-th-order-derivative} and \eqref{eq-relation-sigma-and-rho}, if $ q_\lambda(z)>0 $, then
        \begin{align*}
            \chi_{Y_\alpha} (z) &= \frac{(-1)^m}{\Gamma (m-\alpha)} \int_0^\infty t^{m - \alpha - 1}  \left. \Tr( \frac{\md^m}{\md s^m} \cG_{*s} (\sigma) W(z)) \right\vert_{s = \epsilon + t} \md t  \\ 
            &= \frac{(-1)^m}{\Gamma (m-\alpha)} \int_0^\infty t^{m - \alpha - 1} \left. \frac{\md^m}{\md s^m} \chi_{ \cG_{*s} (\sigma) } (z)  \right\vert_{s = \epsilon + t} \md t  \\ 
            &= \frac{(-1)^m (- q_\lambda(z))^m \me^{-\epsilon q_\lambda(z)} \chi_\sigma (z) }{\Gamma (m-\alpha)} \int_0^\infty t^{m - \alpha - 1} \me^{- t q_\lambda(z)} \md t  \\ 
            &= \frac{\chi_\rho(z) q_\lambda (z)^m }{\Gamma (m - \alpha)} \frac{\Gamma (m - \alpha)}{q_\lambda (z)^{m - \alpha}} = \chi_\rho (z) q_\lambda (z)^\alpha.
        \end{align*}
        This proves \eqref{eq-characteristic-function-Y-alpha-equal-to-rho-q-lambda} when $ q_\lambda(z)>0 $. If $ q_\lambda(z)=0 $, the differentiated characteristic function vanishes, so
        \begin{equation*}
            \chi_{Y_\alpha} (z) = 0 = \chi_\rho (z) q_\lambda (z)^\alpha.
        \end{equation*}
        Thus \eqref{eq-characteristic-function-Y-alpha-equal-to-rho-q-lambda} holds everywhere. Since $ Q_\lambda\neq0 $, there is $ z_0 $ with $ q_\lambda(z_0)>0 $. Hence $ \chi_{Y_\alpha}(z_0)\neq0 $ and $ Y_\alpha\neq0 $.
    \end{proof}

    We now justify the Weyl-conjugation representation of the auxiliary diffusion semigroup. Quantum diffusion semigroups of this form are studied in \cite{DattaPautratRouze2017}.

    \begin{lemma} \label{eq-equiv-form-action-auxiliary-QMS}
        Let $ f_t $ and $ R $ be as in \eqref{eq-definition-f-t-heat-kernel} and \eqref{eq-definition-R}. For $ t>0 $ and $ X\in\cB(\mathsf{h}) $, the auxiliary semigroup satisfies
        \begin{equation*}
            \cG_t (X) = \int_{\R^r} f_t (u) W(\iota^{-1} (R u))^* X W(\iota^{-1} (R u)) \md u,
        \end{equation*}
        where the integral is understood $ \sigma $-weakly, that is, for all $ Y \in \cB_1 (\mathsf{h}) $, we have 
        \begin{equation} \label{eq-auxiliary-qms-equivalent-form-sigma-weak-sense}
            \Tr( Y \cG_t (X) ) = \int_{\R^r} f_t (u) \Tr ( Y W(\iota^{-1} (Ru))^* X W( \iota^{-1} (R u)) ) \md u.
        \end{equation}
        Consequently, we have 
        \begin{equation} \label{eq-auxiliary-qms-predual}
            \cG_{*t } (Y) = \int_{\R^r} f_t (u) W(\iota^{-1} (R u)) Y W( \iota^{-1} ( R u) )^* \md u.
        \end{equation}
    \end{lemma}
    \begin{proof}
        Fix $ X\in\cB(\mathsf{h}) $. For every $ Y\in\cB_1(\mathsf{h}) $,
        \begin{align*}
            &\quad \abs{\int_{\R^r} f_t (u) \Tr ( Y W(\iota^{-1} (Ru))^* X W( \iota^{-1} (R u)) ) \md u} \\
            &\le \int_{\R^r} \abs{f_t (u)} \norm{Y}_1  \norm{ W(\iota^{-1} (Ru))^* X W( \iota^{-1} (R u)) } \md u \\
            &\le \int_{\R^r} \abs{f_t (u)} \norm{Y}_1 \norm{X} \md u \le \norm{Y}_1 \norm{X}.
        \end{align*}
        The right-hand side of \eqref{eq-auxiliary-qms-equivalent-form-sigma-weak-sense} therefore defines a bounded functional on $ \cB_1(\mathsf{h}) $. Denote the corresponding bounded operator by $ \widetilde{\cG}_t(X) $.

        For $ Y\in\cB_1(\mathsf{h}) $, consider
        \begin{equation} \label{eq-an-bochner-integral}
            \int_{\R^r} f_t (u) W(\iota^{-1} (R u)) Y W( \iota^{-1} ( R u) )^* \md u
        \end{equation} 
        as a Bochner integral in $ \cB_1(\mathsf{h}) $. The map
        \begin{equation*}
            u \mapsto W(\iota^{-1} (Ru)) Y W(\iota^{-1} (Ru))^* 
        \end{equation*}
        is continuous in trace norm, as follows by finite-rank approximation. Moreover,
        \begin{equation*}
            \int_{\R^r} \norm{f_t (u) W(\iota^{-1} (R u)) Y W( \iota^{-1} ( R u) )^* }_1 \md u  \le \norm{Y}_1  < \infty,
        \end{equation*}
        so that \eqref{eq-an-bochner-integral} is a well-defined trace-class operator.
    
        Using cyclicity of the trace, we obtain
        \begin{equation*}
            \Tr (Y \widetilde{\cG}_t (X)) = \Tr( \left( \int_{\R^r} f_t(u) W(\iota^{-1} (R u)) Y W(\iota^{-1} (R u))^* \md u \right) X ),
        \end{equation*}
        so $ \widetilde{\cG}_t $ is normal and its predual is the right-hand side of \eqref{eq-auxiliary-qms-predual}.

        It remains to identify $ \widetilde{\cG}_t $ with $ \cG_t $. For a Weyl operator, \eqref{eq-definition-R} and the Gaussian integral give
        \begin{align*}
            \widetilde{\cG}_t (W(z)) &= \int_{\R^r} f_t (u) W(\iota^{-1} (R u))^* W(z) W(\iota^{-1} (R u)) \md u \\ 
            &= \int_{\R^r} f_t (u) \me^{ - 2 \mi \, \Im \langle z, \iota^{-1} (R u)  \rangle } W(z) \md u \\ 
            &= \int_{\R^r} f_t (u) \me^{ - 2 \mi \, \mathbf{z}^T \mathbf{J} R u } \, W(z) \md u \\
            &= \int_{\R^r} (4 \pi t)^{-r / 2} \exp{ - \frac{1}{4t} u^T \1_{r} u  - 2 \mi \mathbf{z}^T \mathbf{J} R u } W(z) \md u \\ 
            &= (4 \pi t)^{-r/2} (4 \pi t)^{r/2} \exp{  \frac{1}{2} (2 \mi \mathbf{z}^T \mathbf{J} R )  ( 2 t \1_r )   (2 \mi \mathbf{z}^T \mathbf{J} R )^T } W(z) \\ 
            &= \exp{- t \, \mathbf{z}^T Q_\lambda \mathbf{z} } W(z) = \cG_t (W(z)).
        \end{align*}
        These equalities hold in the $ \sigma $-weak sense. Normality and $ \sigma $-weak density of the span of the Weyl operators imply $ \widetilde{\cG}_t=\cG_t $. The predual formula follows from cyclicity of the trace.
    \end{proof}

    \begin{lemma} \label{lemma-m-th-order-derivative}
        For every integer $ m \ge 1 $ and every $ t > 0 $, we have
        \begin{equation} \label{eq-m-th-order-derivative}
            \chi_{ \md^m \cG_{*t} (\sigma) / \md t^m } (z) = \frac{\md^m}{\md t^m} \chi_{ \cG_{*t} (\sigma) } (z).
        \end{equation}
        Moreover, with $ q_\lambda $ as in \eqref{eq-definition-q-lambda},
        \begin{equation*}
            \frac{\md^m}{\md t^m} \chi_{ \cG_{*t} (\sigma) } (z) = ( - q_\lambda (z) )^m \me^{- t q_\lambda (z) } \chi_\sigma(z),
        \end{equation*}
        for every $ z\in\C^d $, with value zero when $ q_\lambda(z)=0 $.
    \end{lemma}
    \begin{proof}
        For $ m=1 $, continuity of $ Y\mapsto\Tr(YW(z)) $ in trace norm gives
        \begin{align*}
            \frac{\md}{\md t} \chi_{\cG_{*t} (\sigma) } (z) &= \lim_{h \rightarrow 0} \frac{\chi_{\cG_{*(t + h)} (\sigma)} (z) - \chi_{ \cG_{*(t)} (\sigma)}  (z) }{h} \\ 
            &= \lim_{h \rightarrow 0} \frac{\Tr( (\cG_{*(t+h)} (\sigma) - \cG_{*(t)} (\sigma) ) W(z) )}{h} \\
            &= \Tr( \lim_{h \rightarrow 0}  \frac{ \cG_{*(t+h)} (\sigma) - \cG_{*(t)} (\sigma) }{h} \, W(z) ) \\ 
            &= \Tr( \frac{\md}{\md t} \cG_{*t} (\sigma) W(z) ) = \chi_{ \md \cG_{*t} (\sigma) / \md t } (z).
        \end{align*} 
        Iteration proves \eqref{eq-m-th-order-derivative}. Differentiating \eqref{eq-characteristic-function-evolution-of-sigma-under-auxiliary-QMS} yields
        \begin{equation*}
            \frac{\md^m}{\md t^m} \chi_{\cG_{*t} (\sigma) } (z) = \frac{\md^m}{\md t^m} \me^{-t q_\lambda(z) } \chi_\sigma (z) = (-q_\lambda(z))^m \me^{-t q_\lambda(z) } \chi_\sigma (z),
        \end{equation*}
        which is zero when $ q_\lambda(z)=0 $.
    \end{proof}

    \begin{lemma} \label{lemma-L1-norm-estimate-time-derivative-heat-kernel}
        For the heat kernel \eqref{eq-definition-f-t-heat-kernel}, and every integer $ m\ge1 $,
        \begin{equation} \label{eq-L1-norm-estimate-time-derivative-heat-kernel}
            \norm{ \frac{\md^m}{\md t^m} f_t (u)  }_{L^1 (\R^r)} \le c_{r, m} \, t^{-m},
        \end{equation}
        where $ c_{r,m} $ depends only on $ r $ and $ m $.
    \end{lemma}
    \begin{proof}
        Direct differentiation of the heat kernel gives a constant $ c_{r,m}^{\prime} $, depending only on $ r $ and $ m $, such that
        \begin{equation*}
            \abs{ \frac{\md^m}{\md t^m} f_t (u) } \le c_{r, m}^\prime \, \frac{1}{t^{r/2 + m}} \left( 1 + \frac{\abs{u}^2}{t} \right)^m \exp{- \frac{\abs{u}^2}{4 t}}.
        \end{equation*}
        With $ u=\sqrt{t}v $ and
        \begin{equation*}
            c_{r, m} = c_{r, m}^\prime \int_{\R^r} (1 + \abs{v}^2)^m \exp{-\abs{v}^2 / 4} \, \md v,
        \end{equation*}
        we obtain
        \begin{align*}
            \norm{ \frac{\md^m}{\md t^m} f_t (u) }_{L^1 (\R^r; \md u)} &\le c_{r, m}^\prime \, t^{- r/2 - m} \int  \left( 1 + \abs{u}^2 / t \right)^m \exp{- \abs{u}^2 / (4 t) } \, \md u \\ 
            &\le c_{r, m}^\prime \, t^{-m} \int (1 + \abs{v}^2)^m \exp{ - \abs{v}^2 / 4 } \, \md v \le c_{r, m} \, t^{-m},
        \end{align*}
        which completes the proof of \eqref{eq-L1-norm-estimate-time-derivative-heat-kernel}.
    \end{proof}

    \subsection{The thermal oscillator}
    \label{subsection-proof-thermal-oscillator}

    \begin{proof}[Proof of Example \ref{example-thermal-oscillator}]
        Here $ Zz=(-\beta_0/2+\mi\omega)z $, $ C=(\Gamma_-+\Gamma_+)\1 $, and the displacement is zero. Hence
        \begin{equation*}
            \mathbf{C}_{\mathbf{Z}}
            =(\Gamma_-+\Gamma_+)\1_2-\mi\beta_0\mathbf{J}>0,
            \qquad
            \chi_\rho(z)=\exp{-\frac{1+r}{2(1-r)}\abs{z}^2}.
        \end{equation*}
        The first inequality follows since the two eigenvalues are $ 2\Gamma_- $ and $ 2\Gamma_+ $. It implies irreducibility. The binomial formula, together with analyticity of $ (1-rw)^{-4/3} $ on a disk of radius greater than one, gives
        \begin{equation*}
            c_n\sim\frac{(1-r)^{-4/3}}{\Gamma(-1/3)}n^{-4/3}.
        \end{equation*}
        Consequently, $ \sum_n\abs{c_n}<\infty $, so $ Y $ is a nonzero trace-class operator. The same asymptotic also gives $ \sum_n n\abs{c_n}=\infty $.

        To verify the generator-domain assertion, write $ x=\abs{z}^2 $. Normal ordering of the Weyl operator gives $ \langle e(n),W(z)e(n)\rangle=\me^{-x/2}L_n(x) $, where $ L_n(x)=\sum_{j=0}^n(-1)^j\binom{n}{j}x^j/j! $ is the $ n $th Laguerre polynomial. Thus $ \chi_Y(z)=h(x) $, where
        \begin{equation*}
            h(x):=\sum_{n=0}^\infty c_n\me^{-x/2}L_n(x).
        \end{equation*}
        Since $ \abs{\me^{-x/2}L_n(x)}\le1 $, termwise Laplace transformation is justified for every $ s>0 $. Using the elementary Laplace transform of $ L_n $ and \eqref{eq-thermal-example-coefficients}, we obtain
        \begin{align*}
            \int_0^\infty\me^{-sx}h(x)\,\md x
            &=\frac{1}{s+1/2}\sum_{n=0}^\infty c_n
                \left(\frac{s-1/2}{s+1/2}\right)^n \\
            &=\left((1-r)s+\frac{1+r}{2}\right)^{-4/3}.
        \end{align*}
        Uniqueness of the Laplace transform, and continuity, give
        \begin{equation*}
            \chi_Y(z)
            =\frac{\abs{z}^{2/3}}{(1-r)^{4/3}\Gamma(4/3)}\chi_\rho(z).
        \end{equation*}
        Since $ \abs{\me^{tZ}z}=\me^{-\beta_0 t/2}\abs{z} $, the relative characteristic function satisfies $ g_Y(\me^{tZ}z)=\me^{-\beta_0 t/3}g_Y(z) $. Lemma \ref{lemma-eigenvalue-criterion} now yields $ Y\in\dom\cL_* $ and $ \cL_*(Y)=-\beta_0 Y/3 $, as required.
    \end{proof}

    \subsection{Relation with the classical invariant-measure result}
    \label{subsection-stable-classical-comparison}

    We justify the intertwining assertion in the introduction. Let $ \mathbf{S} $ and $ \omega $ be the covariance matrix and mean of the faithful invariant state $ \rho $. Choose $ 0<\epsilon<1 $ sufficiently small that $ (1-\epsilon)\mathbf{S}+\mi\mathbf{J}>0 $. For $ x\in\R^{2d} $, let $ \sigma_x $ be the Gaussian state with characteristic function
    \begin{equation*}
        \chi_{\sigma_x}(z)
        =\exp{2\mi\Im\langle\omega,z\rangle+\mi x^T\mathbf{z}
          -\frac{1-\epsilon}{2}\mathbf{z}^T\mathbf{S}\mathbf{z}}.
    \end{equation*}
    These states are Weyl translates of $ \sigma_0 $, so $ x\mapsto\sigma_x $ is continuous in trace norm. Let $ \mu $ be the centered Gaussian probability measure with covariance $ \epsilon\mathbf{S} $, and define the Bochner integral
    \begin{equation*}
        \cK f := \int_{\R^{2d}} f(x)\sigma_x\,\mu(\md x),
        \qquad f\in L^1(\mu).
    \end{equation*}
    Since $ \norm{\sigma_x}_1=1 $, this defines a bounded map $ \cK:L^1(\mu)\to\cB_1(\mathsf{h}) $ satisfying $ \norm{\cK f}_1\le\norm{f}_{L^1(\mu)} $. Moreover,
    \begin{equation*}
        \chi_{\cK f}(z)
        =\chi_{\sigma_0}(z)\int_{\R^{2d}}\me^{\mi x^T\mathbf{z}}f(x)\,\mu(\md x).
    \end{equation*}
    The Gaussian factor is nowhere zero. Uniqueness of the Fourier transform of finite measures therefore shows that $ \cK $ is injective.

    Let $ p_t(x,\md y) $ be the Gaussian transition kernel with mean $ \me^{t\mathbf{Z}^T}x $ and covariance $ \epsilon\mathbf{S}_t $. The identity $ \mathbf{S}=\mathbf{S}_t+\me^{t\mathbf{Z}^T}\mathbf{S}\me^{t\mathbf{Z}} $ shows that $ \mu $ is invariant. Denote by $ (\cR_t)_{t\ge0} $ the corresponding semigroup on densities relative to $ \mu $, characterized by
    \begin{equation*}
        \int_{\R^{2d}} h(y)(\cR_t f)(y)\,\mu(\md y)
        =\int_{\R^{2d}} f(x)\left(\int_{\R^{2d}}h(y)p_t(x,\md y)\right)\mu(\md x),
        \qquad h\in L^\infty(\mu).
    \end{equation*}
    Comparing characteristic functions, using invariance of the mean of $ \rho $ and the preceding covariance identity, gives
    \begin{equation*}
        \cT_{*t}(\sigma_x)=\int_{\R^{2d}}\sigma_y\,p_t(x,\md y),
        \qquad \cT_{*t}\cK f=\cK\cR_t f.
    \end{equation*}

    To apply \cite[Theorem~5.1]{MetafunePallaraPriola2002}, we identify $ (\cR_t)_{t\ge0} $ as a classical Ornstein--Uhlenbeck semigroup. Conditioning the stationary joint Gaussian law associated with $ p_t $ gives transition mean $ \me^{t\mathbf{B}}x $ and covariance $ \epsilon(\mathbf{S}-\me^{t\mathbf{B}}\mathbf{S}\me^{t\mathbf{B}^T}) $, where
    \begin{equation*}
        \mathbf{B}:=\mathbf{S}\mathbf{Z}\mathbf{S}^{-1}
        =-\mathbf{Z}^T-\mathbf{C}\mathbf{S}^{-1}.
    \end{equation*}
    Here we used $ \mathbf{Z}^T\mathbf{S}+\mathbf{S}\mathbf{Z}=-\mathbf{C} $. Its generator $ \cA $ on $ L^1(\mu) $ therefore acts on $ \cC_c^\infty(\R^{2d}) $ by
    \begin{equation*}
        \cA f=\frac{\epsilon}{2}\Tr(\mathbf{C}D^2f)+(\mathbf{B}x)\cdot\nabla f.
    \end{equation*}
    The matrix $ \mathbf{B} $ is stable and satisfies $ \mathbf{B}\mathbf{S}+\mathbf{S}\mathbf{B}^T=-\mathbf{C} $, so
    \begin{equation*}
        \int_0^\infty\me^{s\mathbf{B}}(\epsilon\mathbf{C})\me^{s\mathbf{B}^T}\,\md s
        =\epsilon\mathbf{S}>0.
    \end{equation*}
    The same integral over every interval $ [0,t] $, $ t>0 $, is positive definite: if its quadratic form vanishes on a vector, the analytic function $ s\mapsto\mathbf{C}^{1/2}\me^{s\mathbf{B}^T}x $ vanishes identically, contradicting the preceding equality unless $ x=0 $. Thus the hypotheses of the cited theorem hold, including when $ \mathbf{C} $ is singular.

    Consequently, for each $ \eta\in\C $ with $ \Re\eta<0 $, there is a nonzero $ f\in\dom\cA $ satisfying $ \cR_t f=\me^{\eta t}f $. The intertwining identity and injectivity of $ \cK $ give a nonzero trace-class operator $ \cK f $ with $ \cT_{*t}(\cK f)=\me^{\eta t}\cK f $. Taking the difference quotient in trace norm yields $ \cK f\in\dom\cL_* $ and $ \cL_*(\cK f)=\eta\cK f $.

    \section{Proofs for strictly unstable drift} \label{appendix-unstable-proofs}

    \subsection{The expanding subspace and dual eigenoperators}

    Assume $ s(\mathbf{Z})>0 $. Choose $ \lambda\in\sigma(\mathbf{Z}) $ with $ \Re\lambda>0 $ and $ v\in\C^{2d}\setminus\{0\} $ satisfying $ \mathbf{Z}v=\lambda v $. Irreducibility gives
    \begin{equation*}
        \mathbf{S}_t = \int_0^t \me^{s \mathbf{Z}^T} \mathbf{C} \me^{s \mathbf{Z}} \md s > 0, \quad \forall t > 0.
    \end{equation*}

    If $ \lambda $ is real, choose $ v $ real and set $ k=1 $. Otherwise set $ k=2 $. Define the real $ 2d\times k $ matrix
    \begin{equation*}
        M_\lambda := (v), \quad k = 1; \quad M_\lambda := \left( \Re v \, \Im v \right) \in M_{2d \times k} (\R), \quad k = 2.
    \end{equation*}
    In either case,
    \begin{equation*}
        \rank M_\lambda = k,
    \end{equation*}
    and 
    \begin{equation} \label{eq-reduce-boldface-Z-by-M-lambda}
        \mathbf{Z}M_\lambda = M_\lambda Z_\lambda,
    \end{equation}
    where
    \begin{equation*}
        Z_\lambda := (\Re \lambda), \quad k = 1; \quad Z_\lambda := \begin{pmatrix}
                \Re \lambda & \Im \lambda \\ 
                -\Im \lambda & \Re \lambda
            \end{pmatrix}, \quad k = 2.
    \end{equation*}
    Consequently,
    \begin{equation} \label{eq-absolute-value-exp-Z-lambda-x}
        \abs{\me^{- t Z_\lambda} x} = \me^{-\Re(\lambda) t} \abs{x}, \quad \forall t \ge 0, \quad \forall x \in \R^k.
    \end{equation}
    Define the restricted diffusion matrix
    \begin{equation} \label{eq-definition-C-lambda}
        C_\lambda := M_\lambda^T \mathbf{C} M_\lambda \ge 0.
    \end{equation}
    All eigenvalues of $ -Z_\lambda $ have negative real part. Hence
    \begin{equation} \label{eq-definition-S-lambda}
        S_\lambda := \int_0^\infty \me^{- s Z_\lambda^T} C_\lambda \me^{-s Z_\lambda} \md s
    \end{equation}
    is well-defined, and Lemma \ref{lemma-strict-positivity-S-lambda} shows that it is positive definite.

    For $ f\in L^1(\R^k) $, define $ X_f\in\cB(\mathsf{h}) $ by
    \begin{equation*}
        X_f \xi := \int_{\R^k} f(x) W( \iota^{-1} (M_\lambda x) ) \xi \, \md x, \quad \xi  \in \mathsf{h}.
    \end{equation*}
    The integrand is strongly measurable by strong continuity of the Weyl operators, and its norm is integrable. In particular,
    \begin{align}
        \norm{X_f} &= \sup_{\norm{\xi}_\mathsf{h} = 1} \norm{X_f \xi} = \sup_{\norm{\xi}_\mathsf{h} = 1} \norm{ \int_{\R^k} f(x) W( \iota^{-1} (M_\lambda x) ) \xi \md x } \nonumber \\
        &\le \sup_{\norm{\xi}_\mathsf{h} = 1} \int_{\R^k} \abs{f(x)} \norm{W(\iota^{-1} (M_\lambda x)) \xi} \md x = \norm{f}_{L^1(\R^{k}; \md x)}. \label{eq-bound-of-X-f-by-L1-norm}
    \end{align}

    To incorporate the displacement, define
    \begin{equation} \label{eq-definition-zeta-lambda}
        \zeta_\lambda := (Z_\lambda^T)^{-1} M_\lambda^T \boldsymbol{\zeta} \in \R^k.
    \end{equation}
    The inverse exists because $ \Re\lambda>0 $. For $ t\ge0 $, set
    \begin{equation} \label{eq-definition-S-lambda-t}
        S_{\lambda, t} := \int_0^t \me^{-s Z_\lambda^T} C_\lambda \me^{-s Z_\lambda} \md s.
    \end{equation}
    Splitting the integral in \eqref{eq-definition-S-lambda} at time $ t $ gives
    \begin{equation} \label{eq-S-lambda-t-and-S-lambda}
        S_{\lambda, t} + \me^{- t Z_\lambda^T} S_{\lambda} \me^{- t Z_\lambda} = S_\lambda.
    \end{equation}

    Lemma \ref{lemma-T-X-f-and-G-f} gives the intertwining identity
    \begin{equation} \label{eq-action-of-QMS-on-X-f}
        \cT_t (X_f) = X_{G_t (f)},
    \end{equation}
    where 
    \begin{equation} \label{eq-definition-auxiliary-semigroup-G-t}
        (G_t f) (u) = \me^{- k \Re (\lambda) t } \exp{ -\frac{1}{2} u^T S_{\lambda, t} u + \mi \zeta^T_\lambda (\1_k - \me^{- t Z_\lambda}) u } \, f \left( \me^{- t Z_\lambda} u \right), \quad u \in \R^k.
    \end{equation}

    \begin{proof}[Proof of Proposition \ref{proposition-eigenvector-Gaussian-generator}]
        Fix $ \eta\in\C $ with $ \Re\eta<0 $. For $ u \neq 0 $, define
        \begin{equation} \label{eq-f-eta-u}
            f_\eta (u) := \exp{ \mi \zeta_\lambda^T u - \frac{1}{2} u^T S_\lambda u} \abs{u}^{- \eta / \Re(\lambda) - k}.
        \end{equation}
        By Lemma \ref{lemma-f-eta-L1-function}, $ f_\eta\in L^1(\R^k) $. Thus $ X_\eta:=X_{f_\eta} $ is bounded by \eqref{eq-bound-of-X-f-by-L1-norm}, and it is nonzero by Lemma \ref{lemma-X-trivial-kernel}. Using \eqref{eq-absolute-value-exp-Z-lambda-x} and \eqref{eq-S-lambda-t-and-S-lambda}, we compute
        \begin{align*}
            G_t (f_\eta) (u) &= \me^{-k \Re(\lambda) t} \exp{ - \frac{1}{2} u^T S_{\lambda, t} u + \mi \zeta_\lambda^T (\1_k - \me^{- t Z_\lambda}) u } \\ 
            &\quad \cdot \exp{\mi \zeta_\lambda^T \me^{- t Z_\lambda} u - \frac{1}{2} u^T \me^{- t Z_\lambda^T} S_\lambda \me^{- t Z_\lambda} u } \abs{ \me^{- t Z_\lambda} u }^{- \eta / \Re(\lambda) - k} \\
            &= \me^{-k \Re(\lambda) t} \exp{-\frac{1}{2} u^T S_\lambda u + \mi \zeta_\lambda^T u} \left( \me^{- \Re(\lambda) t} \abs{u} \right)^{- \eta / \Re(\lambda) - k} \\ 
            &= \me^{\eta t} \exp{-\frac{1}{2} u^T S_\lambda u + \mi \zeta_\lambda^T u} \abs{u}^{-\eta/\Re(\lambda) - k} = \me^{\eta t} f_\eta (u).
        \end{align*}
        Equation \eqref{eq-action-of-QMS-on-X-f} now gives 
        \begin{equation*}
            \cT_t (X_\eta) = X_{G_t (f_\eta)} = X_{\me^{\eta t} f_\eta} = \me^{\eta t} X_{\eta}, \quad \forall t \ge 0,
        \end{equation*}
        which proves \eqref{eq-eigenvector-Gaussian-generator}.
    \end{proof}

    \subsection{Spectrum and absence of trace-class eigenvectors}

    We now pass from the dual eigenoperators to the spectrum of the predual generator.

    \begin{proof}[Proof of Theorem \ref{theorem-whole-spectrum-strictly-unstable-drift}]
        Let $ \eta \in \C $ with $ \Re \eta < 0 $. By Proposition \ref{proposition-eigenvector-Gaussian-generator} we have a nonzero bounded operator $ X_\eta $ such that \eqref{eq-eigenvector-Gaussian-generator} holds. Therefore, for every $ Y \in \dom \cL_* $ we have
        \begin{equation*}
            \Tr (\cL_*(Y) X_\eta) = \lim_{t \rightarrow 0^+} \Tr( \frac{\cT_{*t} (Y) - Y}{t} X_\eta ) = \lim_{t \rightarrow 0^+} \Tr ( Y \frac{\cT_t (X_\eta) - X_\eta}{t} ) = \eta \Tr(Y X_\eta),
        \end{equation*}
        which implies that 
        \begin{equation*}
            \Tr( ((\eta - \cL_* )(Y)) X_\eta ) = 0.
        \end{equation*}
        Since $ X_\eta \neq 0 $, we conclude that 
        \begin{equation*}
            \ran (\eta - \cL_*) \subset \{ Y \in \cB_1 (\mathsf{h}) : \Tr (Y X_\eta) = 0 \}
        \end{equation*}
        and the set on the right is a proper closed hyperplane. Thus $ \eta - \cL_* $ does not have dense range, and
        \begin{equation} \label{eq-expanding-left-belongs-to-spectrum}
            \{ \eta \in \C : \Re \eta < 0 \} \subset \sigma (\cL_*).
        \end{equation}
        On the other hand, $ (\cT_{*t})_{t \ge 0} $ is a strongly continuous contraction semigroup on $ \cB_1 (\mathsf{h}) $. Hence, by the Hille--Yosida theorem, 
        \begin{equation} \label{eq-expanding-spectrum-belongs-to-left}
            \sigma (\cL_*) \subset \{ \eta \in \C : \Re \eta \le 0 \}.
        \end{equation}
        Closedness of $ \sigma(\cL_*) $, together with \eqref{eq-expanding-left-belongs-to-spectrum} and \eqref{eq-expanding-spectrum-belongs-to-left}, proves \eqref{eq-spectrum-predual-generator-expanding-case}.
    \end{proof}

    \begin{proof}[Proof of Theorem \ref{theorem-unstable-spectral-types}]
        By Theorem \ref{theorem-whole-spectrum-strictly-unstable-drift}, the spectrum is the closed left half-plane. Its boundary $ \mi\R $ lies in the approximate point spectrum by \cite[Chapter IV, Proposition 1.10]{EngelNagel2000}.

        Suppose that $ \cL_*(Y)=\eta Y $ for some nonzero $ Y\in\dom\cL_* $. Then $ \cT_{*t}(Y)=\me^{\eta t}Y $. Taking characteristic functions yields
        \begin{equation*}
            \chi_Y (z) = \me^{- \eta t} \phi_t (z) \chi_Y(\me^{t Z} z), 
        \end{equation*} 
        and, for every $ \omega\in\C^d $,
        \begin{equation*}
            \abs{ \chi_Y (\omega) } = \abs{\Tr (Y W(\omega))} \le \norm{Y}_1 \norm{W(\omega)} \le \norm{Y}_1,
        \end{equation*}
        so we obtain
        \begin{equation} \label{eq-upper-bound-abs-chi-Y}
            \abs{\chi_Y (z)} \le \exp{ - \Re (\eta) t - \frac{1}{2} \mathbf{z}^T \mathbf{S}_t \mathbf{z}} \norm{Y}_1, \quad \forall z \in \C^d.
        \end{equation}
        Choose $ \lambda\in\sigma(\mathbf{Z}) $ with $ \Re\lambda>0 $ and a nonzero eigenvector $ w\in\C^{2d} $ of $ \mathbf{Z}^T $ such that
        \begin{equation*}
            \mathbf{Z}^T w = \lambda w.
        \end{equation*}
        Consider the subspace
        \begin{equation*}
            E := \{ \mathbf{z} \in \R^{2d} : w^T \mathbf{z} = 0 \}.
        \end{equation*}
        When $ \lambda $ is real, choose $ w $ real; then $ \dim E=2d-1 $. When $ \lambda $ is nonreal, $ \Re w $ and $ \Im w $ are linearly independent and $ \dim E=2d-2 $. Thus $ E $ is a proper real subspace and $ \R^{2d}\setminus E $ is dense. For $ \mathbf{z}\notin E $ and $ t\ge0 $,
        \begin{equation*}
            \abs{w^* \me^{t \mathbf{Z}} \mathbf{z}} = \abs{ (\me^{t \mathbf{Z}^T} w)^* \mathbf{z} } = \abs{ \me^{\overline{\lambda} t} w^* \mathbf{z} } = \me^{\Re(\lambda) t} \abs{w^* \mathbf{z}}.
        \end{equation*}
        On the other hand, 
        \begin{equation*}
            \abs{w^* \me^{t \mathbf{Z}} \mathbf{z}} \le \abs{w} \abs{\me^{t \mathbf{Z}} \mathbf{z}}.
        \end{equation*}
        Setting $ c_{w,z}:=\abs{w^*\mathbf{z}}/\abs{w}>0 $, we obtain
        \begin{equation*}
            \abs{\me^{t \mathbf{Z}} \mathbf{z}} \ge \frac{\abs{w^* \mathbf{z}}}{\abs{w}} \me^{\Re(\lambda) t} =  c_{w, z} \me^{\Re (\lambda) t}.
        \end{equation*}
        By \eqref{eq-positive-diffusion-gramian}, $ \mathbf{S}_1>0 $. Let $ m_1 $ be its smallest eigenvalue. Retaining the last unit interval in the integral defining $ \mathbf{S}_t $, for $ t\ge1 $ we obtain
        \begin{align*}
            \mathbf{z}^T \mathbf{S}_t \mathbf{z} &\ge \int_{t - 1}^t ( \me^{s \mathbf{Z}} \mathbf{z} )^T \mathbf{C} ( \me^{s \mathbf{Z}} \mathbf{z} ) \md s = ( \me^{(t-1) \mathbf{Z}} \mathbf{z} )^T \mathbf{S}_1 (\me^{(t-1) \mathbf{Z} } \mathbf{z} ) \\
            &\ge m_1 \abs{\me^{(t-1) \mathbf{Z} } \mathbf{z}}^2 \ge c_{w, z}^2 m_1 \me^{2 \Re(\lambda) (t-1)}. 
        \end{align*}
        Substituting this bound into \eqref{eq-upper-bound-abs-chi-Y} gives
        \begin{equation*}
            \abs{\chi_Y (z)} \le \norm{Y}_1 \exp{ - \Re(\eta) t - \frac{1}{2} c_{w, z}^2 m_1 \me^{2 \Re (\lambda) (t-1)} }.
        \end{equation*}
        The right-hand side tends to zero as $ t\to\infty $, because $ \Re\lambda>0 $. Hence
        \begin{equation*}
            \chi_Y (\mathbf{z}) = 0, \quad \mathbf{z} \in \R^{2 d} \setminus E.  
        \end{equation*}
        By continuity and density, $ \chi_Y $ vanishes everywhere, so $ Y=0 $, a contradiction. This proves \eqref{eq-expanding-case-point-spectrum}.

        For $ \Re\eta<0 $, the proof of Theorem \ref{theorem-whole-spectrum-strictly-unstable-drift} shows that $ \eta-\cL_* $ has nondense range. Since it is injective,
        \begin{equation} \label{eq-left-open-half-plane-residual-spectrum}
            \{ \eta: \Re \eta  < 0 \} \subset \sigma_{\mathrm{r}} (\cL_*). 
        \end{equation}
        For $ \eta=0 $, trace preservation gives $ \Tr(\cL_*(Y))=0 $ for every $ Y\in\dom\cL_* $. Consequently,
        \begin{equation*}
            \ran (- \cL_*) \subset \{ Y \in \cB_1 (\mathsf{h}) : \Tr Y = 0 \}.
        \end{equation*}
        This is a proper closed hyperplane. Since $ \cL_* $ is injective, $ 0\in\sigma_{\mathrm{r}}(\cL_*) $. Together with \eqref{eq-left-open-half-plane-residual-spectrum}, this proves \eqref{eq-expanding-case-residual-spectrum}.
    \end{proof}

    \subsection{Supporting estimates and intertwining identities}

    We justify the positivity, injectivity, and continuity assertions used in the construction above.

    \begin{lemma} \label{lemma-strict-positivity-S-lambda}
    The matrix $ S_\lambda $ defined in \eqref{eq-definition-S-lambda} is positive definite.
    \end{lemma}
    \begin{proof}
        Irreducibility implies
        \begin{equation*}
            \mathbf{S}_1 = \int_0^1 \me^{s \mathbf{Z}^T} \mathbf{C} \me^{s \mathbf{Z}} \md s > 0.
        \end{equation*}
        Equation \eqref{eq-reduce-boldface-Z-by-M-lambda} gives
        \begin{equation} \label{eq-Z-M-lambda-exponential-form}
            \me^{t \mathbf{Z}} M_\lambda = M_\lambda \me^{t Z_\lambda}, \quad \forall t \in \R.
        \end{equation}
        Therefore, since $ M_\lambda $ has full column rank,
        \begin{equation*}
            S_{\lambda, 1} = \int_0^1 \me^{-s Z_\lambda^T} C_\lambda \me^{-s Z_\lambda} \md s = M_\lambda^T \left( \int_0^1 \me^{-s \mathbf{Z}^T} \mathbf{C} \me^{-s \mathbf{Z}} \md s \right) M_\lambda = M_\lambda^T \me^{-\mathbf{Z}^T} \mathbf{S}_1 \me^{- \mathbf{Z}} M_\lambda > 0.
        \end{equation*}
        Since $ S_\lambda \ge S_{\lambda, 1} $, it follows that $ S_{\lambda} > 0 $.
    \end{proof}

    \begin{lemma} \label{lemma-X-trivial-kernel}
        For $ f\in L^1(\R^k) $, $ X_f=0 $ implies $ f=0 $ almost everywhere.
    \end{lemma}
    \begin{proof}
        Recall that $ e(0) $ is the vacuum vector. If $ X_f=0 $, evaluation in the coherent vector $ W(y)e(0) $ gives
        \begin{align}
            0 &= \left\langle W(y)e(0), X_f W(y)e(0) \right\rangle \nonumber \\ 
            &= \int_{\R^k} f(x) \exp{-\frac{1}{2} \abs{ \iota^{-1}(M_\lambda x) }^2 + 2 \mi \Im \left\langle y, \iota^{-1} (M_\lambda x) \right\rangle } \md x \nonumber \\ 
            &= \int_{\R^k} f(x) \exp{-\frac{1}{2} x^T M_\lambda^T M_\lambda x} \me^{- \mi (2 M_\lambda^T \mathbf{J} \mathbf{y})^T x } \md x. \label{eq-fourier-transform-contradiction}
        \end{align}
        Define
        \begin{equation*}
            h (x) := f(x) \exp{-\frac{1}{2} x^T M_\lambda^T M_\lambda x}.
        \end{equation*}
        Equation \eqref{eq-fourier-transform-contradiction} gives
        \begin{equation*}
            \widehat{h} ( 2 M_\lambda^T \mathbf{J} \mathbf{y} ) = 0.  
        \end{equation*}
        Since $ \mathbf{J} $ is invertible and $ M_\lambda $ has full column rank, $ 2M_\lambda^T\mathbf{J}\mathbf{y} $ ranges over $ \R^k $. Uniqueness of the Fourier transform gives $ h=0 $ almost everywhere. The Gaussian factor is strictly positive, so $ f=0 $ almost everywhere.
    \end{proof}

    \begin{lemma} \label{lemma-T-X-f-and-G-f}
        We have $ \cT_t (X_f) = X_{G_t (f)} $ as in \eqref{eq-action-of-QMS-on-X-f}.
    \end{lemma}
    \begin{proof}
        We first show that, for every $ Y \in \cB_1 (\mathsf{h}) $,
        \begin{equation} \label{eq-trace-of-Y-X-f}
            \Tr ( Y X_f ) = \int_{\R^k} f(x) \Tr (Y W(\iota^{-1} (M_\lambda x))) \md x.
        \end{equation}
        To justify this identity, write a singular-value decomposition of $ Y $ and let
        \begin{equation*}
            Y_n := \sum_{j = 1}^n s_j \ketbra{e_j}{g_j},
        \end{equation*}
        where $ (e_j) $ and $ (g_j) $ are orthonormal families and $ (s_j) $ are the singular values. Then $ \norm{Y_n-Y}_1\to0 $. By the definition of $ X_f $,
        \begin{align}
            \Tr (Y_n X_f) &= \Tr ( \sum_{j=1}^n s_j \ketbra{e_j}{g_j} X_f ) \nonumber \\
            &= \sum_{j=1}^n \int_{\R^k} s_j \langle g_j, f(x) W( \iota^{-1} (M_\lambda x) ) e_j \rangle \md x \nonumber \\ 
            &= \int_{\R^k} f(x) \Tr( Y_n  W(\iota^{-1} (M_\lambda x) ) ) \md x. \label{eq-approximants-Tr-Yn-X-f}
        \end{align}
        Note that 
        \begin{equation*}
            \abs{ \Tr(Y X_f) - \Tr (Y_n X_f) } \le \norm{Y - Y_n}_1 \norm{X_f} \le \norm{Y - Y_n}_1 \norm{f}_{L^1(\R^k)} \rightarrow 0,
        \end{equation*}
        and 
        \begin{align*}
            &\quad \abs{ \int_{\R^k} f(x) \Tr( Y W(\iota^{-1} (M_\lambda x) ) ) \md x  -  \int_{\R^k} f(x) \Tr( Y_n W(\iota^{-1} (M_\lambda x) ) ) \md x } \\
            &=  \abs{ \int_{\R^k} f(x) \Tr( (Y-Y_n) W(\iota^{-1} (M_\lambda x) ) ) \md x } \le \int_{\R^k} \abs{f(x)} \norm{Y - Y_n}_1 \norm{W( \iota^{-1} (M_\lambda x) )} \md x \\
            &\le \norm{f}_{L^1(\R^k; \md x)} \norm{Y - Y_n}_1 \rightarrow 0,
        \end{align*}
        as $ n\to\infty $. Taking the limit in \eqref{eq-approximants-Tr-Yn-X-f} proves \eqref{eq-trace-of-Y-X-f}. We can therefore compute
        \begin{align}
            \Tr ( Y \cT_t (X_f) ) &= \Tr(\cT_{*t} (Y) X_f) \nonumber \\
            &= \int_{\R^k} f(x) \Tr( \cT_{*t} (Y) W(\iota^{-1} (M_\lambda x) ) ) \md x \nonumber \\
            &= \int_{\R^k} f(x) \Tr( Y \cT_{t}( W(\iota^{-1} (M_\lambda x) ) ) ) \md x \nonumber \\
            &= \Tr ( Y \int_{\R^k} f(x) \cdot \phi_t ( \iota^{-1} (M_\lambda x) ) \cdot W( \me^{t Z} \iota^{-1} (M_\lambda x) ) \, \md x ), \label{eq-action-semigroup-on-X-f}
        \end{align}
        where $ \phi_t $ is defined in \eqref{eq-definition-phi-t}. Make the change of variables
        \begin{equation*}
            u := \me^{t Z_\lambda} x,
        \end{equation*}
        so that $ \md x = \me^{- k (\Re \lambda) t} \md u $.

        Using \eqref{eq-definition-C-lambda}, \eqref{eq-definition-zeta-lambda}, \eqref{eq-definition-S-lambda-t}, and \eqref{eq-Z-M-lambda-exponential-form}, we obtain
        \begin{align}
            \phi_t ( \iota^{-1} (M_\lambda x) ) &= \exp{\begin{aligned}
                &-\frac{1}{2} \int_0^t \Re \left\langle \me^{s Z} \iota^{-1} (M_\lambda x), C \me^{s Z} \iota^{-1} (M_\lambda x) \right\rangle \md s \\
                &+\mi \int_0^t \Re \left\langle \zeta, \me^{s Z} \iota^{-1} (M_\lambda x) \right\rangle \md s
            \end{aligned}} \nonumber \\
            &= \exp{ -\frac{1}{2} \int_0^t x^T M_\lambda^T \me^{s \mathbf{Z}^T} \mathbf{C} \me^{s \mathbf{Z}} M_\lambda x \, \md s + \mi \int_0^t \boldsymbol{\zeta}^T \me^{s \mathbf{Z}} M_\lambda x \, \md s } \nonumber \\ 
            &= \exp{ -\frac{1}{2} u^T \me^{- t Z_\lambda^T} \int_0^t \me^{s Z_\lambda^T} M_\lambda^T \mathbf{C} M_\lambda \me^{s Z_\lambda} \md s \, \me^{-t Z_\lambda} u + \mi \boldsymbol{\zeta}^T \int_0^t M_\lambda \me^{s Z_\lambda} \md s \,\me^{-t Z_\lambda} u } \nonumber \\
            &= \exp{ -\frac{1}{2} u^T \left( \int_0^t \me^{-r Z_\lambda^T} C_\lambda \me^{- r Z_\lambda} \md r \right) u + \mi \boldsymbol{\zeta}^T M_\lambda \left( \int_0^t \me^{-r Z_\lambda} \md r \right) u } \nonumber \\
            &= \exp{ -\frac{1}{2} u^T S_{\lambda, t} u + \mi \zeta_\lambda^T (\1_k - \me^{-t Z_\lambda} ) u }, \label{eq-phi-t-inverse-M-lambda-x-by-u}
        \end{align}
        and 
        \begin{equation} \label{eq-weyl-x-by-u}
            W(\me^{t Z} \iota^{-1} (M_\lambda x)) = W( \iota^{-1} (\me^{t \mathbf{Z}} M_\lambda x ) ) = W( \iota^{-1} (M_\lambda \me^{t Z_\lambda} \me^{- t Z_\lambda} u ) ) = W(\iota^{-1} ( M_\lambda u )).
        \end{equation}

        Substituting \eqref{eq-phi-t-inverse-M-lambda-x-by-u} and \eqref{eq-weyl-x-by-u} into \eqref{eq-action-semigroup-on-X-f} gives
        \begin{equation*}
            \Tr (Y \cT_t (X_f) ) = \Tr( Y X_{G_t (f)} ), \quad \forall Y \in \cB_1 (\mathsf{h}),
        \end{equation*}
        which proves \eqref{eq-action-of-QMS-on-X-f}.
    \end{proof}

    \begin{lemma}
        The operators $ (G_t)_{t\ge0} $ defined in \eqref{eq-definition-auxiliary-semigroup-G-t} form a strongly continuous contraction semigroup on $ L^1(\R^k) $.
    \end{lemma}
    \begin{proof}
        Clearly $ G_0=\1 $. For $ s,t>0 $, \eqref{eq-definition-S-lambda-t} gives
        \begin{align*}
            &\quad (G_t G_s(f) ) (u) \\
            &= \me^{- k \Re (\lambda) t } \exp{ -\frac{1}{2} u^T S_{\lambda, t} u + \mi \zeta^T_\lambda (\1_k - \me^{- t Z_\lambda}) u } (G_s f) ( \me^{- t Z_\lambda} u ) \\
            &=  \me^{- k \Re (\lambda) t } \exp{ -\frac{1}{2} u^T S_{\lambda, t} u + \mi \zeta^T_\lambda (\1_k - \me^{- t Z_\lambda}) u } \\ 
            &\quad \cdot \me^{- k \Re (\lambda) s } \exp{ -\frac{1}{2} u^T \me^{-t Z_\lambda^T} S_{\lambda, s} \me^{-t Z_\lambda} u + \mi \zeta^T_\lambda (\1_k - \me^{- s Z_\lambda}) \me^{-t Z_\lambda} u } f (\me^{-s Z_\lambda} \me^{-t Z_\lambda} u) \\ 
            &= \me^{-k \Re(\lambda) (t + s)} \exp{-\frac{1}{2} u^T S_{\lambda, t + s} u + \mi \zeta_\lambda^T (\1_k - \me^{-(t + s) Z_\lambda} ) u } f(\me^{-(s+t) Z_\lambda} u) \\
            &= G_{t + s} (f) (u).
        \end{align*}
        Moreover,
        \begin{equation*}
            \norm{G_t f}_{L^1(\R^k)} \le \me^{- k \Re(\lambda) t} \int_{\R^k} \abs{f (\me^{-t Z_\lambda} u )} \md u = \me^{-k \Re(\lambda) t} \me^{k \Re(\lambda) t} \norm{f}_{L^1 (\R^k)} = \norm{f}_{L^1 (\R^k)},
        \end{equation*}
        so the semigroup is contractive.

        To prove strong continuity, first take $ f\in\cC_c(\R^k) $. Then $ G_tf\to f $ pointwise as $ t\downarrow0 $. Let $ K:=\supp f $. For $ 0\le t\le1 $,
        \begin{equation*}
            \supp ( G_t f ) \subset \me^{t Z_\lambda} K.
        \end{equation*}
        The set $ K_1:=\{\me^{tZ_\lambda}x:0\le t\le1,\ x\in K\} $ is compact, and
        \begin{equation*}
            \abs{ (G_t f) (x) } \le \norm{f}_\infty \1_{K_1} (x),
        \end{equation*}
        where $ \1_{K_1} $ is its indicator function. Thus $ \abs{G_tf-f} $ has an integrable bound independent of $ t $. Dominated convergence gives
        \begin{equation} \label{eq-strong-continuity-Cc-functions}
            \lim_{t \rightarrow 0^+} \norm{G_t f - f}_{L^1(\R^k)} = 0, \quad f \in \cC_c (\R^k).
        \end{equation}
        For $ g\in L^1(\R^k) $, choose $ f\in\cC_c(\R^k) $ with $ \norm{g-f}_{L^1}<\epsilon $. Contractivity gives
        \begin{align*}
            \norm{G_t g - g}_{L^1(\R^k)} &\le \norm{G_t (g - f) }_{L^1(\R^k)} + \norm{ G_t f - f }_{L^1(\R^k)} + \norm{f - g}_{L^1(\R^k)} \\
            &< 2 \epsilon + \norm{G_t f - f}_{L^1(\R^k)},
        \end{align*}
        and strong continuity follows from \eqref{eq-strong-continuity-Cc-functions}.
    \end{proof}

    \begin{lemma} \label{lemma-f-eta-L1-function}
        The function $ f_\eta $, as defined in \eqref{eq-f-eta-u}, belongs to $ L^1 (\R^k) $.
    \end{lemma}
    \begin{proof}
        By definition,
        \begin{equation*}
            \abs{ f_\eta (x) } = \exp{ -\frac{1}{2} x^T S_\lambda x } \abs{x}^{-\Re(\eta)/\Re(\lambda) - k}.
        \end{equation*}
        Since $ \Re\lambda>0 $ and $ \Re\eta<0 $,
        \begin{equation} \label{eq-ratio-real-parts-eta-lambda}
            - \Re(\eta) / \Re(\lambda) > 0.
        \end{equation}
        Let $ m_\lambda>0 $ be the smallest eigenvalue of $ S_\lambda $. Then
        \begin{equation*}
            \abs{f_\eta (x)} \le \exp{- \frac{m_\lambda}{2} \abs{x}^2 } \abs{x}^{-k - \Re(\eta)/\Re(\lambda)}.
        \end{equation*}
        For $ k=2 $, polar coordinates give
        \begin{equation*}
            \int_{\R^k} \abs{f_\eta (x)} \md x \le 2 \pi \int_0^\infty \me^{- \frac{m_\lambda}{2} r^2} r^{-1 - \Re(\eta)/\Re(\lambda)} \md r.
        \end{equation*}
        The integral converges by \eqref{eq-ratio-real-parts-eta-lambda}. For $ k=1 $, the same estimate holds with $ 2 $ in place of $ 2\pi $, so the conclusion follows in that case as well.
    \end{proof}

    \section{Proofs for periodic critical drift} \label{appendix-periodic-proofs}

    \subsection{The semigroup at a return time}

    The first three lemmas apply at every return time, including arbitrary positive times when $ \mathbf{Z}=0 $.

    Fix $ r>0 $ with $ \me^{r\mathbf{Z}}=\1 $. By \eqref{eq-positive-diffusion-gramian}, $ \mathbf{S}_r>0 $. At this return time,
        \begin{equation} \label{eq-action-of-GQMS-at-time-tau}
            \cT_r (W(z)) = \phi_r (z) W(\me^{r Z} z) = \exp{-\frac{1}{2} \mathbf{z}^T \mathbf{S}_r \mathbf{z} + \mi \mathbf{b}_r^T \mathbf{z} } W(z).
        \end{equation}
        For $ Y\in\cB_1(\mathsf{h}) $, Lemma \ref{lemma-equivalent-representation-predual-semigroup-at-tau} yields the representation
        \begin{equation} \label{eq-predual-semigroup-at-tau-equivalent-form}
            \cT_{* r} (Y) = \int_{\R^{2d}} f_r ( u ) W(\iota^{-1} (u)) Y W( \iota^{-1} (u) )^* \md u,
        \end{equation}
        where
        \begin{equation} \label{eq-definition-f-tau}
            f_r (u) := \frac{2^d}{\pi^d \sqrt{\det \mathbf{S}_r}} \exp{-\frac{1}{2} (2 \mathbf{J} u + \mathbf{b}_r)^T \mathbf{S}_r^{-1} (2 \mathbf{J} u + \mathbf{b}_r) }, \quad u \in \R^{2d}.
        \end{equation}

    \begin{lemma} \label{lemma-equivalent-representation-predual-semigroup-at-tau}
        At every return time $ r>0 $, the predual has the representation \eqref{eq-predual-semigroup-at-tau-equivalent-form}, with Gaussian density \eqref{eq-definition-f-tau}.
    \end{lemma}
    \begin{proof}
        The Gaussian integral formula gives
        \begin{equation*}
            \int_{\R^{2d}} f_r (u) \me^{-2 \mi \mathbf{z}^T \mathbf{J} u} \md u = \exp{ -\frac{1}{2} \mathbf{z}^T \mathbf{S}_r \mathbf{z} + \mi \mathbf{b}_r^T \mathbf{z} } = \phi_r (z).
        \end{equation*}
        Equivalently, with the Fourier convention of Section \ref{section-spectral-notation},
        \begin{equation} \label{eq-f-tau-phi-tau-relation-2}
            \phi_r (z) = \widehat{f_r} ( - 2 \mathbf{J} \mathbf{z} ).
        \end{equation}
        Define 
        \begin{equation*}
            \Phi (Y) = \int_{\R^{2d}} f_r ( u ) W(\iota^{-1} (u)) Y W( \iota^{-1} (u) )^* \md u.
        \end{equation*}
        Taking characteristic functions and using the Weyl relations gives
        \begin{align*}
            \chi_{\Phi (Y)} (z) &= \Tr \left( \int_{\R^{2d}} f_r ( u ) W(\iota^{-1} (u)) Y W( \iota^{-1} (u) )^* \md u \, W(z) \right) \\
            &= \int_{\R^{2d}} f_r (u) \me^{- 2 \mi \mathbf{z}^T \mathbf{J} u} \, \md u \cdot \Tr(Y W(z)) = \phi_r(z) \chi_Y (z).
        \end{align*}
        On the other hand, \eqref{eq-action-of-GQMS-at-time-tau} gives
        \begin{equation*}
            \chi_{\cT_{*r} (Y)} (z) = \Tr ( Y \cT_{r} W(z) ) = \phi_r (z) \chi_Y (z).
        \end{equation*}
        Injectivity of the Fourier--Weyl transform implies $ \Phi(Y)=\cT_{*r}(Y) $, proving the representation.
    \end{proof}

    \begin{lemma} \label{lemma-spectrum-predual-semigroup-at-tau}
        We have
        \begin{equation*}
            \sigma (\cT_{*r}) = \overline{\phi_r ( \C^{d} )},
        \end{equation*}
        where $ \phi_r(\C^d):=\{\phi_r(z):z\in\C^d\} $.
    \end{lemma}
    \begin{proof}
        Let $ \lambda \notin \overline{\phi_r (\C^d)} $. Since $ \phi_r (z) \rightarrow 0 $ as $ \abs{z} \rightarrow +\infty $, we have $ 0 \in \overline{\phi_r (\C^d)} $. Thus $ \lambda \neq 0 $, and we may define 
        \begin{equation*}
            F_\lambda (x) := \frac{1}{\lambda - \phi_r (\iota^{-1}(\mathbf{J} x / 2))} - \frac{1}{\lambda}, \quad x \in \R^{2d}.
        \end{equation*}
        By \eqref{eq-f-tau-phi-tau-relation-2}, this can be written as
        \begin{equation} \label{eq-F-lambda-by-f-tau}
            F_\lambda (x) = \frac{1}{\lambda - \widehat{f_r} (x)} - \frac{1}{\lambda} = \frac{\widehat{f_r} (x)}{\lambda \left( \lambda - \widehat{f_r} (x) \right)}, \quad x \in \R^{2d}.
        \end{equation}
        The denominator is bounded away from zero, so $ F_\lambda $ is well-defined.
    
        We show that $ F_\lambda $ is a Schwartz function. By \eqref{eq-f-tau-phi-tau-relation-2}, $ \widehat{f_r} $ is a nondegenerate Gaussian multiplied by a linear phase, and hence belongs to $ \cS(\R^{2d}) $. Moreover,
        \begin{equation*}
            \delta := \inf_{x \in \R^{2d}} \abs{ \lambda - \widehat{f_r} (x) } > 0.
        \end{equation*} 
        Define the function
        \begin{equation*}
            E_\lambda (w) := \frac{w}{\lambda (\lambda - w)},
        \end{equation*}
        so that $ F_\lambda = E_\lambda \circ \widehat{f_r} $. For every integer $ m \ge 1 $, $ E_\lambda^{(m)} (w) = m! / (\lambda - w)^{m+1} $. Hence,
        \begin{equation*}
            \abs{  E_\lambda^{(m)} (\widehat{f_r} (x)) } =  \frac{m!}{\abs{\lambda - \widehat{f_r} (x)}^{m+1}} \le \frac{m!}{\delta^{m+1}}.
        \end{equation*}
        Thus all derivatives of positive order of $ E_\lambda $ are bounded on the range of $ \widehat{f_r} $.
    
        For the zeroth-order term, we have 
        \begin{equation*}
            \abs{F_\lambda (x)} = \frac{\abs{\widehat{f_r} (x)}}{\abs{\lambda} \abs{\lambda - \widehat{f_r} (x)} } \le \frac{\abs{\widehat{f_r} (x)}}{ \abs{\lambda} \delta },
        \end{equation*}
        so $ F_\lambda $ decays rapidly. For each nonzero multi-index $ \alpha $, the chain rule expresses $ \partial^\alpha F_\lambda $ as a finite sum of terms of the form
        \begin{equation*}
            E_\lambda^{(m)} ( \widehat{f_r} (x) ) \prod_{j=1}^m \partial^{\alpha_j} \widehat{f_r} (x),
        \end{equation*}
        where $ \alpha_j \neq 0 $ and $ \alpha_1 + \cdots + \alpha_m = \alpha $. The first factor is bounded, while every derivative of $ \widehat{f_r} $ belongs to the Schwartz space. Hence, for every $ M \ge 0 $, 
        \begin{equation*}
            \sup_{x \in \R^{2d}} (1 + \abs{x})^M \abs{ \partial^\alpha F_\lambda (x) } < + \infty.
        \end{equation*}
        Therefore $ F_\lambda\in\cS(\R^{2d}) $. Fourier inversion gives $ h_\lambda\in\cS(\R^{2d}) $ such that
        \begin{equation*}
            \widehat{h_\lambda} = F_\lambda.
        \end{equation*}

        Define the bounded operator $ R_\lambda $ on $ \cB_1 (\mathsf{h}) $ by 
        \begin{equation*}
            R_\lambda (Y) := \frac{1}{\lambda} Y + \int_{\R^{2d}} h_\lambda (x) W(\iota^{-1} x) Y W(\iota^{-1} (x))^* \, \md x.
        \end{equation*}
        Using Fubini's theorem, the Weyl relations, and the substitution $ v=u+x $, we obtain
        \begin{align*}
            &\quad (\lambda \1 - \cT_{*r}) R_\lambda (Y)  \\
            &= Y + \lambda \int_{\R^{2d}} h_\lambda (x) W(\iota^{-1} (x)) Y W(\iota^{-1} (x))^* \md x  \\ 
            &\quad - \frac{1}{\lambda} \int_{\R^{2d}} f_r (x) W(\iota^{-1} (x)) Y W(\iota^{-1} (x))^* \md x  \\ 
            &\quad - \int_{\R^{2d}} f_r (u) W(\iota^{-1} (u)) \left( \int_{\R^{2d}} h_\lambda(x) W(\iota^{-1} (x)) Y W(\iota^{-1} (x))^* \md x \right) W(\iota^{-1} (u))^* \md u  \\
            &= Y + \int_{\R^{2d}} \left( \lambda h_\lambda (x) - \frac{1}{\lambda} f_r (x) \right) W(\iota^{-1} (x)) Y W(\iota^{-1} (x))^* \md x  \\
            &\quad - \int_{\R^{2d}}\int_{\R^{2d}} f_r(u) h_\lambda(x) W(\iota^{-1}(u+x)) Y W( \iota^{-1} (u+x) )^*  \md u \, \md x  \\
            &= Y + \int_{\R^{2d}} \left( \lambda h_\lambda (x) - \frac{1}{\lambda} f_r (x) \right) W(\iota^{-1} (x)) Y W(\iota^{-1} (x))^* \md x  \\
            &\quad - \int_{\R^{2d}}\int_{\R^{2d}} f_r(v - x) h_\lambda(x) \md x \, W(\iota^{-1}(v)) Y W( \iota^{-1} (v) )^*  \md v  \\ 
            &= Y + \int_{\R^{2d}} \left( \lambda h_\lambda - \frac{1}{\lambda} f_r - f_r \ast h_\lambda  \right) (x) \,  W(\iota^{-1} (x)) Y W(\iota^{-1} (x))^* \md x.
        \end{align*}
        Taking Fourier transforms and using \eqref{eq-F-lambda-by-f-tau} gives
        \begin{equation*}
            \lambda \widehat{h_\lambda} - \frac{1}{\lambda} \widehat{f_r} - \widehat{f_r} \widehat{h_\lambda} = (\lambda - \widehat{f_r}) F_\lambda - \frac{1}{\lambda} \widehat{f_r} = (\lambda - \widehat{f_r}) \frac{\widehat{f_r}}{\lambda (\lambda - \widehat{f_r})} - \frac{1}{\lambda} \widehat{f_r} = 0,
        \end{equation*}
        and hence 
        \begin{equation*}
            (\lambda \1 - \cT_{*r} ) R_\lambda (Y) = Y, \quad \forall Y \in \cB_1 (\mathsf{h}).
        \end{equation*}
        The Weyl conjugations commute with one another, because the scalar phases cancel. Thus $ R_\lambda $ commutes with $ \cT_{*r} $ and is a two-sided bounded inverse of $ \lambda\1-\cT_{*r} $. This proves $ \sigma(\cT_{*r})\subset\overline{\phi_r(\C^d)} $.

        For the reverse inclusion, fix $ z\in\C^d $ and set $ \lambda=\phi_r(z) $. For every $ Y\in\cB_1(\mathsf{h}) $,
        \begin{equation*}
            \chi_{ (\lambda \1 - \cT_{*r}) (Y) } (z) = \Tr( ((\lambda \1 - \cT_{*r})(Y)) W(z) ) = \phi_r (z) \chi_Y(z) - \phi_r (z) \chi_Y(z) = 0.
        \end{equation*}
        The nonzero functional $ Y\mapsto\Tr(YW(z)) $ annihilates the range of $ \lambda\1-\cT_{*r} $. Hence $ \lambda\in\sigma(\cT_{*r}) $. Closedness of the spectrum proves the reverse inclusion.
    \end{proof}

    \begin{lemma} \label{lemma-closure-of-phi-tau}
        We have 
        \begin{equation*}
            \overline{\phi_r (\C^d)} = \phi_r (\C^d) \cup \{ 0 \}.
        \end{equation*}
    \end{lemma}
    \begin{proof}
        Since $ \abs{\phi_r(z)}\to0 $ as $ \abs{z}\to\infty $, the inclusion $ \phi_r(\C^d)\cup\{0\}\subset\overline{\phi_r(\C^d)} $ is immediate. Conversely, let $ \lambda\neq0 $ belong to this closure, and choose $ z_k $ with $ \phi_r(z_k)\to\lambda $. For sufficiently large $ k $,
        \begin{equation*}
            \exp{-\frac{1}{2} \mathbf{z}_k^T \mathbf{S}_r \mathbf{z}_k } = \abs{\phi_r (z_k)} \ge \frac{\abs{\lambda}}{2} > 0.
        \end{equation*}
        Positive definiteness of $ \mathbf{S}_r $ gives $ c_r>0 $ such that $ \mathbf{z}^T\mathbf{S}_r\mathbf{z}\ge c_r\abs{\mathbf{z}}^2 $. Consequently,
        \begin{equation*}
            c_r \abs{\mathbf{z}_k}^2 \le \mathbf{z}_k^T \mathbf{S}_r \mathbf{z}_k \le - 2 \log \frac{\abs{\lambda}}{2},
        \end{equation*}
        so that $ (z_k)_{k \in \N} $ is a bounded sequence. Since $ \C^d \cong \R^{2d} $ is finite-dimensional, there exists a subsequence $ (z_{k_j})_{j \in \N} $ such that $ z_{k_j} \rightarrow z $ for some $ z \in \C^d $. Since $ \phi_r $ is continuous, we have
        \begin{equation*}
            \lim_{j \rightarrow +\infty} \phi_r (z_{k_j}) = \phi_r (z) = \lambda.
        \end{equation*}
        Thus $ \lambda\in\phi_r(\C^d) $, as required.
    \end{proof}

    \subsection{The zero-drift generator}

    \begin{proof}[Proof of Theorem \ref{theorem-zero-drift-case-spectrum}]
        By \eqref{eq-evolution-Wz-zero-drift-case}, for each $ z \in \C^d $, the Weyl operator belongs to the domain of the $ \sigma $-weak generator and satisfies
        \begin{equation*}
            \cL (W(z)) = \psi(z) W(z).
        \end{equation*}
        Consequently, for any $ Y \in \dom \cL_* $ we have 
        \begin{equation*}
            \Tr( (\psi(z) - \cL_*) (Y) \, W(z) ) = 0,
        \end{equation*}
        so that the range of $ \psi (z) - \cL_* $ is not dense. Hence 
        \begin{equation*}
            \psi(\C^d) \subset \sigma (\cL_*).
        \end{equation*}

        Conversely, let $ \lambda\in\sigma(\cL_*) $. Spectral inclusion \cite[Chapter IV, Theorem 3.6]{EngelNagel2000} gives $ \me^{t\lambda}\in\sigma(\cT_{*t}) $ for every $ t>0 $. Lemmas \ref{lemma-spectrum-predual-semigroup-at-tau} and \ref{lemma-closure-of-phi-tau}, applied with return time $ r=t $, imply
        \begin{equation*}
            \sigma (\cT_{*t}) = \overline{ \{ \me^{t \psi (z)} : z \in \C^d \} } = \{ \me^{t \psi(z)} : z\in\C^d \} \cup \{ 0 \}.
        \end{equation*}
        Since $ \me^{t\lambda}\neq0 $, there are $ z_t\in\C^d $ and $ n_t\in\Z $ such that
        \begin{equation} \label{eq-relation-difference-lambda-psi}
            t( \lambda - \psi (z_t) ) = 2 \pi \mi n_t,
        \end{equation}
        which implies that 
        \begin{equation*}
            \Re \lambda = - \frac{1}{2} \mathbf{z}_t^T \mathbf{C} \mathbf{z}_t, \quad \Im \lambda = \frac{2 \pi n_t}{t} + \boldsymbol{\zeta}^T \mathbf{z}_t, \quad \forall t > 0.
        \end{equation*}
        Since $ \mathbf{C}>0 $, the vectors $ \mathbf{z}_t $ all lie in a fixed bounded ellipsoid. Thus there is $ M>0 $, independent of $ t $, such that
        \begin{equation} \label{eq-im-lambda-upperbound}
            \abs{\Im \lambda - \boldsymbol{\zeta}^T \mathbf{z}_t} \le \abs{\Im \lambda} + \abs{\boldsymbol{\zeta}^T \mathbf{z}_t} < M.
        \end{equation} 
        If $ n_t\neq0 $, then
        \begin{equation} \label{eq-im-lambda-lowerbound}
            \abs{\Im \lambda - \boldsymbol{\zeta}^T \mathbf{z}_t } = \abs{\frac{2 \pi n_t}{t}} \ge \frac{2 \pi}{t}.
        \end{equation}
        For sufficiently small $ t>0 $, \eqref{eq-im-lambda-lowerbound} contradicts \eqref{eq-im-lambda-upperbound}. Hence $ n_t=0 $ for such $ t $, and \eqref{eq-relation-difference-lambda-psi} gives $ \lambda=\psi(z_t)\in\psi(\C^d) $.
    \end{proof}

    \begin{proof}[Proof of Proposition \ref{proposition-zero-drift-shape}]
        When $ \zeta = 0 $, we have $ \psi (z) = - \mathbf{z}^T \mathbf{C} \mathbf{z} / 2 $, so that $ \sigma(\cL_*) = \psi(\C^d) = (-\infty, 0] $ follows directly.

        Suppose now that $ \zeta \neq 0 $, and let $ \lambda = \psi (z) $ for some $ z \in \C^d $. Then
        \begin{equation} \label{eq-relation-lambda-z-C-z}
            \Re \lambda = - \frac{1}{2} \mathbf{z}^T \mathbf{C} \mathbf{z}, \quad \Im \lambda = \boldsymbol{\zeta}^T \mathbf{z}.
        \end{equation}
        By the Cauchy--Schwarz inequality, 
        \begin{equation*}
            (\Im \lambda)^2 = \left( (\mathbf{C}^{-1/2} \boldsymbol{\zeta} )^T (\mathbf{C}^{1/2} \mathbf{z}) \right)^2 \le ( \boldsymbol{\zeta}^T \mathbf{C}^{-1} \boldsymbol{\zeta} ) ( \mathbf{z}^T \mathbf{C} \mathbf{z} ). 
        \end{equation*}
        Using \eqref{eq-relation-lambda-z-C-z}, we obtain 
        \begin{equation} \label{eq-parabola-equivalent-form}
            (\Im \lambda)^2 \le - 2 (\Re \lambda) (\boldsymbol{\zeta}^T \mathbf{C}^{-1} \boldsymbol{\zeta}).
        \end{equation}
        Therefore, we have 
        \begin{equation*}
            \psi(\C^d) \subset \left\{ \lambda \in \C : \Re \lambda \le - \frac{(\Im \lambda)^2}{2 \boldsymbol{\zeta}^T \mathbf{C}^{-1} \boldsymbol{\zeta} } \right\}.
        \end{equation*}
        For the reverse inclusion, fix $ \lambda\in\C $ satisfying \eqref{eq-parabola-equivalent-form} and set
        \begin{equation*}
            \mathbf{z}_0 = \frac{\Im \lambda}{ \boldsymbol{\zeta}^T \mathbf{C}^{-1} \boldsymbol{\zeta} } \mathbf{C}^{-1} \boldsymbol{\zeta}.
        \end{equation*}
        It follows that 
        \begin{equation*}
            \boldsymbol{\zeta}^T \mathbf{z}_0 = \Im \lambda, \quad \mathbf{z}_0^T \mathbf{C} \mathbf{z}_0 =  \frac{(\Im \lambda)^2}{\boldsymbol{\zeta}^T \mathbf{C}^{-1} \boldsymbol{\zeta}}.
        \end{equation*}
        Choose $ \mathbf{w}_0\in\R^{2d}\setminus\{0\} $ with $ \boldsymbol{\zeta}^T\mathbf{w}_0=0 $ and normalize it so that $ \mathbf{w}_0^T\mathbf{C}\mathbf{w}_0=1 $. This is possible because $ 2d\ge2 $ and $ \mathbf{C}>0 $. Set
        \begin{equation*}
            r := \sqrt{ - 2 \Re \lambda - \frac{(\Im \lambda)^2}{ \boldsymbol{\zeta}^T \mathbf{C}^{-1} \boldsymbol{\zeta} } }, \quad \mathbf{w} := \mathbf{z}_0 + r \mathbf{w}_0.
        \end{equation*}
        Since $ \mathbf{z}_0^T\mathbf{C}\mathbf{w}_0=0 $, direct calculation gives
        \begin{equation*}
            \Re \psi(\iota^{-1}(\mathbf{w})) = \Re \lambda, \quad \Im \psi(\iota^{-1}(\mathbf{w})) = \Im \lambda.
        \end{equation*}
        Thus $ \lambda=\psi(\iota^{-1}(\mathbf{w})) $, proving
        \begin{equation*}
            \left\{ \lambda \in \C : \Re \lambda \le - \frac{(\Im \lambda)^2}{2 \boldsymbol{\zeta}^T \mathbf{C}^{-1} \boldsymbol{\zeta} } \right\} \subset \psi(\C^d).
        \end{equation*}
 
        Therefore the spectrum is the closed region to the left of the parabola
        \begin{equation*}
            \Re \lambda = - \frac{(\Im \lambda)^2}{2 \boldsymbol{\zeta}^T \mathbf{C}^{-1} \boldsymbol{\zeta} }.
        \end{equation*}
    \end{proof}

    \begin{proof}[Proof of Theorem \ref{theorem-zero-drift-spectral-types}]
        Suppose that $ \cL_* (Y) = \lambda Y $ for some nonzero $ Y \in \dom\cL_* $. Equation \eqref{eq-evolution-Wz-zero-drift-case} gives
        \begin{equation*}
            \me^{t \psi (z)} \chi_Y (z) = \me^{t \lambda} \chi_Y (z), \quad \forall z \in \C^d, \quad \forall t \ge 0.
        \end{equation*}
        Hence $ \chi_Y (z) \neq 0 $ implies $ \psi(z) = \lambda $. Because $ \chi_Y $ is continuous and nonzero somewhere, it is nonzero on a nonempty open set. This would force $ \psi $ to be constant there. However, \eqref{eq-definition-psi-z} gives
        \begin{equation*}
            \Re \psi (z) = - \frac{1}{2} \mathbf{z}^T \mathbf{C} \mathbf{z}
        \end{equation*}
        with $ \mathbf{C} > 0 $. Therefore, $ \sigma_{\mathrm{p}} (\cL_*) = \varnothing $.

        The proof of Theorem \ref{theorem-zero-drift-case-spectrum} shows that $ \psi(z)-\cL_* $ has nondense range for every $ z $. Since it is injective, the entire spectrum is residual.
    \end{proof}

    \subsection{The nonzero periodic generator}

    \begin{proof}[Proof of Theorem \ref{theorem-spectrum-predual-generator-periodic-critical-drift}]
        Set $ r=\tau $. By Lemma \ref{lemma-spectrum-predual-semigroup-at-tau}, $ \sigma(\cT_{*\tau})=\overline{\phi_\tau(\C^d)} $. Spectral inclusion \cite[Chapter IV, Theorem 3.6]{EngelNagel2000} and Lemma \ref{lemma-closure-of-phi-tau} show that, if $ \lambda\in\sigma(\cL_*) $, then
        \begin{equation*}
            \me^{\tau \lambda} = \phi_\tau (z) = \exp{-\frac{1}{2} \mathbf{z}^T \mathbf{S}_\tau \mathbf{z} + \mi \mathbf{b}_\tau^T \mathbf{z} }
        \end{equation*}
        for some $ z \in \C^d $, which implies that 
        \begin{equation*}
            \lambda = -\frac{1}{2 \tau} \mathbf{z}^T \mathbf{S}_\tau \mathbf{z} + \frac{\mi}{\tau} ( \mathbf{b}_\tau^T \mathbf{z} + 2 \pi n )
        \end{equation*}
        for some $ n\in\Z $. This proves one inclusion.

        For the converse, fix $ z\in\C^d $ and $ n\in\Z $, and set
        \begin{equation*}
            \mu := - \frac{1}{2 \tau} \mathbf{z}^T \mathbf{S}_\tau \mathbf{z} + \frac{\mi}{\tau} ( \mathbf{b}_\tau^T \mathbf{z} + 2 \pi n ).
        \end{equation*}
        Then
        \begin{equation*}
            \me^{\mu \tau} = \exp{ -\frac{1}{2} \mathbf{z}^T \mathbf{S}_\tau \mathbf{z} + \mi \mathbf{b}_\tau^T \mathbf{z} + 2 \pi \mi n } = \phi_\tau (z),
        \end{equation*}
        so \eqref{eq-action-of-GQMS-at-time-tau} becomes
        \begin{equation*}
            \cT_{\tau} (W(z)) = \phi_\tau (z) W(z) = \me^{\mu \tau} W(z).
        \end{equation*}
        Define the bounded operator
        \begin{equation} \label{eq-definition-eigenvector-X-mu-z}
            X_{\tau, z, n} := \int_0^\tau \me^{-\mu s} \cT_s (W(z)) \md s.
        \end{equation}
        The integral is understood in the $ \sigma $-weak sense: for every $ Y\in\cB_1(\mathsf{h}) $,
        \begin{equation*}
            \Tr (Y X_{\tau, z, n}) = \int_0^\tau \me^{-\mu s} \Tr( Y \cT_s (W(z)) ) \md s.
        \end{equation*}
        It satisfies
        \begin{equation*}
            \norm{X_{\tau, z, n}} \le \int_0^\tau \me^{-s \Re \mu} \md s < \infty.
        \end{equation*}
        The integrand is $ \tau $-periodic, because
        \begin{equation} \label{eq-X-integrand-tau-periodic}
            \me^{-\mu (s + \tau)} \cT_{s + \tau} (W(z)) = \me^{-\mu (s + \tau)} \cT_s \cT_\tau (W(z)) = \me^{-\mu (s + \tau)} \me^{\mu \tau} \cT_s (W(z)) = \me^{-\mu s} \cT_s (W(z)). 
        \end{equation}
        Normality of $ \cT_t $, the substitution $ r=t+s $, and \eqref{eq-X-integrand-tau-periodic} give
        \begin{equation*}
            \cT_t (X_{\tau, z, n}) = \int_0^\tau \me^{-\mu s} \cT_{t + s} (W(z)) \md s = \me^{\mu t} \int_t^{\tau + t} \me^{-\mu r} \cT_r (W(z)) \md r = \me^{\mu t} X_{\tau, z, n}.
        \end{equation*}
        Thus $ \cL ( X_{\tau, z, n} ) = \mu X_{\tau, z, n} $.
    
        The operator $ X_{\tau,z,n} $ may vanish for some $ z $. Lemmas \ref{lemma-density-D-tau} and \ref{lemma-nonzero-X-mu-z} show that it is nonzero on the dense set $ \cD_\tau $ defined in \eqref{eq-definition-D-tau-set}. Choose $ z_k\in\cD_\tau $ with $ z_k\to z $, keep $ n $ fixed, and set
        \begin{equation*}
            \mu_k = - \frac{1}{2 \tau} \mathbf{z}_k^T \mathbf{S}_\tau \mathbf{z}_k + \frac{\mi}{\tau} ( \mathbf{b}_\tau^T \mathbf{z}_k + 2 \pi n ).
        \end{equation*}
        The corresponding $ X_{\tau,z_k,n} $ are nonzero dual eigenoperators. As in the proof of Theorem \ref{theorem-whole-spectrum-strictly-unstable-drift}, the range of $ \mu_k-\cL_* $ is annihilated by a nonzero bounded functional. Consequently,
        \begin{equation*}
            \mu_k \in \sigma(\cL_*), \quad \mu_k \rightarrow - \frac{1}{2 \tau} \mathbf{z}^T \mathbf{S}_\tau \mathbf{z} + \frac{\mi}{\tau} ( \mathbf{b}_\tau^T \mathbf{z} + 2 \pi n ) = \mu.
        \end{equation*}
        Closedness of $ \sigma(\cL_*) $ gives $ \mu\in\sigma(\cL_*) $.

        To exclude eigenvalues, suppose that $ Y\neq0 $ and
        $ \cL_*(Y)=\lambda Y $. Then
        $ \cT_{*\tau}(Y)=\me^{\tau\lambda}Y $. Since
        $ \me^{\tau\mathbf{Z}}=\1 $, the characteristic-function formula gives
        \begin{equation*}
            \me^{\tau\lambda}\chi_Y(z)
            =\exp{-\frac12\mathbf{z}^T\mathbf{S}_\tau\mathbf{z}
                       +\mi\mathbf{b}_\tau^T\mathbf{z}}\chi_Y(z).
        \end{equation*}
        Injectivity and continuity of the Fourier--Weyl transform imply
        that $ \chi_Y $ does not vanish on some nonempty open set.
        Taking absolute values there gives
        \begin{equation*}
            \mathbf{z}^T\mathbf{S}_\tau\mathbf{z}
            =-2\tau\Re\lambda.
        \end{equation*}
        This is impossible because $ \mathbf{S}_\tau>0 $. Thus
        $ \sigma_{\mathrm{p}}(\cL_*)=\varnothing $.
    \end{proof}

    The final two lemmas ensure that the period average \eqref{eq-definition-eigenvector-X-mu-z} gives a nonzero eigenoperator on a dense set of phase-space points. The restriction is needed: if $ z=0 $ and $ n\neq0 $, then
    \begin{equation*}
        \mu = \frac{2 \pi \mi n}{\tau} , \quad W(0) = \1,
    \end{equation*}
    so that 
    \begin{equation*}
        X_{\tau, 0, n} = \int_0^\tau \me^{-2 \pi \mi n s / \tau} \cT_s(\1) \md s =  \frac{\tau}{-2\pi\mi n}( \me^{- 2 \pi \mi n} - 1 ) \1 = 0.
    \end{equation*}
    To avoid shorter or constant orbits, define
    \begin{equation} \label{eq-definition-D-tau-set}
        \cD_\tau := \{ z \in \C^d : s \mapsto \me^{s Z} z \text{ is injective on } [0, \tau) \}.
    \end{equation}

    \begin{lemma} \label{lemma-density-D-tau}
        The set $ \cD_\tau $ is dense in $ \C^d $.
    \end{lemma}
    \begin{proof}
        Recall that $ \tau $ is the least positive period of $ \me^{t\mathbf{Z}} $. For $ z\in\C^d $, its return times form the closed additive subgroup
        \begin{equation*}
            \mathfrak{t}_z := \{ t \in \R : \me^{t Z} z = z \}.
        \end{equation*}
        Since $ \tau\in\mathfrak{t}_z $, the classification of closed subgroups of $ \R $ gives $ \mathfrak{t}_z=\R $ or $ \mathfrak{t}_z=(\tau/n)\Z $ for some integer $ n\ge1 $. The orbit is injective on $ [0,\tau) $ precisely when $ \mathfrak{t}_z=\tau\Z $. Consequently,
        \begin{equation*}
            \C^d \setminus \cD_\tau = \bigcup_{n=2}^\infty \ker( \me^{(\tau / n) \mathbf{Z}} - \1 ).
        \end{equation*}
        Indeed, a nonconstant orbit with a shorter period belongs to the corresponding kernel; a constant orbit belongs to every kernel. For each $ n \ge 2 $, minimality of $ \tau $ implies 
        \begin{equation*}
            \me^{(\tau / n) \mathbf{Z}} \neq \1.
        \end{equation*}
        Thus each kernel in the union is a proper closed linear subspace of $ \R^{2d} $, and therefore has empty interior. By Baire's theorem, their countable union has empty interior. Its complement is therefore dense, and we have $ \overline{\cD_\tau} = \C^{d} $.
    \end{proof} 

    \begin{lemma} \label{lemma-nonzero-X-mu-z}
        For every $ z \in \cD_\tau $ and $ n \in \Z $, the operator $ X_{\tau, z, n} $ defined in \eqref{eq-definition-eigenvector-X-mu-z} is nonzero.
    \end{lemma}
    \begin{proof}
        Suppose that $ X_{\tau,z,n}=0 $ for some $ z\in\cD_\tau $. Evaluation in each normalized coherent state from \eqref{eq-coherent-characteristic-function} gives
        \begin{equation} \label{eq-zero-for-all-alpha-coherent-state}
            0 = \me^{-\abs{\alpha}^2} \Tr ( \ketbra{e(\alpha)}{e(\alpha)} X_{\tau, z, n} ) = \int_0^\tau \me^{-\mu s} \phi_s (z) \exp{ - \frac{1}{2} \abs{\me^{s Z} z}^2 + 2 \mi \Im \langle \alpha, \me^{s Z} z \rangle } \md s.
        \end{equation}
        For brevity, set
        \begin{equation*}
            f (s) := \me^{-\mu s} \phi_s(z) \exp{-\frac{1}{2} \abs{\me^{s Z}z}^2}.
        \end{equation*}
        The function $ f $ is continuous and nowhere zero. It is also $ \tau $-periodic, since $ \phi_{s+\tau}(z)=\phi_s(z)\phi_\tau(z) $ and $ \me^{-\mu\tau}\phi_\tau(z)=1 $. Indeed,
        \begin{equation*}
            f(s + \tau) = \me^{-\mu (s + \tau)} \phi_{s+\tau} (z) \exp{-\frac{1}{2} \abs{\me^{(s + \tau) Z} z}^2} = \me^{-\mu (s + \tau)} \phi_\tau (z) \phi_s (z) \exp{-\frac{1}{2} \abs{\me^{s Z} z}^2} = f(s).
        \end{equation*}
        Thus $ f(s) $ can be regarded as a continuous function on the circle
        \begin{equation*}
            \T_\tau := \R / \tau \Z,
        \end{equation*}
        which is compact in the quotient topology. Consider the functions
        \begin{equation*}
            s \mapsto \exp{ 2 \mi \Im \langle \alpha, \me^{s Z} z \rangle }, \quad \alpha \in \C^d. 
        \end{equation*}
        Their linear span is a self-adjoint subalgebra of $ \cC(\T_\tau) $. It contains the constants by taking $ \alpha=0 $; multiplication corresponds to addition of the parameters, and complex conjugation to changing their signs. It also separates points. Indeed, for distinct $ s,t\in[0,\tau) $, the assumption $ z\in\cD_\tau $ gives $ \me^{sZ}z\neq\me^{tZ}z $. Nondegeneracy of the symplectic form allows a choice of $ \alpha\in\C^d $ such that
        \begin{equation*}
            \exp{ 2 \mi \Im \langle \alpha, \me^{s Z} z \rangle } \neq \exp{ 2 \mi \Im \langle \alpha, \me^{t Z} z \rangle }.
        \end{equation*}
        Therefore, by the complex Stone--Weierstrass theorem, the linear span of these functions is dense in $ \cC(\mathbb T_\tau) $; see \cite[Theorem 7.33]{Rud76}.

        Equation \eqref{eq-zero-for-all-alpha-coherent-state} implies
        \begin{equation*}
            \int_0^\tau f(s) g(s)\,\md s = 0
        \end{equation*}
        for every function $ g $ in this dense subalgebra. Since integration against $ f(s)\, \md s $ is continuous with respect to the uniform norm, the same identity holds for every $ g \in \cC(\T_\tau) $. In particular, taking $ g(s) = \overline{f(s)} $, we have 
        \begin{equation*}
            0 = \int_0^\tau f(s) \overline{f(s)} \md s = \int_0^\tau \abs{f(s)}^2 \md s.
        \end{equation*}
        This is impossible because $ f $ is nowhere zero. Hence $ X_{\tau,z,n}\neq0 $.
    \end{proof}

    \section{Proofs for mixed stable--critical drift} \label{appendix-mixed-stable-critical}

    Throughout this section, $ s(\mathbf{Z})=0 $ and $ \mathbf{Z} $ has at least one eigenvalue with negative real part. We prove Theorem \ref{theorem-spectrum-mixed-stable-critical-drift}.

    We first construct a faithful Gaussian state whose covariance increases along its orbit. For $ \delta>0 $, define
    \begin{equation} \label{eq-definition-S-tilde-delta}
        \widetilde{\mathbf{S}}_\delta := \int_0^\infty \me^{-2 \delta s} \me^{s \mathbf{Z}^T} \mathbf{C} \me^{s \mathbf{Z}} \md s. 
    \end{equation}

    \begin{lemma} \label{lemma-well-definedness-S-tilde-delta}
        The matrix $ \widetilde{\mathbf{S}}_\delta $ is well defined and positive definite.
    \end{lemma}
    \begin{proof}
        Since every eigenvalue of $ \mathbf{Z} $ has nonpositive real part, $ \me^{s\mathbf{Z}} $ has at most polynomial growth. This follows from the Jordan normal form; we record the estimate for later use. Choose an invertible complex matrix $ P $ such that
        \begin{equation*}
            P^{-1} \mathbf{Z} P = \Lambda + A
        \end{equation*}
        is the Jordan normal form of $ \mathbf{Z} $, where $ \Lambda $ is diagonal and $ A $ is nilpotent, $ \Lambda A=A\Lambda $, and $ A^{2d}=0 $. Each Jordan block of $ A $ shifts the coordinates and discards one coordinate, so $ \norm{A}\le1 $. Moreover,
        \begin{equation*}
            \me^{s \mathbf{Z}} = P \me^{s \Lambda} \me^{s A} P^{-1} = P \me^{s \Lambda} \left( \sum_{k=0}^{2d - 1} \frac{s^k A^k}{k!} \right) P^{-1},
        \end{equation*}
        where the sum terminates because $ A^{2d}=0 $. Since every eigenvalue of $ \Lambda $ has nonpositive real part,
        \begin{equation*}
            \norm{\me^{s \Lambda}} = \max_{\lambda \in \sigma (\mathbf{Z})} \me^{s \Re \lambda} \le 1, \quad s \ge 0,
        \end{equation*}
        Here and below, matrix norms are operator norms induced by the Euclidean norm. It follows that
        \begin{equation*}
            \norm{\me^{s \mathbf{Z}}} \le \norm{P}  \norm{\me^{s \Lambda}} \sum_{k=0}^{2d - 1} \frac{s^k \norm{A}^k}{k!} \norm{P^{-1}} \le \sum_{k=0}^{2d-1} \frac{s^k}{k!} \norm{P} \norm{P^{-1}}.
        \end{equation*}
        Finally, for $ s \ge 0 $, 
        \begin{equation*}
            \sum_{k=0}^{2d-1} \frac{s^k}{k!} \le \sum_{k=0}^{2d - 1} s^k \le \sum_{k=0}^{2d-1} \binom{2d - 1}{k} s^k = (1 + s)^{2d - 1}.
        \end{equation*}
        Thus, with $ M=\norm{P}\norm{P^{-1}} $,
        \begin{equation} \label{eq-norm-estimate-exp-s-Z-aperiodic-case}
            \norm{\me^{s \mathbf{Z}}} \le M (1 + s)^{2d - 1}, \quad s \ge 0.
        \end{equation}
        Hence \eqref{eq-norm-estimate-exp-s-Z-aperiodic-case} gives
        \begin{equation} \label{eq-estimate-upper-bound-S-tilde-delta}
            \norm{ \me^{-2 \delta s} \me^{s \mathbf{Z}^T} \mathbf{C} \me^{s \mathbf{Z}} } \le \me^{- 2 \delta s} \norm{\mathbf{C}} \norm{ \me^{s \mathbf{Z}} }^2 \le M^2 \me^{-2 \delta s} (1 + s)^{4 d - 2} \norm{\mathbf{C}}.
        \end{equation}
        The last expression is integrable over $ [0,\infty) $, so the integral defining $ \widetilde{\mathbf{S}}_\delta $ converges in operator norm. Since $ \mathbf{C}\ge0 $, every integrand is positive semidefinite. Restricting the integral to $ [0,1] $, on which $ \me^{-2\delta s}\ge\me^{-2\delta} $, gives
        \begin{equation*}
            \widetilde{\mathbf{S}}_\delta \ge \int_0^1 \me^{-2 \delta s} \me^{s \mathbf{Z}^T} \mathbf{C} \me^{s \mathbf{Z}} \md s \ge \me^{-2 \delta} \mathbf{S}_1 > 0.
        \end{equation*}
        where $ \mathbf{S}_1>0 $ follows from irreducibility.
    \end{proof}
    
    \begin{lemma} \label{lemma-S-tilde-delta-lyapunov-equation}
        The matrix $ \widetilde{\mathbf{S}}_\delta $ satisfies the Lyapunov-type identity
        \begin{equation} \label{eq-S-tilde-delta-lyapunov-equation}
            \mathbf{Z}^T \widetilde{\mathbf{S}}_\delta + \widetilde{\mathbf{S}}_\delta \mathbf{Z} + \mathbf{C} = 2 \delta \widetilde{\mathbf{S}}_\delta.
        \end{equation}
    \end{lemma}
    \begin{proof}
        For notational simplicity, write
        \begin{equation*}
            F(s) := \me^{- 2 \delta s} \me^{s \mathbf{Z}^T} \mathbf{C} \me^{s \mathbf{Z}},
        \end{equation*}
        Differentiating gives
        \begin{equation} \label{eq-derivative-F-s}
            \frac{\md}{\md s} F(s) = \mathbf{Z}^T F(s) + F(s) \mathbf{Z} - 2 \delta F(s).
        \end{equation}
        Integrating \eqref{eq-derivative-F-s} over $ [0,t] $ and using $ F(0)=\mathbf{C} $ yields
        \begin{equation} \label{eq-fundamental-theorem-of-calculus-F-s}
            F (t) - \mathbf{C} = \mathbf{Z}^T \int_0^t F(s) \md s + \int_0^t F(s) \md s \, \mathbf{Z} - 2 \delta \int_0^t F(s) \md s.
        \end{equation}
        By Lemma \ref{lemma-well-definedness-S-tilde-delta}, the integrals above converge to $ \widetilde{\mathbf{S}}_\delta $. Moreover, \eqref{eq-norm-estimate-exp-s-Z-aperiodic-case} gives
        \begin{equation*}
            \norm{F(t)} \le M^2 \me^{-2 \delta t} (1 + t)^{4d - 2} \norm{\mathbf{C}} \rightarrow 0
        \end{equation*}
        as $ t\to\infty $. Letting $ t\to\infty $ in \eqref{eq-fundamental-theorem-of-calculus-F-s} and rearranging yields \eqref{eq-S-tilde-delta-lyapunov-equation}.
    \end{proof}

    \begin{lemma}
        For sufficiently small $ \delta > 0 $, we have 
        \begin{equation} \label{eq-S-tilde-delta-greater}
            \widetilde{\mathbf{S}}_\delta - \mi \mathbf{J} > 0.
        \end{equation}
    \end{lemma}
    \begin{proof}
        For $ \delta\ge0 $, set
        \begin{equation} \label{eq-definition-K-delta}
            \mathbf{K}_\delta := \int_0^1 \me^{- 2 \delta s} \me^{s \mathbf{Z}^T} ( \mathbf{C}_\mathbf{Z} - 2 \mi \delta \mathbf{J} ) \me^{s \mathbf{Z}} \md s.
        \end{equation}
        This finite-time integral is well defined, and the map $ \delta\mapsto\mathbf{K}_\delta $ is continuous at $ \delta=0 $. Irreducibility gives
        \begin{equation*}
            \mathbf{K}_0 = \int_0^1 \me^{s \mathbf{Z}^T} \mathbf{C}_\mathbf{Z} \me^{s \mathbf{Z}} \md s > 0. 
        \end{equation*}
        Hence there exists $ c>0 $ such that $ x^*\mathbf{K}_0x\ge c\abs{x}^2 $ for all $ x\in\C^{2d} $. By the Cauchy--Schwarz inequality and the definition of the operator norm,
        \begin{equation} \label{eq-upper-bound-x-star-K-delta-x}
            \abs{  x^* (\mathbf{K}_\delta - \mathbf{K}_0) x } \le \norm{ \mathbf{K}_\delta - \mathbf{K}_0 } \abs{x}^2. 
        \end{equation}
        Both $ \mathbf{K}_\delta $ and $ \mathbf{K}_0 $ are Hermitian, so these quadratic forms are real. Therefore, \eqref{eq-upper-bound-x-star-K-delta-x} gives
        \begin{equation*}
            x^* (\mathbf{K}_\delta - \mathbf{K}_0) x \ge - \abs{x^* ( \mathbf{K}_\delta - \mathbf{K}_0 ) x} \ge - \norm{ \mathbf{K}_\delta - \mathbf{K}_0 } \abs{x}^2.
        \end{equation*}
        Consequently,
        \begin{equation} \label{eq-lower-bound-x-star-K-delta-x}
            x^* \mathbf{K}_\delta x = x^* \mathbf{K}_0 x + x^* (\mathbf{K}_\delta - \mathbf{K}_0) x \ge c\abs{x}^2 - \norm{\mathbf{K}_\delta - \mathbf{K}_0} \abs{x}^2 = ( c - \norm{\mathbf{K}_\delta - \mathbf{K}_0} ) \abs{x}^2.
        \end{equation}

        The definition of $ \mathbf{K}_\delta $ gives
        \begin{equation*}
            \mathbf{K}_\delta - \mathbf{K}_0 = \int_0^1 \me^{s \mathbf{Z}^T} \left( (\me^{-2 \delta s} - 1) \mathbf{C}_\mathbf{Z} - 2 \mi \delta \me^{-2 \delta s} \mathbf{J} \right) \me^{s \mathbf{Z}} \md s.
        \end{equation*}
        Using $ 1 - \me^{- 2 \delta s} \le 2 \delta s $ and $ \me^{-2 \delta s} \le 1 $, we obtain 
        \begin{equation*}
            \norm{\mathbf{K}_\delta - \mathbf{K}_0} \le 2 \delta \int_0^1 \norm{\me^{s \mathbf{Z}}}^2 ( s \norm{\mathbf{C}_\mathbf{Z}} + \norm{\mathbf{J}} ) \md s \rightarrow 0
        \end{equation*}
        as $ \delta\to0^+ $; the integral above is a fixed finite constant. Thus, for sufficiently small $ \delta>0 $,
        \begin{equation*}
            \norm{\mathbf{K}_\delta - \mathbf{K}_0} < \frac{c}{2}.
        \end{equation*}
        Substitution into \eqref{eq-lower-bound-x-star-K-delta-x} yields $ x^*\mathbf{K}_\delta x\ge(c/2)\abs{x}^2 $, and hence
        \begin{equation*}
            \mathbf{K}_\delta \ge \frac{c}{2} \1 > 0.
        \end{equation*}

        Next, observe that
        \begin{equation*}
            \frac{\md}{\md s} \left( \me^{-2 \delta s} \me^{s \mathbf{Z}^T} \mathbf{J} \me^{s \mathbf{Z}} \right) = \me^{-2 \delta s} \me^{s \mathbf{Z}^T} ( \mathbf{Z}^T \mathbf{J} + \mathbf{J} \mathbf{Z} - 2 \delta \mathbf{J} ) \me^{s \mathbf{Z}},
        \end{equation*}
        so 
        \begin{equation*}
            \me^{-2 \delta s} \me^{s \mathbf{Z}^T} ( \mathbf{C}_\mathbf{Z} - 2 \mi \delta \mathbf{J} ) \me^{s \mathbf{Z}} = \me^{- 2 \delta s} \me^{s \mathbf{Z}^T} \mathbf{C} \me^{s \mathbf{Z}} + \mi \frac{\md}{\md s} \left(   \me^{-2 \delta s} \me^{s \mathbf{Z}^T} \mathbf{J} \me^{s \mathbf{Z}} \right).
        \end{equation*}
        Integrating over $ [0, t] $, the fundamental theorem of calculus yields
        \begin{equation} \label{eq-integral-form-C-Z-2-i-delta-J}
            \int_0^t \me^{- 2 \delta s} \me^{s \mathbf{Z}^T} ( \mathbf{C}_\mathbf{Z} - 2 \mi \delta \mathbf{J} ) \me^{s \mathbf{Z}} \md s = \int_0^t \me^{-2 \delta s} \me^{s \mathbf{Z}^T} \mathbf{C} \me^{s \mathbf{Z}} \md s + \mi \me^{- 2 \delta t} \me^{t \mathbf{Z}^T} \mathbf{J} \me^{t \mathbf{Z}} - \mi \mathbf{J}.
        \end{equation}
        By \eqref{eq-norm-estimate-exp-s-Z-aperiodic-case},
        \begin{equation*}
            \norm{\me^{-2 \delta t} \me^{t \mathbf{Z}^T} \mathbf{J} \me^{t \mathbf{Z}} } \le M^2 \norm{\mathbf{J}} \me^{- 2 \delta t} (1 + t)^{4d - 2} \rightarrow 0, 
        \end{equation*}
        as $ t\to\infty $. Letting $ t\to\infty $ in \eqref{eq-integral-form-C-Z-2-i-delta-J} gives
        \begin{equation} \label{eq-S-tilde-delta-i-J-integral-form}
            \int_0^\infty \me^{- 2 \delta s} \me^{s \mathbf{Z}^T} ( \mathbf{C}_\mathbf{Z} - 2 \mi \delta \mathbf{J} ) \me^{s \mathbf{Z}} \md s = \widetilde{\mathbf{S}}_\delta - \mi \mathbf{J}.
        \end{equation}
        Finally, split the integral in \eqref{eq-S-tilde-delta-i-J-integral-form} over the intervals $ [n,n+1] $ and substitute $ s=n+r $ on each interval. Recalling \eqref{eq-definition-K-delta}, we obtain
        \begin{align}
            \widetilde{\mathbf{S}}_\delta - \mi \mathbf{J} &= \sum_{n=0}^\infty \int_n^{n+1} \me^{- 2 \delta s} \me^{s \mathbf{Z}^T} ( \mathbf{C}_\mathbf{Z} - 2 \mi \delta \mathbf{J} ) \me^{s \mathbf{Z}} \md s \nonumber \\ 
            &= \sum_{n = 0}^\infty \me^{-2 \delta n} \me^{n \mathbf{Z}^T } \left( \int_0^1 \me^{-2 \delta r} \me^{r \mathbf{Z}^T} (\mathbf{C}_{\mathbf{Z}} - 2 \mi \delta \mathbf{J}) \me^{r \mathbf{Z} } \md r \right) \me^{n \mathbf{Z}} \nonumber \\
            &= \sum_{n=0}^{\infty} \me^{-2\delta n} \me^{n \mathbf{Z}^T} \mathbf{K}_\delta \me^{n \mathbf{Z}}. \label{eq-S-tilde-delta-minus-i-J-by-K-delta}
        \end{align}
        The series converges in operator norm because the exponential factor dominates the polynomial bound in \eqref{eq-norm-estimate-exp-s-Z-aperiodic-case}. For sufficiently small $ \delta>0 $, every summand in \eqref{eq-S-tilde-delta-minus-i-J-by-K-delta} is positive semidefinite, and the term corresponding to $ n=0 $ is $ \mathbf{K}_\delta>0 $. Hence
        \begin{equation*}
            \widetilde{\mathbf{S}}_\delta - \mi \mathbf{J} \ge \mathbf{K}_\delta > 0,
        \end{equation*}
        which completes the proof.
    \end{proof}

    Fix henceforth $ \delta>0 $ sufficiently small that $ \widetilde{\mathbf{S}}_\delta-\mi\mathbf{J}>0 $. We may then define a faithful centered Gaussian state $ \sigma_\delta $ by
    \begin{equation} \label{eq-definition-sigma-delta}
        \chi_{\sigma_\delta}(z):=\exp{-\frac{1}{2}\mathbf{z}^T\widetilde{\mathbf{S}}_\delta\mathbf{z}},
        \qquad z\in\C^d, \quad \mathbf{z}=\iota(z).
    \end{equation}
    This state is generally not invariant under $ (\cT_t)_{t\ge0} $.
    
    For convenience, define
    \begin{equation} \label{eq-definition-boldface-W-t}
        \mathbf{W}_t := 2 \delta \int_0^t \me^{s \mathbf{Z}^T} \widetilde{\mathbf{S}}_\delta \me^{s \mathbf{Z}} \md s 
    \end{equation}
    \begin{lemma}
        For every $ t>0 $, $ \mathbf{W}_t>0 $, and
        \begin{equation} \label{eq-characteristic-function-predual-sigma-delta}
            \chi_{\cT_{*t} (\sigma_\delta)} (z) = \exp{ -\frac{1}{2} \mathbf{z}^T \left( \widetilde{\mathbf{S}}_\delta + \mathbf{W}_t \right) \mathbf{z} + \mi \mathbf{b}_t^T \mathbf{z} }.
        \end{equation}
    \end{lemma}
    \begin{proof}
        By Lemma \ref{lemma-S-tilde-delta-lyapunov-equation},
        \begin{equation*}
            \frac{\md}{\md t} \left( \me^{t \mathbf{Z}^T} \widetilde{\mathbf{S}}_\delta \me^{t \mathbf{Z}} + \mathbf{S}_t \right) = \me^{t \mathbf{Z}^T} \left( \mathbf{Z}^T \widetilde{\mathbf{S}}_\delta + \widetilde{\mathbf{S}}_\delta \mathbf{Z} + \mathbf{C} \right) \me^{t \mathbf{Z}} = 2 \delta \me^{t \mathbf{Z}^T} \widetilde{\mathbf{S}}_\delta \me^{t \mathbf{Z}},
        \end{equation*}
        where $ \mathbf{S}_t $ is defined in \eqref{eq-definition-boldface-S-t}. Integrating this identity from $ 0 $ to $ t $ gives
        \begin{equation*}
            \me^{t \mathbf{Z}^T} \widetilde{\mathbf{S}}_\delta \me^{t \mathbf{Z}} + \mathbf{S}_t - \widetilde{\mathbf{S}}_\delta = 2 \delta \int_0^t \me^{s \mathbf{Z}^T} \widetilde{\mathbf{S}}_\delta \me^{s \mathbf{Z}} \md s.
        \end{equation*}  
        By \eqref{eq-definition-boldface-W-t}, this becomes
        \begin{equation*}
            \me^{t \mathbf{Z}^T} \widetilde{\mathbf{S}}_\delta \me^{t \mathbf{Z}} + \mathbf{S}_t = \widetilde{\mathbf{S}}_\delta + \mathbf{W}_t.
        \end{equation*}
        Consequently,
        \begin{equation*}
            \chi_{\cT_{*t} (\sigma_\delta)} (z) = \exp{ -\frac{1}{2} \mathbf{z}^T \left( \me^{t \mathbf{Z}^T} \widetilde{\mathbf{S}}_\delta \me^{t \mathbf{Z}} + \mathbf{S}_t \right) \mathbf{z} + \mi \mathbf{b}_t^T \mathbf{z} } =  \exp{ -\frac{1}{2} \mathbf{z}^T (\widetilde{\mathbf{S}}_\delta + \mathbf{W}_t) \mathbf{z} + \mi \mathbf{b}_t^T \mathbf{z} },
        \end{equation*}
        which proves \eqref{eq-characteristic-function-predual-sigma-delta}. Moreover,
        \begin{equation} \label{eq-x-T-W-t-x}
            x^T \mathbf{W}_t x = 2 \delta \int_0^t (\me^{s \mathbf{Z}} x)^T \widetilde{\mathbf{S}}_\delta ( \me^{s \mathbf{Z}} x) \md s > 0, \quad x \in \R^{2d} \setminus \{ 0 \}.
        \end{equation}
        because $ \widetilde{\mathbf{S}}_\delta>0 $ and $ \me^{s\mathbf{Z}} $ is invertible.
    \end{proof}

    Let $ p_t $ denote the Gaussian probability density on $ \R^{2d} $ with mean $ \mathbf{b}_t $ and covariance $ \mathbf{W}_t $:
    \begin{equation} \label{eq-definition-Gaussian-density-p-t}
        p_t (x) := \frac{1}{(2 \pi)^d \sqrt{\det \mathbf{W}_t}} \exp{ -\frac{1}{2} (x - \mathbf{b}_t)^T \mathbf{W}_t^{-1} (x - \mathbf{b}_t) }.
    \end{equation}
    For $ Y\in\cB_1(\mathsf{h}) $, define
    \begin{equation} \label{eq-definition-Psi-t}
        \Psi_t (Y) := \int_{\R^{2d}} p_t (x) W( \iota^{-1} (\mathbf{J} x / 2) ) Y W( \iota^{-1} (\mathbf{J} x / 2) )^* \md x, \quad t > 0,
    \end{equation}
    and set $ \Psi_0=\1 $. This is a well-defined Bochner integral in trace norm, and
    \begin{equation} \label{eq-upper-bound-Psi-t}
        \norm{\Psi_t (Y)}_1 \le \norm{Y}_1.
    \end{equation}
    The proof is identical to the Bochner-integrability argument for \eqref{eq-auxiliary-qms-predual}; see the proof of Lemma \ref{eq-equiv-form-action-auxiliary-QMS}. Although $ \Psi_t $ resembles the predual of the auxiliary quantum diffusion semigroup, the family $ (\Psi_t)_{t\ge0} $ does not in general form a semigroup.

    \begin{lemma}
        For every $ Y\in\cB_1(\mathsf{h}) $,
        \begin{equation} \label{eq-characteristic-function-Psi-t-Y}
            \chi_{\Psi_t (Y)} (z) = \exp{-\frac{1}{2} \mathbf{z}^T \mathbf{W}_t \mathbf{z} + \mi \mathbf{b}_t^T \mathbf{z} } \chi_Y (z), \quad t > 0.
        \end{equation}
        In particular,
        \begin{equation} \label{eq-coincidence-Psi-t-predual-ob-sigma-delta}
            \Psi_t ( \sigma_\delta ) = \cT_{*t} (\sigma_\delta), \quad \forall t \ge 0.
        \end{equation}
    \end{lemma}
    \begin{proof}
        For $ t>0 $, \eqref{eq-definition-Psi-t} and the Gaussian integral formula give
        \begin{align*}
            \chi_{\Psi_t (Y)} (z) &= \Tr (\Psi_t (Y) W(z)) = \int_{\R^{2d}} p_t(x) \Tr( Y W(\iota^{-1} (\mathbf{J} x / 2))^* W(z) W(\iota^{-1} (\mathbf{J} x / 2)) ) \md x \\ 
            &= \int_{\R^{2d}} p_t (x) \me^{\mi \mathbf{z}^T x} \md x \cdot \chi_Y(z) = \exp{-\frac{1}{2} \mathbf{z}^T \mathbf{W}_t \mathbf{z} + \mi \mathbf{b}_t^T \mathbf{z} } \chi_Y (z).
        \end{align*}
        Combining \eqref{eq-definition-sigma-delta} and \eqref{eq-characteristic-function-predual-sigma-delta}, we obtain
        \begin{equation*}
            \chi_{\Psi_t (\sigma_\delta)}(z) = \exp{-\frac{1}{2} \mathbf{z}^T \mathbf{W}_t \mathbf{z} + \mi \mathbf{b}_t^T \mathbf{z} } \chi_{\sigma_\delta} (z) = \exp{-\frac{1}{2} \mathbf{z}^T (\mathbf{W}_t + \widetilde{\mathbf{S}}_\delta) \mathbf{z} + \mi \mathbf{b}_t^T \mathbf{z} } = \chi_{\cT_{*t} (\sigma_\delta)} (z).
        \end{equation*}
        Injectivity of the Fourier--Weyl transform proves \eqref{eq-coincidence-Psi-t-predual-ob-sigma-delta}.
    \end{proof}

    We now adapt the construction from Appendix \ref{appendix-stable-proofs} to obtain approximate eigenvectors of $ \cL_* $. Fix $ \eta\in\C $ with $ \Re\eta<0 $, choose an eigenvalue $ \lambda $ of $ \mathbf{Z} $ with $ \Re\lambda<0 $, and set $ \alpha:=\eta/(2\Re\lambda) $. Then $ \Re\alpha>0 $. As in Appendix \ref{appendix-stable-proofs}, choose $ v\in\C^{2d}\setminus\{0\} $ such that
    \begin{equation} \label{eq-copied-v-as-eigenvector-of-Z-T}
        \mathbf{Z}^T v = \lambda v, \quad \Re \lambda < 0.
    \end{equation} 
    Define $ Q_\lambda $ and $ q_\lambda $ as in \eqref{eq-definition-Q-lambda} and \eqref{eq-definition-q-lambda}, respectively. Then
    \begin{equation} \label{eq-copied-q-lambda-etlambda-z}
        q_\lambda (\me^{t Z} z) = \me^{2 (\Re \lambda) t} q_\lambda (z).
    \end{equation}

    \begin{lemma}
        There exists $ \widetilde{Y}_\alpha  \in \cB_1 (\mathsf{h}) $ such that 
        \begin{equation} \label{eq-relation-tilde-Y-alpha}
            \chi_{\widetilde{Y}_\alpha} (z) = \chi_{\sigma_\delta} (z) q_\lambda (z)^\alpha.
        \end{equation}
    \end{lemma}
    \begin{proof}
        Let $ (\cG_t)_{t\ge0} $ be the quantum diffusion semigroup associated with $ q_\lambda $ and $ Q_\lambda $, as defined in \eqref{eq-definition-auxiliary-quantum-diffusion-semigroup-G-t}. The mixed stable--critical case need not admit a faithful normal invariant state, so we instead use $ \widetilde{\mathbf{S}}_\delta $. Let $ \widetilde{\sigma} $ be the centered Gaussian state with covariance matrix $ \widetilde{\mathbf{S}}_\delta-2\epsilon Q_\lambda $, where
        \begin{equation*}
            0 < \epsilon < \frac{\lambda_{\min} (\widetilde{\mathbf{S}}_\delta - \mi \mathbf{J}) }{2 \norm{Q_\lambda}}.
        \end{equation*}
        The same argument as in Lemma \ref{lemma-construction-Y-alpha} shows that, for $ m\in\N $ with $ m>\Re\alpha $, the operator
        \begin{equation} \label{eq-definition-widetilde-Y-alpha}
            \widetilde{Y}_\alpha := \frac{(-1)^m}{\Gamma (m-\alpha)} \int_0^\infty t^{m - \alpha - 1} \left. \frac{\md^m}{\md s^m} \cG_{*s} (\widetilde{\sigma}) \right\vert_{s = \epsilon + t} \md t,
        \end{equation}
        is trace class and satisfies \eqref{eq-relation-tilde-Y-alpha}.
    \end{proof}

    \begin{lemma}
        We have 
        \begin{equation} \label{eq-composition-predual-semigroup-and-Psi-t}
            \cT_{*s} ( \Psi_t (\widetilde{Y}_\alpha) ) = \me^{\eta s} \Psi_{t + s} (\widetilde{Y}_\alpha).
        \end{equation}
    \end{lemma}
    \begin{proof}
        A direct computation using \eqref{eq-characteristic-function-Psi-t-Y}, \eqref{eq-coincidence-Psi-t-predual-ob-sigma-delta}, \eqref{eq-copied-q-lambda-etlambda-z}, and \eqref{eq-relation-tilde-Y-alpha} gives
        \begin{align*}
            \chi_{\cT_{*t} (\widetilde{Y}_\alpha)} (z) &= \phi_t (z) \cdot \chi_{\widetilde{Y}_\alpha} (\me^{t Z} z) = \phi_t (z) \cdot \chi_{\sigma_\delta} (\me^{t Z} z) \cdot q_\lambda (\me^{t Z} z)^\alpha \\  
            &= \chi_{\cT_{*t}(\sigma_\delta)} (z) \cdot \me^{\eta t} \cdot q_\lambda (z)^\alpha = \me^{\eta t} \cdot \chi_{\Psi_t (\sigma_\delta)} (z) \cdot  q_\lambda (z)^\alpha \\ 
            &= \me^{\eta t} \cdot \exp{ -\frac{1}{2} \mathbf{z}^T \mathbf{W}_t \mathbf{z} + \mi \mathbf{b}_t^T \mathbf{z} } \cdot \chi_{\sigma_\delta} (z) \cdot  q_\lambda (z)^\alpha \\ 
            &= \me^{\eta t} \cdot \exp{ -\frac{1}{2} \mathbf{z}^T \mathbf{W}_t \mathbf{z} + \mi \mathbf{b}_t^T \mathbf{z} } \cdot \chi_{\widetilde{Y}_\alpha} (z) = \me^{\eta t} \chi_{ \Psi_t (\widetilde{Y}_\alpha) } (z),
        \end{align*}
        for all $ t \ge 0 $. By injectivity of the Fourier--Weyl transform,
        \begin{equation} \label{eq-action-predual-semigroup-on-tilded-Y-alpha}
            \cT_{*t} ( \widetilde{Y}_\alpha ) = \me^{\eta t} \Psi_t (\widetilde{Y}_\alpha), \quad t \ge 0.
        \end{equation}
        Therefore, \eqref{eq-action-predual-semigroup-on-tilded-Y-alpha} gives
        \begin{equation*}
            \cT_{*s} ( \Psi_t (\widetilde{Y}_\alpha) ) = \cT_{*s} ( \me^{-\eta t} \cT_{*t} (\widetilde{Y}_\alpha) ) = \me^{-\eta t} \cT_{*(t + s)} (\widetilde{Y}_\alpha) = \me^{-\eta t} \me^{\eta (t + s)} \Psi_{t+s} (\widetilde{Y}_\alpha) = \me^{\eta s} \Psi_{t + s} (\widetilde{Y}_\alpha),
        \end{equation*}
        which proves \eqref{eq-composition-predual-semigroup-and-Psi-t}.
    \end{proof}

    We next define two real subspaces from the generalized eigenspaces associated with stable and critical eigenvalues, respectively:
    \begin{equation} \label{eq-definition-E-0-and-E-minus}
        E_- := \R^{2d} \bigcap \bigoplus_{ \substack{\mu \in \sigma(\mathbf{Z}) \\ \Re \mu < 0} } \ker \left( (\mathbf{Z}  - \mu \1)^{2d} \right), \quad E_0 := \R^{2d} \bigcap \bigoplus_{ \substack{\mu \in \sigma(\mathbf{Z}) \\ \Re \mu = 0} } \ker \left( (\mathbf{Z}  - \mu \1)^{2d} \right).
    \end{equation}
    The kernels are taken in $ \C^{2d} $. The following standard facts follow from the spectral decomposition of a finite-dimensional linear system; see \cite[Sections 3.8 and 9.1]{Teschl2012ODE}. A related decomposition is used in \cite[Section 3]{girotti2025gaussian} to study Gaussian QMSs with a normal invariant state.

    \begin{lemma} 
        The subspaces $ E_- $ and $ E_0 $ are real $ \mathbf{Z} $-invariant subspaces satisfying 
        \begin{equation} \label{eq-decomposition-R-2d}
            \R^{2d} = E_- \oplus E_0,
        \end{equation}
        and 
        \begin{equation} \label{eq-spectra-restrictions-of-drift-matrix}
            \sigma(\mathbf{Z} \restriction_{E_0} ) \subset \mi \R, \quad \sigma(\mathbf{Z} \restriction_{E_-}) \subset \{ \Re \mu < 0 \}.
        \end{equation}
        Furthermore, 
        \begin{equation} \label{eq-equivalent-from-E-minus}
            E_- = \{ x \in \R^{2d} : \me^{t \mathbf{Z}} x \rightarrow 0 \},
        \end{equation}
        and 
        \begin{equation} \label{eq-exponential-stability-on-stable-manifold}
            \abs{\me^{t \mathbf{Z}} x} \le M \me^{- c t} \abs{x}, \quad x \in E_-, \quad t \ge 0,
        \end{equation}
        for suitable $ M, c > 0 $. Moreover, there exists a real invertible matrix $ \mathbf{P} $ such that 
        \begin{equation} \label{eq-block-diagonalization-drift-matrix}
            \mathbf{P}^{-1} \mathbf{Z} \mathbf{P} = \begin{pmatrix}
                \mathbf{Z} \restriction_{E_-} & 0 \\ 
                0 & \mathbf{Z} \restriction_{E_0}
            \end{pmatrix}.
        \end{equation}
    \end{lemma}

    \begin{proof}[Proof of Theorem \ref{theorem-spectrum-mixed-stable-critical-drift}]
        Define the time averages
        \begin{equation} \label{eq-definition-time-average-Y}
            \cY_t := \frac{1}{t} \int_0^t \Psi_s (\widetilde{Y}_\alpha) \md s.
        \end{equation}
        The integral in \eqref{eq-definition-time-average-Y} is well defined in trace norm: by \eqref{eq-action-predual-semigroup-on-tilded-Y-alpha}, $ \Psi_s(\widetilde{Y}_\alpha)=\me^{-\eta s}\cT_{*s}(\widetilde{Y}_\alpha) $, and the right-hand side is trace-norm continuous for $ s\ge0 $. Moreover, \eqref{eq-upper-bound-Psi-t} gives
        \begin{equation*}
            \norm{\cY_t}_{1} = \norm{\frac{1}{t} \int_0^t \Psi_s (\widetilde{Y}_\alpha) \md s}_1 \le \frac{1}{t} \int_0^t \norm{ \Psi_s ( \widetilde{Y}_\alpha ) }_1 \md s \le \norm{ \widetilde{Y}_\alpha }_1.
        \end{equation*}
        By \eqref{eq-composition-predual-semigroup-and-Psi-t},
        \begin{align*}
            \cT_{*h} (\cY_t) &= \cT_{*h} \left( \frac{1}{t} \int_0^t \Psi_s (\widetilde{Y}_\alpha) \md s \right) = \frac{1}{t} \int_0^t \cT_{*h} (\Psi_{s} (\widetilde{Y}_\alpha)) \md s \\ 
            &= \frac{1}{t} \int_0^t \me^{\eta h} \Psi_{s + h} (\widetilde{Y}_\alpha) \md s = \frac{\me^{\eta h}}{t} \int_h^{t+h} \Psi_s (\widetilde{Y}_\alpha) \md s \\ 
            &= \me^{\eta h} \cY_t  + \frac{\me^{\eta h}}{t} \left( \int_{t}^{t + h} \Psi_s ( \widetilde{Y}_\alpha ) \md s - \int_0^h \Psi_s (\widetilde{Y}_\alpha) \md s \right).
        \end{align*}
        Therefore, 
        \begin{equation*}
            \frac{\cT_{*h} (\cY_t) - \cY_t }{h} = \frac{\me^{\eta h} - 1}{h} \cY_t + \frac{\me^{\eta h}}{t} \left( \frac{1}{h} \int_{t}^{t + h} \Psi_s ( \widetilde{Y}_\alpha ) \md s - \frac{1}{h} \int_0^h \Psi_s (\widetilde{Y}_\alpha) \md s \right). 
        \end{equation*}
        Letting $ h \downarrow 0 $ we obtain that $ \cY_t \in \dom \cL_* $ and 
        \begin{equation*}
            \cL_* (\cY_t) = \lim_{h \downarrow 0} \frac{\cT_{*h} (\cY_t) - \cY_t }{h} = \eta \cY_t + \frac{\Psi_t (\widetilde{Y}_\alpha) - \widetilde{Y}_\alpha }{t},
        \end{equation*}
        which implies that 
        \begin{equation*}
            (\cL_{*} - \eta) (\cY_t) = \frac{\Psi_t (\widetilde{Y}_\alpha) - \widetilde{Y}_\alpha }{t}.
        \end{equation*}
        By \eqref{eq-upper-bound-Psi-t},
        \begin{equation} \label{eq-action-predual-generator-minus-eta-on-cY-t}
            \norm{(\cL_{*} - \eta) (\cY_t)}_1 = \norm{\frac{\Psi_t (\widetilde{Y}_\alpha) - \widetilde{Y}_\alpha }{t}}_1 \le \frac{2 \norm{\widetilde{Y}_\alpha}_1}{t} \rightarrow 0 
        \end{equation}
        as $ t\to\infty $. We next normalize $ (\cY_t)_{t>0} $ to obtain an approximate eigensequence. The characteristic function of $ \cY_t $ gives the lower bound
        \begin{equation} \label{eq-lower-bound-by-characteristic-function} 
            \norm{\cY_t}_1 \ge \abs{ \Tr(\cY_t W(z)) } = \abs{ \chi_{\cY_t} (z) }, \quad \forall z \in \C^d.
        \end{equation}
        Recall that $ \lambda\in\C $ with $ \Re\lambda<0 $ and $ v\in\C^{2d}\setminus\{0\} $ satisfy $ \mathbf{Z}^Tv=\lambda v $. Taking adjoints gives
        \begin{equation} \label{eq-left-eigenvector-v-star-Z}
            v^* \mathbf{Z} = \overline{\lambda} v^*.
        \end{equation}
        Now, by \eqref{eq-decomposition-R-2d} we have $ \R^{2d} = E_- \oplus E_0 $. Choose a real basis adapted to this decomposition. By \eqref{eq-block-diagonalization-drift-matrix}, in these coordinates, write the drift matrix and the row functional $ u \mapsto v^* u $ as
        \begin{equation*}
            \mathbf{Z} = \begin{pmatrix}
                \mathbf{Z} \restriction_{E_-} & 0 \\ 
                0 & \mathbf{Z} \restriction_{E_0}
            \end{pmatrix}, \quad v^* = (x, y).
        \end{equation*}
        Here $ x $ acts on the stable coordinates and $ y $ on the critical coordinates. Equation \eqref{eq-left-eigenvector-v-star-Z} becomes $ x(\mathbf{Z}\restriction_{E_-})=\overline{\lambda}x $ and $ y(\mathbf{Z}\restriction_{E_0})=\overline{\lambda}y $. In particular,
        \begin{equation*}
            y \, (\mathbf{Z} \restriction_{E_0} -\overline{\lambda} \1 ) = 0.
        \end{equation*}
        By \eqref{eq-spectra-restrictions-of-drift-matrix} we know that every eigenvalue of $ \mathbf{Z} \restriction_{E_0} $ has real part zero, whereas $ \Re \overline{\lambda} < 0 $. Therefore, $ \mathbf{Z} \restriction_{E_0} -\overline{\lambda} \1 $ is invertible and $ y = 0 $. Since $ v^* \neq 0 $, we must have $ x \neq 0 $. This says that the functional defining $ q_\lambda $, that is \eqref{eq-definition-q-lambda}, vanishes on the critical subspace and is nonzero on the stable subspace. Choose a stable basis vector on which $ x $ is nonzero. Returning to the original coordinates gives a real vector $ \mathbf z_0=\iota(z_0)\in E_- $ satisfying
        \begin{equation*}
            q_\lambda (z_0) = \abs{v^* \mathbf{z}_0}^2 > 0.
        \end{equation*}

        By \eqref{eq-equivalent-from-E-minus} and \eqref{eq-exponential-stability-on-stable-manifold}, $ \me^{s\mathbf{Z}}\mathbf{z}_0 $ decays exponentially. Both integrals below are therefore finite. By rescaling $ \mathbf{z}_0 $, choose $ z=rz_0 $ with $ r>0 $ so that
        \begin{equation*}
            r^2 \delta \int_0^\infty (\me^{s \mathbf{Z}} \mathbf{z}_0 )^T \widetilde{\mathbf{S}}_\delta ( \me^{s \mathbf{Z}} \mathbf{z}_0 ) \md s + r \int_0^\infty \abs{ \boldsymbol{\zeta}^T \me^{s \mathbf{Z}} \mathbf{z}_0 } \md s \le \frac{1}{2}.
        \end{equation*}
        Henceforth, $ z $ denotes this fixed vector. We also retain
        \begin{equation} \label{eq-q-lambda-z-strictly-positive}
            q_\lambda (z) = r^2 q_\lambda (z_0) > 0.
        \end{equation}
        By \eqref{eq-x-T-W-t-x} and \eqref{eq-definition-b-t} we have
        \begin{equation*}
            \frac{1}{2} \mathbf{z}^T \mathbf{W}_t \mathbf{z} + \abs{\mathbf{b}_t^T \mathbf{z}} \le \frac{1}{2}.
        \end{equation*}
        Using $ 1-\me^{-u}\le u $ for $ u\ge0 $ and $ \abs{\me^{\mi u}-1}\le\abs{u} $ for real $u$, we obtain
        \begin{equation} \label{eq-absolute-value-difference}
            \abs{ \me^{-\frac{1}{2} \mathbf{z}^T \mathbf{W}_t \mathbf{z} + \mi \mathbf{b}_t^T \mathbf{z} } - 1 } \le \me^{-\frac{1}{2} \mathbf{z}^T \mathbf{W}_t \mathbf{z}} \abs{ \me^{\mi \mathbf{b}_t^T \mathbf{z}} - 1 } + \abs{ \me^{-\frac{1}{2} \mathbf{z}^T \mathbf{W}_t \mathbf{z}} - 1 } \le \frac{1}{2} \mathbf{z}^T \mathbf{W}_t \mathbf{z} + \abs{\mathbf{b}_t^T \mathbf{z}} \le \frac{1}{2}.
        \end{equation}
        For this fixed $ z $, \eqref{eq-characteristic-function-Psi-t-Y} and \eqref{eq-absolute-value-difference} imply
        \begin{equation*}
            \abs{ \chi_{\Psi_t (\widetilde{Y}_\alpha)} (z) - \chi_{\widetilde{Y}_\alpha} (z) } = \abs{ \me^{-\frac{1}{2} \mathbf{z}^T \mathbf{W}_t \mathbf{z} + \mi \mathbf{b}_t^T \mathbf{z} } - 1 } \abs{ \chi_{\widetilde{Y}_\alpha} (z) } \le \frac{1}{2} \abs{ \chi_{\widetilde{Y}_\alpha} (z) }, \quad \forall t \ge 0.
        \end{equation*}
        Similarly, we have 
        \begin{align*}
            \abs{ \chi_{\cY_t} (z) - \chi_{\widetilde{Y}_\alpha} (z) } &= \abs{ \frac{1}{t} \int_0^t \left( \chi_{\Psi_s (\widetilde{Y}_\alpha)} (z) - \chi_{\widetilde{Y}_\alpha} (z) \right) \md s } \\ 
            &\le \frac{1}{t} \int_0^t \abs{\chi_{\Psi_s (\widetilde{Y}_\alpha)} (z) - \chi_{\widetilde{Y}_\alpha}(z) }  \md s  \\
            &\le \frac{1}{t} \int_0^t \frac{1}{2} \abs{ \chi_{\widetilde{Y}_\alpha}(z) } \md s = \frac{1}{2} \abs{ \chi_{\widetilde{Y}_\alpha } (z) }.
        \end{align*}
        Since the Gaussian characteristic function $ \chi_{\sigma_\delta} $ is nowhere zero and \eqref{eq-q-lambda-z-strictly-positive} holds,
        \begin{equation*}
            \chi_{\widetilde{Y}_\alpha} (z) = \chi_{\sigma_\delta} (z) q_\lambda (z)^\alpha \neq 0.
        \end{equation*}
        Consequently, by \eqref{eq-lower-bound-by-characteristic-function} we have
        \begin{equation} \label{eq-uniform-lower-bound-cY-t}
            \norm{\cY_t}_1 \ge \abs{ \chi_{\cY_t} (z) } \ge \abs{ \chi_{\widetilde{Y}_\alpha} (z) } - \abs{ \chi_{\cY_t} (z) - \chi_{\widetilde{Y}_\alpha} (z) } \ge \frac{1}{2} \abs{ \chi_{\widetilde{Y}_\alpha} (z) } > 0,
        \end{equation}
        The lower bound is independent of $ t $.

        Combining \eqref{eq-action-predual-generator-minus-eta-on-cY-t} with \eqref{eq-uniform-lower-bound-cY-t}, we obtain
        \begin{equation*}
            \norm{ (\cL_* - \eta) \frac{\cY_t}{\norm{\cY_t}_1} }_1 \le \frac{2 \norm{ \widetilde{Y}_\alpha }_1}{t \norm{\cY_t}_1} \le \frac{4 \norm{ \widetilde{Y}_\alpha }_1}{ t \abs{\chi_{\widetilde{Y}_\alpha} (z)} } \rightarrow 0,
        \end{equation*}
        as $ t\to\infty $. We have therefore shown that
        \begin{equation*}
            \{ \eta \in \C : \Re \eta < 0 \} \subset \sigma_{\mathrm{ap}} (\cL_*) \subset \sigma(\cL_*) \subset \{ \eta \in \C : \Re \eta \le 0 \}.
        \end{equation*}
        Since $ \sigma(\cL_*) $ is closed,
        \begin{equation*}
            \sigma (\cL_*) = \{ \eta \in \C : \Re \eta \le 0 \}.
        \end{equation*}
        Moreover, \cite[Chapter IV, Proposition 1.10]{EngelNagel2000} gives
        \begin{equation*}
            \mi \R = \partial \sigma (\cL_*) \subset \sigma_{\mathrm{ap}} (\cL_*).
        \end{equation*}
        Together with the open half-plane already constructed, this proves \eqref{eq-spectrum-predual-generator-mixed-stable-critical-drift}.
    \end{proof}

    \section{Proofs for purely critical semisimple drift} \label{appendix-purely-critical-semisimple}

    Throughout this section, $ \sigma(\mathbf{Z})\subset\mi\R $, the matrix $ \mathbf{Z} $ is semisimple, and the flow is nonperiodic. We prove Theorem \ref{theorem-spectrum-purely-critical-semisimple-drift}.

    \begin{lemma}
        There exists a real invertible matrix $ \mathbf{P} $ such that 
        \begin{equation} \label{eq-block-diagonal-representation-of-drift-matrix-purely-critical-case}
            \mathbf{B} := \mathbf{P}^{-1} \mathbf{Z} \mathbf{P} = 0_m \oplus \bigoplus_{j = 1}^r \begin{pmatrix}
                0 & \omega_j \\ 
                - \omega_j & 0
            \end{pmatrix},
            \quad \omega_j > 0, \quad m + 2 r = 2 d,
        \end{equation}
        where $ 0_m $ denotes the $ m \times m $ zero matrix. The frequencies $ \omega_j $ are listed with multiplicity, so they need not be distinct.
    \end{lemma}
    \begin{proof}
        Since $ \mathbf{Z} $ is real, complex conjugation of $ \mathbf{Z}v=\mi\omega v $ gives $ \mathbf{Z}\overline{v}=-\mi\omega\overline{v} $. Thus its nonzero eigenvalues occur in conjugate pairs. Choose bases for the eigenspaces corresponding to the positive imaginary eigenvalues, and list their vectors as $ v_1,\ldots,v_r $, where $ \mathbf{Z}v_j=\mi\omega_jv_j $ and $ \omega_j>0 $. Their conjugates form bases for the corresponding negative imaginary eigenspaces. Choose also a real basis $ h_1,\ldots,h_m $ of $ \ker_\R\mathbf{Z} $; it is then a complex basis of $ \ker_\C\mathbf{Z} $. Diagonalizability implies that
        \begin{equation*}
            h_1, \ldots, h_m, v_1, \overline{v}_1, \ldots, v_r, \overline{v}_r
        \end{equation*}
        is a complex basis of $ \C^{2d} $. Write $  v_j = \Re v_j + \mi \Im v_j $, where $ \Re v_j, \Im v_j \in \R^{2d} $. Comparing real and imaginary parts in 
        \begin{equation*}
            \mathbf{Z} \left( \Re v_j + \mi \Im v_j \right) = \mi \omega_j ( \Re v_j + \mi \Im v_j ) = - \omega_j \Im v_j + \mi \omega_j \Re v_j
        \end{equation*}
        gives
        \begin{equation*}
            \mathbf{Z} \Re v_j = - \omega_j \Im v_j, \quad \mathbf{Z} \Im v_j = \omega_j \Re v_j.
        \end{equation*}
        Consequently, the real plane spanned by $ \Re v_j $ and $ \Im v_j $ is invariant, and the matrix of the restriction in this ordered basis is 
        \begin{equation*}
            \begin{pmatrix}
                0 & \omega_j \\
                - \omega_j & 0
            \end{pmatrix}.
        \end{equation*}
        Finally, $ \Re v_j=(v_j+\overline{v}_j)/2 $ and $ \Im v_j=(v_j-\overline{v}_j)/(2\mi) $. Replacing each pair $ (v_j,\overline{v}_j) $ by $ (\Re v_j,\Im v_j) $ is an invertible change of basis. Hence
        \begin{equation} \label{eq-real-basis-purely-critical-case}
            h_1, \ldots, h_m, \Re v_1, \Im v_1, \Re v_2, \Im v_2, \ldots, \Re v_r, \Im v_r
        \end{equation}
        is a real basis of $ \R^{2d} $. Taking these vectors as the columns of $ \mathbf{P} $ proves \eqref{eq-block-diagonal-representation-of-drift-matrix-purely-critical-case}.
    \end{proof}

    For later use, define the additive frequency group
    \begin{equation} \label{eq-critical-frequency-group}
        \Xi_{\mathbf{Z}}:=\Z\omega_1+\cdots+\Z\omega_r.
    \end{equation}

    We work henceforth in the real basis \eqref{eq-real-basis-purely-critical-case}. For $ \mathbf{z}\in\R^{2d} $, set
    \begin{equation} \label{eq-definition-x-as-inverse-of-z}
        x := \mathbf{P}^{-1} \mathbf{z},
    \end{equation}
    Thus $ x $ is the coordinate vector of $ \mathbf{z} $ in this basis. Write $ x=(x_0,x_1,\ldots,x_r) $, where $ x_0\in\R^m $ corresponds to the zero block and
    \begin{equation} \label{eq-sub-entries-of-x-j}
        x_j = (x_{j,1}, x_{j,2}) \in \R^2 
    \end{equation}
    corresponds to the $ j $-th rotation block. Here $ m+2r=2d $.

    Therefore, the action of $ \me^{t \mathbf{B}} $ on $ x $ is given by
    \begin{equation*}
        \me^{t \mathbf{B}} x = \left( x_0, \begin{pmatrix}
            \cos(\omega_1 t) & \sin(\omega_1 t) \\ 
            - \sin (\omega_1 t) & \cos (\omega_1 t)
        \end{pmatrix} x_1, \ldots, \begin{pmatrix}
            \cos(\omega_r t) & \sin(\omega_r t) \\ 
            - \sin (\omega_r t) & \cos (\omega_r t)
        \end{pmatrix} x_r \right).
    \end{equation*}
    In particular,
    \begin{equation*}
        \abs{\me^{t \mathbf{B}} x} = \abs{x}.
    \end{equation*}

    For a sufficiently regular function $ f $, define
    \begin{equation} \label{eq-definition-D-as-B-times-nabla}
        D f(x) := \left. \frac{\md}{\md t} f(\me^{t \mathbf{B}} x) \right\vert_{t = 0} = (\mathbf{B} x) \cdot \nabla f (x). 
    \end{equation}
    This is the transport, or drift, operator familiar from classical Ornstein--Uhlenbeck theory. It should not be confused with the real-linear drift $ Z $.

    We now express the Gaussian formula in these coordinates. As usual, real phase-space vectors are identified with complex phase-space vectors when writing Weyl operators and characteristic functions. We can then define 
    \begin{equation} \label{eq-definition-psi-boldface-P}
        \psi_{\mathbf{P}} (x) := -\frac{1}{2} x^T \mathbf{P}^T \mathbf{C} \mathbf{P} x + \mi \boldsymbol{\zeta}^T \mathbf{P} x. 
    \end{equation}
    This function is simply $ \psi\circ\mathbf{P} $, where $ \psi $ is defined in \eqref{eq-definition-psi-z}.

    \begin{lemma} \label{lemma-decomposition-of-function-psi-boldface-P}
        Let $ \mathbf{B} $ and $ \psi_\mathbf{P} $ be defined by \eqref{eq-block-diagonal-representation-of-drift-matrix-purely-critical-case} and \eqref{eq-definition-psi-boldface-P}. Then there is a finite set $ \cF\subset\R\setminus\{0\} $ and, for each $ \mu\in\cF $, a polynomial $ \psi_\mu $ of degree at most two such that
        \begin{equation} \label{eq-decomposition-of-action-of-psi-P-on-e-t-B-x}
            \psi_{\mathbf{P}} (\me^{t \mathbf{B}} x) = \psi_0 (x)  +\sum_{\mu \in \cF} \me^{\mi \mu t} \psi_{\mu} (x).
        \end{equation}
    \end{lemma}
    \begin{proof}
        For one rotation block, we have the following matrix identity
        \begin{equation} \label{eq-decomposition-rotation-matrix}
            \begin{pmatrix}
                \cos (\omega_j t) & \sin(\omega_j t) \\ 
                - \sin (\omega_j t) & \cos(\omega_j t)
            \end{pmatrix} =  \frac{\me^{\mi \omega_j t}}{2} \begin{pmatrix}
                1 & - \mi \\ 
                \mi & 1
            \end{pmatrix} + \frac{\me^{-\mi \omega_j t}}{2} \begin{pmatrix}
                1 & \mi \\ 
                -\mi & 1
            \end{pmatrix}.
        \end{equation}
        Define the finite set
        \begin{equation} \label{eq-definition-E-frequencies}
            \cE := \{0,\omega_1,-\omega_1,\ldots,\omega_r,-\omega_r\},
        \end{equation}
        with repeated values listed only once. Applying \eqref{eq-decomposition-rotation-matrix} to every block gives
        \begin{equation} \label{eq-decomposition-exp-of-t-boldface-B}
            \me^{t \mathbf{B}} = \sum_{\nu \in \cE} \me^{\mi \nu t} \mathbf{E}_\nu,
        \end{equation}
        where the $ \mathbf{E}_\nu $ are constant matrices. Set $ \mathbf{E}_0:=\diag(\1_m,0,\ldots,0) $. For $ \nu>0 $, the $ j $-th two-dimensional block of $ \mathbf{E}_\nu $ is $ \frac{1}{2} \begin{pmatrix}
            1 & -\mi \\
            \mi & 1
        \end{pmatrix} $ if $ \omega_j = \nu $ and is zero otherwise. For $ \nu < 0 $, the $ j $-th two-dimensional block of $ \mathbf{E}_\nu $ is $ \frac{1}{2} \begin{pmatrix}
            1 & \mi \\ 
            -\mi & 1
        \end{pmatrix} $ if $ \omega_j=-\nu $, and is zero otherwise. All other blocks are zero. These matrices depend only on $ \mathbf{B} $, not on $ x $ or $ t $.
        
        Substituting \eqref{eq-decomposition-exp-of-t-boldface-B} into \eqref{eq-definition-psi-boldface-P} gives
        \begin{align*}
            \psi_\mathbf{P} (\me^{t \mathbf{B}} x) &= - \frac{1}{2} ( \mathbf{P} \me^{t \mathbf{B}} x )^T \mathbf{C} ( \mathbf{P} \me^{t \mathbf{B}} x ) + \mi \boldsymbol{\zeta}^T \mathbf{P} \me^{t \mathbf{B}} x \\ 
            &= -\frac{1}{2} \sum_{\sigma, \nu \in \cE} \me^{\mi (\sigma + \nu) t} (\mathbf{P} \mathbf{E}_\sigma x)^T \mathbf{C} (\mathbf{P} \mathbf{E}_\nu x) + \mi \sum_{\sigma \in \cE} \me^{\mi \sigma  t} \boldsymbol{\zeta}^T \mathbf{P} \mathbf{E}_\sigma x.
        \end{align*}
        The above formula displays all the dependence on time. Each quadratic term has frequency $ \sigma + \nu $; each linear term has frequency $ \sigma $. Since $ 0 \in \cE $, every linear frequency also belongs to $ \cE + \cE := \{ \sigma + \nu : \sigma, \nu \in \cE \} $. For each $ \mu \in \cE + \cE $, define 
        \begin{equation} \label{eq-definition-psi-mu}
            \psi_\mu (x) := -\frac{1}{2} \sum_{ \substack{ \sigma, \nu \in \cE \\ \sigma + \nu = \mu } } ( \mathbf{P} \mathbf{E}_{\sigma} x )^T \mathbf{C} ( \mathbf{P} \mathbf{E}_\nu x ) + \mi \boldsymbol{\zeta}^T \mathbf{P} \mathbf{E}_\mu x,
        \end{equation}
        where $ \mathbf{E}_\mu:=0 $ when $ \mu\notin\cE $. This is precisely the coefficient multiplying $ \me^{\mi\mu t} $ in \eqref{eq-decomposition-of-action-of-psi-P-on-e-t-B-x}. The first term is quadratic in $ x $ and the second is linear, so $ \psi_\mu $ has degree at most two. Consequently, setting
        \begin{equation} \label{eq-definition-calligraphic-F-double-index-set}
            \cF := (\cE + \cE) \setminus \{ 0 \},
        \end{equation}
        we obtain \eqref{eq-decomposition-of-action-of-psi-P-on-e-t-B-x}.
    \end{proof}

    \begin{lemma}
        For the polynomials $ \psi_\mu $ defined in \eqref{eq-definition-psi-mu},
        \begin{equation} \label{eq-drift-operator-on-psi-mu}
            D \psi_\mu (x) = \mi \mu \psi_\mu (x).
        \end{equation}
        where $ D $ is defined in \eqref{eq-definition-D-as-B-times-nabla}.
    \end{lemma}
    \begin{proof}
        We first verify that the matrices $ \mathbf{E}_\nu $ satisfy
        \begin{equation} \label{eq-product-of-E-nu-B}
            \mathbf{E}_\nu \mathbf{B} = \mi \nu \mathbf{E}_\nu, \quad \nu \in \cE.
        \end{equation}
        For $ \nu > 0 $, consider a rotation block with $ \omega_j = \nu $. On this block, direct multiplication gives 
        \begin{equation*}
            \frac{1}{2} \begin{pmatrix}
                1 & - \mi \\ 
                \mi & 1
            \end{pmatrix}
            \begin{pmatrix}
                0 & \nu \\
                -\nu & 0
            \end{pmatrix} = \frac{1}{2} \begin{pmatrix}
                \mi \nu & \nu \\
                - \nu & \mi \nu 
            \end{pmatrix} = \mi \nu \frac{1}{2} \begin{pmatrix}
                1 & - \mi \\ 
                \mi & 1
            \end{pmatrix}.
        \end{equation*}
        The same calculation applies to every block with $ \omega_j=\nu $. On blocks with $ \omega_j\neq\nu $, and on the zero block, $ \mathbf{E}_\nu $ vanishes. For negative frequencies, use $ \mathbf{E}_{-\nu}=\overline{\mathbf{E}_\nu} $ and the fact that $ \mathbf{B} $ is real:
        \begin{equation*}
            \mathbf{E}_{\nu} \mathbf{B} = \overline{ \mathbf{E}_{-\nu} \mathbf{B} } = \overline{ - \mi \nu \mathbf{E}_{-\nu} } = \mi \nu \mathbf{E}_\nu.
        \end{equation*}
        Finally, $ \mathbf{E}_0\mathbf{B}=0 $ because $ \mathbf{E}_0 $ comes from the zero block of $ \mathbf{B} $. Thus \eqref{eq-product-of-E-nu-B} holds for all $ \nu\in\cE $. It follows that $ \mathbf{E}_\nu\mathbf{B}^k=(\mi\nu)^k\mathbf{E}_\nu $, and hence
        \begin{equation*}
            \mathbf{E}_\nu \me^{t \mathbf{B}} = \sum_{k=0}^\infty \frac{t^k}{k!} \mathbf{E}_\nu \mathbf{B}^k = \sum_{k=0}^\infty \frac{(\mi \nu t)^k}{k!} \mathbf{E}_\nu = \me^{\mi \nu t} \mathbf{E}_\nu.
        \end{equation*}
        Substituting this identity into \eqref{eq-definition-psi-mu} gives
        \begin{align*}
            \psi_\mu (\me^{t \mathbf{B}} x) &= -\frac{1}{2} \sum_{ \substack{ \sigma, \nu \in \cE \\ \sigma + \nu = \mu } } ( \mathbf{P} \mathbf{E}_\sigma \me^{t \mathbf{B}} x )^T \mathbf{C} ( \mathbf{P} \mathbf{E}_\nu \me^{t \mathbf{B}} x) + \mi \boldsymbol{\zeta}^T \mathbf{P} \mathbf{E}_\mu \me^{t \mathbf{B}} x \\ 
            &= -\frac{1}{2} \sum_{ \substack{ \sigma, \nu \in \cE \\ \sigma + \nu = \mu } } \me^{\mi (\sigma + \nu) t} ( \mathbf{P} \mathbf{E}_\sigma x )^T \mathbf{C} ( \mathbf{P} \mathbf{E}_\nu x ) + \mi \me^{\mi \mu t} \boldsymbol{\zeta}^T \mathbf{P} \mathbf{E}_\mu x = \me^{\mi \mu t} \psi_\mu (x).
        \end{align*} 
        Therefore,
        \begin{equation*}
            D \psi_\mu (x) = \left. \frac{\md}{\md t} \psi_\mu ( \me^{t \mathbf{B}} x ) \right\vert_{t = 0} = \left. \frac{\md}{\md t} \me^{\mi \mu t} \psi_\mu (x) \right\vert_{t = 0}  = \mi \mu \psi_\mu (x),
        \end{equation*}
        which proves \eqref{eq-drift-operator-on-psi-mu}.
    \end{proof}

    Thus our finite-frequency decomposition is a decomposition into eigenvectors of the drift operator $ D $ acting on polynomials of degree at most two.

    \begin{lemma} \label{lemma-real-and-imaginary-parts-of-psi-0}
        Define 
        \begin{equation} \label{eq-definition-q-ell}
            q (x) := - \Re \psi_0 (x), \quad \ell (x) := \Im \psi_0 (x).
        \end{equation}
        Then $ q $ is a positive-definite real quadratic form and $ \ell $ is a real linear form. If the zero-block coordinates of $ x $ vanish, then $ \ell(x)=0 $.
    \end{lemma}
    \begin{proof}
        Setting $ \mu=0 $ in \eqref{eq-definition-psi-mu} gives
        \begin{equation*}
            \psi_0 (x) = - \frac{1}{2} \sum_{\nu \in \cE} ( \mathbf{P} \mathbf{E}_{- \nu} x )^T \mathbf{C} ( \mathbf{P} \mathbf{E}_\nu x ) + \mi \boldsymbol{\zeta}^T \mathbf{P} \mathbf{E}_0 x = - \frac{1}{2} \sum_{\nu \in \cE} \abs{ \mathbf{C}^{1/2} \mathbf{P} \mathbf{E}_\nu x }^2 + \mi \boldsymbol{\zeta}^T \mathbf{P} \mathbf{E}_0 x.
        \end{equation*}
        Therefore,
        \begin{equation} \label{eq-equivalent-form-q}
            q(x) = \frac{1}{2} x^T \left( \sum_{\nu \in \cE} \Re( \mathbf{E}_\nu^* \mathbf{P}^T \mathbf{C} \mathbf{P} \mathbf{E}_\nu ) \right) x = \frac{1}{2} \sum_{\nu \in \cE} \abs{ \mathbf{C}^{1/2} \mathbf{P} \mathbf{E}_\nu x }^2.
        \end{equation}
        Each matrix $ \mathbf{E}_\nu^*\mathbf{P}^T\mathbf{C}\mathbf{P}\mathbf{E}_\nu $ is Hermitian, so its entrywise real part is a real symmetric matrix. The first identity in \eqref{eq-equivalent-form-q} shows that $ q $ is a quadratic form, and the second shows that it is positive semidefinite.

        Suppose that $ q(x)=0 $ for some $ x\in\R^{2d} $. By irreducibility, $ \mathbf{S}_1=\int_0^1\me^{s\mathbf{Z}^T}\mathbf{C}\me^{s\mathbf{Z}}\,\md s>0 $. Every summand in the last expression of \eqref{eq-equivalent-form-q} is nonnegative, so each must vanish:
        \begin{equation*}
            \mathbf{C}^{1/2} \mathbf{P} \mathbf{E}_\nu x = 0, \quad \forall \nu \in \cE.
        \end{equation*}
        By \eqref{eq-decomposition-exp-of-t-boldface-B} we have 
        \begin{equation*}
            \mathbf{C}^{1/2} \me^{t \mathbf{Z}} \mathbf{P} x = \mathbf{C}^{1/2} \mathbf{P} \me^{t \mathbf{B}} x = \sum_{\nu \in \cE} \me^{\mi \nu t} \mathbf{C}^{1/2} \mathbf{P} \mathbf{E}_\nu x = 0, \quad t \ge 0.
        \end{equation*}
        Consequently, 
        \begin{equation*}
            0 = \int_0^1 \abs{ \mathbf{C}^{1/2} \me^{t \mathbf{Z}} \mathbf{P} x }^2 \md t = ( \mathbf{P} x )^T \left( \int_0^1 \me^{t \mathbf{Z}^T} \mathbf{C} \me^{t \mathbf{Z}} \,\md t \right) (\mathbf{P} x) =  (\mathbf{P} x)^T \mathbf{S}_1 \mathbf{P} x.
        \end{equation*}
        Positive definiteness of $ \mathbf{S}_1 $ implies $ \mathbf{P}x=0 $. Since $ \mathbf{P} $ is invertible, $ x=0 $. Therefore $ q $ is positive definite.

        Moreover,
        \begin{equation*}
            \ell (x) = \boldsymbol{\zeta}^T \mathbf{P} \mathbf{E}_0 x,
        \end{equation*}
        so $ \ell $ is a real linear form. Since $ \mathbf{E}_0=\diag(\1_m,0,\ldots,0) $, the equality $ \ell(x)=0 $ follows whenever the zero-block coordinates of $ x $ vanish.
    \end{proof}

    \begin{lemma}
        Define the polynomial 
        \begin{equation} \label{eq-definition-F-psi}
            F_\psi (x) := \sum_{\mu \in \cF} \frac{\psi_\mu (x)}{\mi \mu}.
        \end{equation}
        Then, for any $ Y \in \cB_1 (\mathsf{h}) $ we have 
        \begin{equation} \label{eq-characteristic-function-predual-semigroup-purely-critical-diagonalizable-case}
            \chi_{\cT_{*t} (Y)} ( \iota^{-1} (\mathbf{P} x) ) = \exp{t \psi_0 (x) + F_\psi ( \me^{t \mathbf{B}} x ) - F_\psi (x)} \chi_Y ( \iota^{-1} (\mathbf{P} \me^{t \mathbf{B}}  x )).
        \end{equation}
    \end{lemma}
    \begin{proof}
        Firstly note that, taking $ t = 0 $ in \eqref{eq-decomposition-of-action-of-psi-P-on-e-t-B-x} we obtain 
        \begin{equation*}
            \psi_{\mathbf{P}} (x) = \psi_0 (x) + \sum_{\mu \in \cF} \psi_{\mu} (x).
        \end{equation*}
        By \eqref{eq-drift-operator-on-psi-mu},
        \begin{equation*}
            D F_\psi (x) = \sum_{\mu \in \cF} \frac{D \psi_\mu (x)}{\mi \mu} = \sum_{\mu \in \cF} \psi_\mu (x) = \psi_\mathbf{P} (x) - \psi_0 (x).
        \end{equation*}
        Consequently, the chain rule gives
        \begin{equation} \label{eq-derivative-e-tB-x-of-F-psi}
            \frac{\md}{\md s} F_\psi (\me^{s \mathbf{B}} x) = \sum_{j=1}^{2d} \frac{\partial F_\psi }{\partial x_j } (\me^{s \mathbf{B}} x) ( \mathbf{B} \me^{s \mathbf{B}} x )_j = ( D F_\psi ) (\me^{s \mathbf{B} } x) = \psi_\mathbf{P} (\me^{s \mathbf{B}} x) - \psi_0 (\me^{s \mathbf{B}} x).
        \end{equation}
        On the other hand, \eqref{eq-drift-operator-on-psi-mu} gives
        \begin{equation*}
            \frac{\md}{\md s} \psi_0 (\me^{s \mathbf{B}} x) = (D \psi_0) (\me^{s \mathbf{B}} x) = 0,
        \end{equation*}
        Thus $ s\mapsto\psi_0(\me^{s\mathbf{B}}x) $ is constant, and hence
        \begin{equation} \label{eq-psi-0-constant}
            \psi_0 (\me^{s \mathbf{B}} x) = \psi_0 (x), \quad \forall s \in \R.
        \end{equation}
        Substituting this identity into \eqref{eq-derivative-e-tB-x-of-F-psi} gives
        \begin{equation*}
            \frac{\md}{\md s} F_\psi ( \me^{s \mathbf{B}} x ) = \psi_\mathbf{P} ( \me^{s \mathbf{B}} x ) - \psi_0 (x).
        \end{equation*}
        The fundamental theorem of calculus now gives
        \begin{equation*}
            F_\psi ( \me^{t \mathbf{B}} x ) - F_\psi (x) = \int_0^t \frac{\md}{\md s} F_\psi (\me^{s \mathbf{B}} x) \md s = \int_0^t \psi_\mathbf{P} (\me^{s \mathbf{B}} x) \md s - t \psi_0 (x),
        \end{equation*}
        and rearrangement yields
        \begin{equation} \label{eq-integral-psi-P}
            \int_0^t \psi_\mathbf{P} (\me^{s \mathbf{B}} x) \md s = t \psi_0 (x) + F_\psi ( \me^{t \mathbf{B}} x ) - F_\psi (x).
        \end{equation}
        For $ Y\in\cB_1(\mathsf{h}) $, \eqref{eq-definition-x-as-inverse-of-z} gives
        \begin{equation} \label{eq-characteristic-function-of-predual-semigroup-purely-critical-diagonalizable-case}
            \chi_{\cT_{*t} (Y) } (\iota^{-1}(\mathbf{P} x)) = \chi_{\cT_{*t} (Y) } (z) = \phi_t (z) \Tr(Y W(\me^{t Z} z)) = \phi_t (z) \chi_Y (\me^{t Z} z),
        \end{equation}
        where $ \phi_t $ is defined in \eqref{eq-definition-phi-t}. By \eqref{eq-definition-x-as-inverse-of-z} and \eqref{eq-definition-psi-boldface-P}, we may write
        \begin{align*}
            \phi_t (z) &= \exp{ -\frac{1}{2} \int_0^t (\me^{s \mathbf{Z}} \mathbf{z})^T \mathbf{C} (\me^{s \mathbf{Z}} \mathbf{z}) \md s + \mi \int_0^t \boldsymbol{\zeta}^T \me^{s \mathbf{Z}} \mathbf{z} \md s } \\ 
            &= \exp{ -\frac{1}{2} \int_0^t (\me^{s \mathbf{Z}} \mathbf{P} x )^T \mathbf{C} ( \me^{s \mathbf{Z}} \mathbf{P} x ) \md s + \mi \int_0^t \boldsymbol{\zeta}^T \me^{s \mathbf{Z}} \mathbf{P} x \md s } \\ 
            &= \exp{ -\frac{1}{2} \int_0^t (\mathbf{P} \me^{s \mathbf{B}} x)^T \mathbf{C} (\mathbf{P} \me^{s \mathbf{B}} x) \md s + \mi \int_0^t \boldsymbol{\zeta}^T \mathbf{P} \me^{s \mathbf{B}} x \md s } \\
            &= \exp{\int_0^t \psi_\mathbf{P} ( \me^{s \mathbf{B}} x ) \md s}.
        \end{align*}
        Moreover,
        \begin{equation*}
            \me^{t Z} z = \iota^{-1} (\me^{t \mathbf{Z}} \mathbf{z}) = \iota^{-1} ( \mathbf{P} \me^{t \mathbf{B}}  x),
        \end{equation*}
        Combining this identity with \eqref{eq-integral-psi-P} and \eqref{eq-characteristic-function-of-predual-semigroup-purely-critical-diagonalizable-case} yields
        \begin{equation*}
            \chi_{\cT_{*t} (Y)} ( \iota^{-1} (\mathbf{P} x) ) = \exp{t \psi_0 (x) + F_\psi ( \me^{t \mathbf{B}} x ) - F_\psi (x)} \chi_Y ( \iota^{-1} (\mathbf{P} \me^{t \mathbf{B}}  x )),
        \end{equation*}
        which completes the proof of \eqref{eq-characteristic-function-predual-semigroup-purely-critical-diagonalizable-case}.
    \end{proof}

    \begin{lemma} \label{lemma-global-schwartz-trace-norm-estimate}
        For every $ g\in\cS(\R^{2d}) $, there is a unique Schwartz operator
        $ A_g $ such that
        \begin{equation*}
            \chi_{A_g}(\iota^{-1}(\mathbf{z}))=g(\mathbf{z}),
            \qquad \mathbf{z}\in\R^{2d}.
        \end{equation*}
        For $ r\in\N $, define
        \begin{equation*}
            p_r(g):=\max_{\substack{\abs{\alpha}\le r\\
                                   \abs{\beta}\le r}}
            \sup_{\mathbf{z}\in\R^{2d}}
            \abs{\mathbf{z}^\alpha\partial^\beta g(\mathbf{z})}.
        \end{equation*}
        There are constants $ C_0>0 $ and $ r_0\in\N $, depending only on
        the dimension and the fixed Fourier--Weyl convention, such that
        \begin{equation} \label{eq-global-schwartz-trace-norm-estimate}
            \norm{A_g}_1\le C_0p_{r_0}(g),
            \qquad g\in\cS(\R^{2d}).
        \end{equation}
    \end{lemma}
    \begin{proof}
        Existence and uniqueness follow from
        \cite[Proposition~3.18]{KeylKiukasWerner2016}.  The fixed invertible
        real-linear change of phase-space variables and normalization
        relating the Weyl convention in that reference to ours preserve the
        Schwartz topology, so the proposition applies equally to the present
        convention.

        These seminorms form an increasing defining family for the Schwartz
        topology: the index set in $ p_r $ is finite, is contained in the one
        for $ p_{r+1} $, and every usual seminorm
        $ g\mapsto\sup_{\mathbf{z}}\abs{\mathbf{z}^\alpha
        \partial^\beta g(\mathbf{z})} $ occurs in $ p_r $ for sufficiently large
        $ r $.

        Recall that $ N $ is the total number operator and fix
        an integer $ m>d $.  Since the eigenspace of
        $ N $ associated with $ n $ has dimension
        $ \binom{n+d-1}{d-1} $,
        \begin{equation*}
            c_m:=\Tr((\1+N)^{-m})
            =\sum_{n=0}^\infty
              \binom{n+d-1}{d-1}(1+n)^{-m}<\infty.
        \end{equation*}
        If $ T $ is a Schwartz operator, then
        $ (\1+N)^mT $ is bounded and
        \begin{equation} \label{eq-oscillator-trace-norm-bound}
            \norm{T}_1
            =\norm{(\1+N)^{-m}(\1+N)^mT}_1
            \le c_m\norm{(\1+N)^mT}.
        \end{equation}
        This is the oscillator factorization used in
        \cite[proof of Lemma~3.6]{KeylKiukasWerner2016}.  Moreover,
        \begin{equation*}
            q_m(T):=\norm{(\1+N)^mT}
        \end{equation*}
        is a continuous seminorm on the Schwartz operators: indeed,
        $ (\1+N)^m $ is a polynomial in the canonical position and
        momentum operators, so this is a particular case of
        \cite[Proposition~3.4]{KeylKiukasWerner2016}.  By Proposition~3.18,
        the seminorm $ g\mapsto q_m(A_g) $ is therefore continuous on
        $ \cS(\R^{2d}) $.

        We spell out why this continuity gives one finite maximum of the
        seminorms above.  Continuity at zero gives finitely many indices
        $ r_1,\ldots,r_k $ and numbers $ \varepsilon_1,\ldots,
        \varepsilon_k>0 $ such that
        \begin{equation*}
            p_{r_j}(g)<\varepsilon_j\quad(1\le j\le k)
            \quad\Longrightarrow\quad q_m(A_g)<1.
        \end{equation*}
        Put $ r_0:=\max_j r_j $ and
        $ \varepsilon:=\min_j\varepsilon_j $.  Since the family $ (p_r)_r $ is
        increasing, $ p_{r_0}(g)<\varepsilon $ implies all the preceding
        inequalities.  If $ p_{r_0}(g)>0 $, apply this implication to
        $ \varepsilon g/(2p_{r_0}(g)) $ and use homogeneity; if
        $ p_{r_0}(g)=0 $, then $ g=0 $.  It follows that
        \begin{equation*}
            q_m(A_g)\le\frac{2}{\varepsilon}p_{r_0}(g),
            \qquad g\in\cS(\R^{2d}).
        \end{equation*}
        Combining this with \eqref{eq-oscillator-trace-norm-bound} proves
        \eqref{eq-global-schwartz-trace-norm-estimate}, with
        $ C_0:=2c_m/\varepsilon $. The constants are independent of the
        Gaussian semigroup.
    \end{proof}

    \begin{lemma} \label{lemma-test-operator-constructed-from-smooth-function-with-compact-support}
        For every $ f\in\cC_c^\infty(\R^{2d}) $, there is a unique trace-class operator $ Y_f $ satisfying
        \begin{equation} \label{eq-characteristic-function-of-specific-Y-f}
            \chi_{Y_f} ( \iota^{-1} ( \mathbf{P} x) ) = \me^{- F_\psi (x)} f(x).
        \end{equation}
        For every compact set $ E\subset\R^{2d} $, there are constants $ C_E>0 $ and $ r_E\in\N $ such that
        \begin{equation} \label{eq-upper-bound-L1-norm-of-Y-f}
            \norm{Y_f}_1 \le C_E \max_{\abs{\alpha} \le r_E} \norm{ \partial^\alpha f }_\infty, \quad \supp f \subset E,
        \end{equation}
        where $ \alpha $ is a multi-index of order at most $ r_E $. Furthermore,
        \begin{equation} \label{eq-action-of-predual-generator-on-Y-f}
            Y_f \in \dom \cL_*, \quad \cL_* ( Y_f ) = Y_{D f + \psi_0 f}.
        \end{equation}
    \end{lemma}

    \begin{proof}
        We apply Lemma \ref{lemma-global-schwartz-trace-norm-estimate} to
        \begin{equation*}
            G_f(\mathbf{z})
            :=\me^{-F_\psi(\mathbf{P}^{-1}\mathbf{z})}
              f(\mathbf{P}^{-1}\mathbf{z}).
        \end{equation*}
        This function is smooth and compactly supported, hence Schwartz.
        Define $ Y_f:=A_{G_f} $.  Then
        \eqref{eq-characteristic-function-of-specific-Y-f} holds, and
        uniqueness follows from injectivity of the Fourier--Weyl transform.

        Suppose that $ \supp f\subset E $, where $ E\subset\R^{2d} $ is
        compact.  Then $ \supp G_f\subset\mathbf{P} E $.  For
        $ \abs{\beta}\le r_0 $, the chain and Leibniz rules express
        $ \partial^\beta G_f(\mathbf{z}) $ as a finite linear combination of
        terms
        \begin{equation*}
            c_{\beta,\gamma}(\mathbf{z})
            (\partial^\gamma f)(\mathbf{P}^{-1}\mathbf{z}),
            \qquad \abs{\gamma}\le\abs{\beta},
        \end{equation*}
        where the coefficients depend only on $ \mathbf{P} $ and on derivatives
        of the fixed smooth function
        $ \me^{-F_\psi(\mathbf{P}^{-1}\cdot)} $.  These coefficients and the
        monomials $ \mathbf{z}^\alpha $, $ \abs{\alpha}\le r_0 $, are bounded on
        the compact set $ \mathbf{P} E $.  Consequently,
        \begin{equation*}
            p_{r_0}(G_f)
            \le C_E'
            \max_{\abs{\gamma}\le r_0}\norm{\partial^\gamma f}_\infty,
            \qquad \supp f\subset E.
        \end{equation*}
        Together with \eqref{eq-global-schwartz-trace-norm-estimate}, this
        proves \eqref{eq-upper-bound-L1-norm-of-Y-f}, with
        $ r_E=r_0 $ and $ C_E=C_0C_E' $.

        \begin{samepage}
        Substituting \eqref{eq-characteristic-function-of-specific-Y-f}
        into
        \eqref{eq-characteristic-function-predual-semigroup-purely-critical-diagonalizable-case}
        yields
        \begin{align}
            \chi_{\cT_{*t}(Y_f)}(\iota^{-1}(\mathbf{P} x))
            &=\exp{t\psi_0(x)+F_\psi(\me^{t\mathbf{B}}x)-F_\psi(x)}
              \chi_{Y_f}(\iota^{-1}(\mathbf{P}\me^{t\mathbf{B}}x))
              \nonumber\\
            &=\exp{t\psi_0(x)-F_\psi(x)}f(\me^{t\mathbf{B}}x)
              \nonumber\\
            &=\chi_{Y_{\exp{t\psi_0}f\circ\exp{t\mathbf{B}}}}
              (\iota^{-1}(\mathbf{P} x)).
              \label{eq-characteristic-function-predual-semigroup-on-Y-f}
        \end{align}
        \end{samepage}
        Injectivity of the Fourier--Weyl transform gives
        \begin{equation*}
            \cT_{*t}(Y_f)=Y_{\me^{t\psi_0}f\circ\me^{t\mathbf{B}}}.
        \end{equation*}

        Set
        \begin{equation} \label{eq-definition-h-as-sum-of-psi-0-and-D}
            h(x):=\psi_0(x)f(x)+(Df)(x).
        \end{equation}
        Since $ f\in\cC_c^\infty(\R^{2d}) $, we also have
        $ h\in\cC_c^\infty(\R^{2d}) $, so $ Y_f,Y_h\in\cB_1(\mathsf{h}) $.
        Differentiating
        \eqref{eq-characteristic-function-predual-semigroup-on-Y-f} for
        fixed $ x $, and using \eqref{eq-psi-0-constant}, gives
        \begin{align*}
            \frac{\md}{\md t}
            \chi_{\cT_{*t}(Y_f)}(\iota^{-1}(\mathbf{P} x))
            &=\me^{-F_\psi(x)+t\psi_0(x)}
              h(\me^{t\mathbf{B}}x)\\
            &=\chi_{\cT_{*t}(Y_h)}(\iota^{-1}(\mathbf{P} x)).
        \end{align*}
        Integration yields
        \begin{equation} \label{eq-characteristic-function-of-predual-semigroup-minus-identity-on-Y-f}
            \chi_{\cT_{*t}(Y_f)-Y_f}(\iota^{-1}(\mathbf{P} x))
            =\int_0^t
              \chi_{\cT_{*s}(Y_h)}(\iota^{-1}(\mathbf{P} x))\,\md s.
        \end{equation}
        The map $ s\mapsto\cT_{*s}(Y_h) $ is continuous in trace norm, so
        its Bochner integral exists in $ \cB_1(\mathsf{h}) $.  Since, for fixed
        $ x $, the functional
        $ Y\mapsto\chi_Y(\iota^{-1}(\mathbf{P} x)) $ has norm at most one, it
        commutes with this integral.  Hence
        \eqref{eq-characteristic-function-of-predual-semigroup-minus-identity-on-Y-f}
        and injectivity of the Fourier--Weyl transform imply
        \begin{equation*}
            \cT_{*t}(Y_f)-Y_f
            =\int_0^t\cT_{*s}(Y_h)\,\md s.
        \end{equation*}
        It follows that
        \begin{equation*}
            \norm{\frac{\cT_{*t}(Y_f)-Y_f}{t}-Y_h}_1
            \le
            \sup_{0\le s\le t}\norm{\cT_{*s}(Y_h)-Y_h}_1
            \longrightarrow0
        \end{equation*}
        as $ t\downarrow0 $.  Thus $ Y_f\in\dom\cL_* $ and
        $ \cL_*(Y_f)=Y_h=Y_{Df+\psi_0f} $, proving
        \eqref{eq-action-of-predual-generator-on-Y-f}.
    \end{proof}

    \begin{proof}[Proof of Theorem \ref{theorem-spectrum-purely-critical-semisimple-drift}]
        We first prove the absence of eigenvalues. Suppose that $ 0\neq Y\in\dom\cL_* $ and $ \cL_*Y=\eta Y $ for some $ \eta\in\C $. Then $ \cT_{*t}(Y)=\me^{\eta t}Y $. Define the continuous function
        \begin{equation*}
            f_Y (x) := \me^{F_\psi (x)} \chi_{Y} ( \iota^{-1} (\mathbf{P} x) ),
        \end{equation*}
        where $ F_\psi $ is defined in \eqref{eq-definition-F-psi}. Using the eigenvector relation together with \eqref{eq-characteristic-function-predual-semigroup-purely-critical-diagonalizable-case}, we obtain
        \begin{align*}
            f_Y (x) &= \me^{F_\psi (x)} \me^{-\eta t} \chi_{\cT_{*t} (Y)} (\iota^{-1} (\mathbf{P} x) ) \\ 
            &= \me^{F_\psi (x)} \me^{-\eta t} \me^{t \psi_0 (x) + F_\psi (\me^{t \mathbf{B}} x) - F_\psi (x)} \chi_Y ( \iota^{-1} (\mathbf{P} \me^{t \mathbf{B}} x) ) \\ 
            &= \me^{t (\psi_0 (x) - \eta)} \me^{F_\psi (\me^{t \mathbf{B}} x)} \chi_{Y} ( \iota^{-1} (\mathbf{P} \me^{t \mathbf{B}} x) ) = \me^{t (\psi_0 (x) - \eta)} f_Y (\me^{t \mathbf{B}} x), 
        \end{align*}
        which implies 
        \begin{equation} \label{eq-relation-f-Y-exp-t-B-x-and-f-Y-x}
            f_Y (\me^{t \mathbf{B}} x) = \me^{t (\eta - \psi_0 (x))} f_Y (x), \quad t \ge 0.
        \end{equation}
        This identity also holds for negative times. Indeed, applying \eqref{eq-relation-f-Y-exp-t-B-x-and-f-Y-x} at $ \me^{-t\mathbf{B}}x $ and using \eqref{eq-psi-0-constant} gives
        \begin{equation} \label{eq-relation-f-Y-exp-t-B-x-and-f-Y-x-negative-time}
            f_Y (\me^{-t \mathbf{B}} x) = \me^{t ( \psi_0 (x) - \eta)} f_Y ( x ), \quad t \ge 0.
        \end{equation}
        The block form \eqref{eq-block-diagonal-representation-of-drift-matrix-purely-critical-case} gives $ \mathbf{B}^T=-\mathbf{B} $, so $ \me^{t\mathbf{B}} $ is orthogonal. Thus, for every $ t\in\R $, the point $ \me^{t\mathbf{B}}x $ lies on the compact sphere
        \begin{equation*}
            S_x := \{ u \in \R^{2d} : \abs{u} = \abs{x} \}.
        \end{equation*}
        Since $ f_Y $ is continuous, it is bounded on this sphere. Therefore,
        \begin{equation} \label{eq-uniform-bound-fY-orbit}
            \sup_{t \in \R} \abs{ f_Y (\me^{t \mathbf{B}} x) } \le \max_{u \in S_x} \abs{f_Y (u)} < \infty.
        \end{equation}
        Taking absolute values in \eqref{eq-relation-f-Y-exp-t-B-x-and-f-Y-x} and \eqref{eq-relation-f-Y-exp-t-B-x-and-f-Y-x-negative-time} gives
        \begin{equation*}
            \abs{ f_Y (\me^{t \mathbf{B}} x) } = \me^{ t \Re \eta - t \Re \psi_0 (x) }  \abs{f_Y (x)} , \quad t \in \R.
        \end{equation*}
        Boundedness for both positive and negative times forces $ \Re\eta=\Re\psi_0(x)=-q(x) $ whenever $ f_Y(x)\neq0 $. If $ f_Y(x_*)\neq0 $ at some point, continuity gives a nonempty open neighborhood $ \cO $ of $ x_* $ on which $ f_Y $ does not vanish. The positive-definite quadratic form $ q $ would then be constant on $ \cO $, which is impossible. Hence $ f_Y=0 $ everywhere. Since the exponential factor in the definition of $ f_Y $ never vanishes, $ \chi_Y=0 $; injectivity of the Fourier--Weyl transform gives $ Y=0 $, a contradiction. Thus $ \cL_* $ has no eigenvalues.

        We then show that 
        \begin{equation} \label{eq-subset-of-approximate-spectrum}
            (-\infty, 0) + \mi \Xi_{\mathbf{Z}} \subset \sigma_{\mathrm{ap}} (\cL_*),
        \end{equation}
        where $ \Xi_{\mathbf{Z}} $ is defined in \eqref{eq-critical-frequency-group}. Fix $ k=(k_1,\ldots,k_r)\in\Z^r $, where $ r $ is the number of rotation blocks, and put
        \begin{equation} \label{eq-definition-nu-as-sum-of-frequencies}
            \nu := \sum_{j=1}^r k_j \omega_j,
        \end{equation}
        Choose $ \eta\in\C $ with $ \Re\eta<0 $ and
        \begin{equation} \label{eq-relation-eta-and-nu}
            \Im \eta = \nu.
        \end{equation}
        Define the polynomial
        \begin{equation} \label{eq-definition-polynomial-p-k}
            p_k (x) = \prod_{k_j \ge 0} (x_{j,1} - \mi x_{j, 2})^{k_j} \prod_{k_j < 0} (x_{j, 1} + \mi x_{j, 2})^{- k_j},
        \end{equation}
        where $ x_{j,1} $ and $ x_{j,2} $ are defined in \eqref{eq-sub-entries-of-x-j}. The block-diagonal form \eqref{eq-block-diagonal-representation-of-drift-matrix-purely-critical-case} gives
        \begin{equation*}
            p_k (\me^{t \mathbf{B}} x) =  \prod_{j=1}^r \me^{\mi k_j \omega_j t} p_k (x) = \me^{\mi t \sum_{j=1}^r k_j \omega_j} p_k(x) = \me^{\mi \nu t} p_k (x),
        \end{equation*}
        and hence
        \begin{equation} \label{eq-action-of-drift-operator-on-p-k}
            D p_k (x) = \mi \nu p_k (x),
        \end{equation}
        where $ D $ is defined in \eqref{eq-definition-D-as-B-times-nabla}.

        Suppose that $ \eta\notin\sigma_{\mathrm{ap}}(\cL_*) $. By the definition of the approximate point spectrum, there is $ c_\eta>0 $ such that
        \begin{equation} \label{eq-c-eta-as-lower-bound}
            \norm{ (\cL_* - \eta) (Y) }_1 \ge c_\eta  \norm{Y}_1, \quad \forall Y \in \dom \cL_*.
        \end{equation}
        Choose
        \begin{equation*}
            0 < \epsilon_{\eta} < \min \{ c_\eta , 1 \},
        \end{equation*}
        and $ g_\eta\in\cC_c^\infty(\R) $ satisfying
        \begin{equation*}
            g_\eta (0) = 1, \quad \supp g_\eta\subset (-\epsilon^2_{\eta}, \epsilon^2_{\eta} ).
        \end{equation*}
        Set 
        \begin{equation} \label{eq-definition-h-eta}
            h_\eta (x) := p_k (x) g_\eta ( \abs{ \psi_0 (x) - \Re \eta }^2 ),
        \end{equation}
        By Lemma \ref{lemma-properties-of-function-h-eta}, $ h_\eta\neq0 $ and $ h_\eta\in\cC_c^\infty(\R^{2d}) $. For $ n\in\N $, equations \eqref{eq-drift-operator-on-psi-mu} and \eqref{eq-action-of-drift-operator-on-h-eta} give
        \begin{equation} \label{eq-action-drift-operator-on-power-and-h-eta}
            D \left( (\psi_0 - \Re \eta )^n h_\eta \right) = n (\psi_0 - \Re \eta )^{n-1} ( D \psi_0 ) h_\eta + (\psi_0 - \Re \eta )^n ( D h_\eta ) = \mi \nu (\psi_0 - \Re \eta)^n h_\eta.
        \end{equation}
        Since $ (\psi_0-\Re\eta)^nh_\eta $ is smooth and supported in $ \supp h_\eta $, it belongs to $ \cC_c^\infty(\R^{2d}) $. By \eqref{eq-action-of-predual-generator-on-Y-f}, \eqref{eq-relation-eta-and-nu}, and \eqref{eq-action-drift-operator-on-power-and-h-eta},
        \begin{align*}
            (\cL_* - \eta) Y_{ (\psi_0 - \Re \eta)^n h_\eta } &=  Y_{ D (\psi_0 - \Re \eta)^n h_\eta + \psi_0 (\psi_0 - \Re \eta)^n h_\eta } - \eta Y_{ (\psi_0 - \Re \eta)^n h_\eta} \\
            &= Y_{ (\mi \nu + \psi_0 - \eta) (\psi_0 - \Re \eta)^n h_\eta } = Y_{ (\psi_0  - \Re \eta)^{n + 1} h_\eta }.
        \end{align*}
        Applying \eqref{eq-c-eta-as-lower-bound} with $ Y=Y_{(\psi_0-\Re\eta)^nh_\eta} $ yields
        \begin{equation*}
            \norm{(\cL_* - \eta) Y_{ (\psi_0 - \Re \eta)^n h_\eta }}_1 = \norm{ Y_{ (\psi_0  - \Re \eta)^{n + 1} h_\eta } }_1 \ge c_\eta \norm{ Y_{ (\psi_0 - \Re \eta)^n h_\eta } }_1.
        \end{equation*}
        Consequently,
        \begin{equation} \label{eq-lower-bound-of-power-n-times-g-eta}
            \norm{ Y_{(\psi_0 - \Re \eta)^n h_\eta}  }_1 \ge c_\eta^n \norm{ Y_{h_\eta} }_1.
        \end{equation}

        To estimate the derivatives, set
        \begin{equation*}
            E_\eta := \supp h_\eta, \quad r_\eta := r_{E_\eta},
        \end{equation*}
        By \eqref{eq-absolute-value-difference-psi-0-and-Re-eta}, $ \abs{\psi_0-\Re\eta}<\epsilon_\eta $ wherever $ h_\eta\neq0 $. Since $ E_\eta $ is the closure of this set, continuity gives $ \abs{\psi_0-\Re\eta}\le\epsilon_\eta<1 $ on $ E_\eta $. Compactness of $ E_\eta $ allows us to choose $ M\ge1 $ such that
        \begin{equation*}
            \sup_{x \in E_\eta} \abs{ \partial^\gamma h_\eta (x) } \le M, \qquad \sup_{x \in E_\eta} \abs{ \partial^\gamma (\psi_0 - \Re \eta ) (x) } \le M, \quad \abs{\gamma}\le r_\eta.
        \end{equation*}

        Fix $ n\ge r_\eta $ and a multi-index $ \alpha\in\N^{2d} $ with $ \abs{\alpha}\le r_\eta $. The generalized Leibniz formula gives
        \begin{equation} \label{eq-derivatives-power-n-and-g-eta}
            \partial^\alpha \left( (\psi_0 - \Re \eta)^n h_\eta \right) = \sum_{\gamma_0 + \cdots + \gamma_n = \alpha} \frac{\alpha!}{\gamma_0! \cdots \gamma_n!} ( \partial^{\gamma_0} h_\eta) \prod_{j=1}^{n} \partial^{\gamma_j} (\psi_0 - \Re \eta).
        \end{equation}
        Here each $ \gamma_j $ is a multi-index, the sum constraint is componentwise, and $ \alpha!:=\prod_{l=1}^{2d}\alpha_l! $. Let $ s $ be the number of indices $ j\in\{1,\ldots,n\} $ for which $ \gamma_j\neq0 $. Each nonzero multi-index has order at least one, so
        \begin{equation*}
            s \le \sum_{j=1}^n \abs{\gamma_j} \le \abs{\alpha} \le r_\eta.
        \end{equation*}
        Exactly $ n-s $ factors in the product are therefore undifferentiated copies of $ \psi_0 - \Re \eta $. Consequently, for $ x \in E_\eta $, we have 
        \begin{equation} \label{eq-upperbound-products-of-derivatives}
            \abs{( \partial^{\gamma_0} h_\eta (x) ) \prod_{j=1}^{n} \partial^{\gamma_j} (\psi_0 - \Re \eta) (x) } \le M^{s+1} \epsilon_\eta^{n-s} \le M^{r_\eta+1} \epsilon_\eta^{n-r_\eta}.  
        \end{equation}
        The last inequality uses $ M \ge 1 $, $ 0 < \epsilon_\eta < 1 $ and $ s \le r_\eta $. 
        The multinomial theorem, applied in each coordinate, gives the exact identity 
        \begin{equation} \label{eq-application-of-multinomial-theorem}
            \sum_{\gamma_0 + \cdots + \gamma_n = \alpha} \frac{\alpha!}{\gamma_0! \cdots \gamma_n!} = \prod_{l=1}^{2d} (n + 1)^{\alpha_l} = (n + 1)^{\abs{\alpha}}.
        \end{equation}
        Substituting \eqref{eq-upperbound-products-of-derivatives} and \eqref{eq-application-of-multinomial-theorem} into \eqref{eq-derivatives-power-n-and-g-eta} yields
        \begin{equation*}
            \sup_{x \in E_\eta} \abs{ \partial^\alpha \left( ( \psi_0 - \Re \eta)^n h_\eta \right) (x) } \le M^{r_\eta + 1} (n+1)^{r_\eta} \, \epsilon_\eta^{n - r_\eta}, \quad n \ge r_\eta.
        \end{equation*} 
        For $ x \notin E_\eta $, the function $ h_\eta $ vanishes on a neighborhood of $ x $; hence every derivative $ \partial^{\gamma_0} h_\eta (x) $ is zero. Then, \eqref{eq-derivatives-power-n-and-g-eta} vanishes outside $ E_\eta $ as well. Taking the maximum over $ \abs{\alpha} \le r_\eta $ proves
        \begin{equation*}
            \max_{\abs{\alpha} \le r_\eta} \norm{ \partial^\alpha \left( (\psi_0 - \Re \eta)^n h_\eta \right) }_\infty \le M^{r_\eta +1} (n+1)^{r_\eta} \epsilon_\eta^{n - r_\eta}, \quad n \ge r_\eta.
        \end{equation*}
        Then, the fixed-support trace-norm estimate yields
        \begin{equation*}
            \norm{Y_{h_\eta}}_1 \le c^{-n}_\eta \norm{ Y_{ (\psi_0 - \Re \eta)^n h_\eta } }_1 \le C_{E_\eta} M^{r_\eta + 1} (n + 1)^{r_\eta} ( \epsilon_\eta / c_\eta )^{n - r_\eta} \rightarrow 0
        \end{equation*}
        as $ n\to\infty $, after absorbing the fixed factor $ c_\eta^{-r_\eta} $ into $ C_{E_\eta} $. This contradicts $ Y_{h_\eta}\neq0 $. Hence $ \eta\in\sigma_{\mathrm{ap}}(\cL_*) $, proving \eqref{eq-subset-of-approximate-spectrum}.

        Finally, nonperiodicity makes $ \Xi_{\mathbf{Z}} $ dense in $ \R $: otherwise this additive subgroup would be discrete, so all $ \omega_j $ would be integer multiples of one positive frequency, giving a common period. Since the approximate point spectrum is closed, \eqref{eq-subset-of-approximate-spectrum} implies
        \begin{equation*}
            \{ \eta \in \C : \Re \eta \le 0 \} \subset \sigma_{\mathrm{ap}} (\cL_*).
        \end{equation*}
        The reverse inclusion follows from contractivity:
        \begin{equation*}
            \sigma_{\mathrm{ap}} (\cL_*) \subset \sigma (\cL_*) \subset \{ \eta : \Re \eta \le 0 \}.
        \end{equation*}
        Consequently, 
        \begin{equation*}
            \sigma_{\mathrm{ap}} (\cL_*) = \{ \eta \in \C : \Re \eta \le 0 \},
        \end{equation*}
        which completes the proof of \eqref{eq-spectrum-purely-critical-semisimple-drift}.
        
    \end{proof}

    \begin{lemma} \label{lemma-properties-of-function-h-eta}
        The function $ h_\eta $ defined in \eqref{eq-definition-h-eta} is a nonzero compactly supported smooth function satisfying
        \begin{equation} \label{eq-action-of-drift-operator-on-h-eta}
            D h_\eta = \mi \nu h_\eta,
        \end{equation}
        where $ D $ is defined in \eqref{eq-definition-D-as-B-times-nabla} and $ \nu $ is given by \eqref{eq-definition-nu-as-sum-of-frequencies}.
    \end{lemma} 
    \begin{proof}
        The function $ p_k $ is a polynomial and $ g_\eta $ is smooth. Moreover,
        \begin{equation*}
            \abs{ \psi_0 (x) - \Re \eta }^2 = (\Re \psi_0 (x) - \Re \eta )^2 + ( \Im \psi_0 (x) )^2
        \end{equation*}
        is a real polynomial of degree at most four by \eqref{eq-definition-q-ell}. Hence $ h_\eta $ is smooth.

        To see that $ h_\eta\neq0 $, set
        \begin{equation*}
            u_* = (0_m, (1,0), \ldots, (1, 0)), \quad x_* = \sqrt{\frac{-\Re \eta}{q(u_*)}} u_*.
        \end{equation*}
        Since $ u_*\neq0 $, we have $ q(u_*)>0 $. Quadratic homogeneity gives $ q(x_*)=-\Re\eta $. The zero-block coordinate of $ x_* $ vanishes, so $ \ell(x_*)=0 $ and $ \psi_0(x_*)=-q(x_*)=\Re\eta $. Moreover, every factor in $ p_k(x_*) $ is a power of the positive number $ \sqrt{-\Re\eta/q(u_*)} $. Therefore,
        \begin{align*}
            p_k (x_*) &= \prod_{k_j \ge 0}  ( (x_*)_{j, 1} - \mi (x_*)_{j, 2} )^{k_j} \prod_{k_j < 0} ( (x_*)_{j, 1} + \mi (x_*)_{j, 2} )^{- k_j} \\ 
            &= \prod_{k_j \ge 0} \left( \sqrt{ \frac{-\Re \eta}{q(u_*)} } \right)^{k_j} \prod_{k_j < 0} \left( \sqrt{\frac{-\Re \eta}{q(u_*)} } \right)^{-k_j} = \left( \frac{-\Re \eta}{q(u_*)} \right)^{ \frac{1}{2} \sum_{j=1}^r \abs{k_j} } > 0.
        \end{align*}
        Since $ g_\eta(0)=1 $,
        \begin{equation*}
            h_\eta (x_*) = p_k (x_*) g_\eta (0) = p_k (x_*) > 0.
        \end{equation*}
        
        It remains to prove compactness of the support. If $ h_\eta(x)\neq0 $, then the support condition on $ g_\eta $ implies
        \begin{equation} \label{eq-absolute-value-difference-psi-0-and-Re-eta}
            \abs{ \psi_0 (x) - \Re \eta } < \epsilon_\eta.
        \end{equation}
        Consequently, 
        \begin{equation*}
            \abs{ -\Re \eta - q(x)} = \abs{ \Re( \psi_0 (x) - \Re \eta ) } \le \abs{ \psi_0 (x) - \Re \eta } < \epsilon_\eta,
        \end{equation*}
        and therefore $ q(x) < \epsilon_\eta - \Re \eta $. Recall that the support of $ h_\eta $ is the closure of the set where $ h_\eta \neq 0 $. By continuity of $ q $, we obtain
        \begin{equation} \label{eq-support-of-function-h-eta}
            \supp h_\eta \subset \{ x : q(x) \le \epsilon_\eta - \Re \eta \}.
        \end{equation}
        Since $ q $ is positive definite, the constant $ c_q:=\min_{\abs{u}=1}q(u) $ is strictly positive. By homogeneity, $ q(x)\ge c_q\abs{x}^2 $. Hence \eqref{eq-support-of-function-h-eta} gives
        \begin{equation*}
            \supp h_\eta \subset \left\{ x : \abs{x} \le \sqrt{ \frac{\epsilon_\eta - \Re \eta}{c_q} } \right\}.
        \end{equation*}
        The support is closed and bounded, hence compact. This proves $ h_\eta \in \cC_c^\infty (\R^{2d}) $.

        Finally, \eqref{eq-psi-0-constant} and \eqref{eq-action-of-drift-operator-on-p-k} give
        \begin{equation*}
            h_\eta (\me^{t \mathbf{B}} x) = p_k(\me^{t \mathbf{B}} x) g_\eta \left( \abs{ \psi_0 (\me^{t \mathbf{B}} x) - \Re \eta  }^2 \right) = \me^{\mi \nu t} p_k (x) g_\eta \left( \abs{ \psi_0 (x) - \Re \eta }^2 \right) = \me^{\mi \nu t} h_\eta (x),
        \end{equation*}
        which immediately implies \eqref{eq-action-of-drift-operator-on-h-eta}.
    \end{proof}

    \section{Proofs for purely critical nonsemisimple drift}
    \label{appendix-purely-critical-nonsemisimple}

    Throughout this section,
    $ \sigma(\mathbf{Z})\subset\mi\R $ and $ \mathbf{Z} $ is not
    semisimple.  We prove Theorem
    \ref{theorem-spectrum-purely-critical-nonsemisimple-drift}.
    For $ f\in\cS(\R^{2d}) $, let $ A_f $ denote the unique Schwartz
    operator satisfying
    \begin{equation*}
        \chi_{A_f}(\iota^{-1}(x))=f(x),
        \qquad x\in\R^{2d}.
    \end{equation*}
    Lemma \ref{lemma-global-schwartz-trace-norm-estimate} supplies the
    global inverse Fourier--Weyl estimate
    \eqref{eq-global-schwartz-trace-norm-estimate}.

    \begin{lemma}[A coordinate along the Jordan flow]
        \label{lemma-nonsemisimple-jordan-clock}
        There are $ \omega\in\R $ and nonzero complex row vectors
        $ u,v\in(\C^{2d})^* $ such that
        \begin{equation} \label{eq-nonsemisimple-left-jordan-chain}
            u\mathbf{Z}=\mi\omega u+v,
            \qquad
            v\mathbf{Z}=\mi\omega v.
        \end{equation}
        The set
        \begin{equation*}
            \cO:=\{x\in\R^{2d}:vx\ne0\}
        \end{equation*}
        is open, dense, invariant under $ (\me^{t\mathbf{Z}})_{t\in\R} $,
        and conic.  Moreover,
        \begin{equation} \label{eq-definition-nonsemisimple-clock}
            \tau(x):=
            \frac{\Re\bigl((ux)\overline{(vx)}\bigr)}{\abs{vx}^2},
            \qquad x\in \cO,
        \end{equation}
        is smooth and satisfies
        \begin{equation} \label{eq-nonsemisimple-clock-translation}
            \tau(\me^{t\mathbf{Z}}x)=\tau(x)+t,
            \qquad x\in \cO,\quad t\in\R.
        \end{equation}
    \end{lemma}

    \begin{proof}
        Since $ \mathbf{Z} $ is not semisimple, its transpose has a Jordan
        chain of length at least two at some eigenvalue $ \mi\omega $, which
        gives \eqref{eq-nonsemisimple-left-jordan-chain}.  It follows that
        \begin{equation} \label{eq-action-on-left-jordan-chain}
            v\me^{t\mathbf{Z}}=\me^{\mi\omega t}v,
            \qquad
            u\me^{t\mathbf{Z}}=\me^{\mi\omega t}(u+tv).
        \end{equation}
        The assertions about $ \cO $ are immediate; density follows because
        the restriction of the nonzero complex functional $ v $ to
        $ \R^{2d} $ is not identically zero.  Substitution of
        \eqref{eq-action-on-left-jordan-chain} into
        \eqref{eq-definition-nonsemisimple-clock} gives
        \eqref{eq-nonsemisimple-clock-translation}.
    \end{proof}

    Since the spectrum of $ \mathbf{Z} $ is purely imaginary, the Jordan
    normal form gives constants $ M_0>0 $ and an integer $ m\ge1 $ such
    that
    \begin{equation} \label{eq-polynomial-growth-purely-critical-nonsemisimple}
        \norm{\me^{t\mathbf{Z}}}
        \le M_0(1+\abs{t})^m,
        \qquad t\in\R.
    \end{equation}

    We extend every function in $ \cC_c^\infty(\cO) $ by zero outside $ \cO $.

    \begin{lemma}[Superexponential decay of localized operators]
        \label{lemma-nonsemisimple-superexponential-decay}
        For every $ f\in\cC_c^\infty(\cO) $, there are $ c_f>0 $ and a
        polynomial $ P_f $ such that
        \begin{equation} \label{eq-nonsemisimple-superexponential-decay}
            \norm{\cT_{*t}(A_f)}_1
            \le P_f(t)\me^{-c_ft^3},
            \qquad t\ge1.
        \end{equation}
    \end{lemma}

    \begin{proof}
        The assertion is immediate if $ f=0 $. Otherwise, let $ K:=\supp f\Subset \cO $, and put
        \begin{equation*}
            \delta:=\min_{y\in K}\abs{vy}>0,
            \qquad
            M:=\max_{y\in K}\abs{uy}.
        \end{equation*}
        By irreducibility, $ \mathbf{S}_1>0 $.  For $ k\in\N $ and $ y\in K $,
        \begin{align*}
            &\int_{-k-1}^{-k}
            (\me^{s\mathbf{Z}}y)^T\mathbf{C}(\me^{s\mathbf{Z}}y)\,\md s \\
            &\qquad
            =(\me^{-(k+1)\mathbf{Z}}y)^T\mathbf{S}_1
            (\me^{-(k+1)\mathbf{Z}}y)
            \ge\lambda_{\min}(\mathbf{S}_1)\abs{\me^{-(k+1)\mathbf{Z}}y}^2.
        \end{align*}
        On the other hand, \eqref{eq-action-on-left-jordan-chain} gives
        \begin{equation*}
            \abs{\me^{-(k+1)\mathbf{Z}}y}
            \ge\frac{\abs{u\me^{-(k+1)\mathbf{Z}}y}}{\abs{u}}
            \ge\frac{(k+1)\delta-M}{\abs{u}}.
        \end{equation*}
        Summing the complete unit intervals contained in $ [-t,0] $, and
        using $ \mathbf{C}\ge0 $ on the remaining interval, gives constants
        $ c_0>0 $ and $ t_0\ge1 $ such that
        \begin{equation} \label{eq-uniform-backward-cubic-diffusion}
            \int_{-t}^0
            (\me^{s\mathbf{Z}}y)^T\mathbf{C}(\me^{s\mathbf{Z}}y)\,\md s
            \ge c_0t^3,
            \qquad y\in K,\quad t\ge t_0.
        \end{equation}

        The characteristic-function formula shows that the Fourier--Weyl
        transform of $ \cT_{*t}(A_f) $ is
        \begin{equation} \label{eq-definition-nonsemisimple-G-t}
            G_t(x):=
            \exp{-\frac12x^T\mathbf{S}_tx+\mi\mathbf{b}_t^Tx}
            f(\me^{t\mathbf{Z}}x).
        \end{equation}
        Thus $ G_t\in\cC_c^\infty(\R^{2d}) $ and
        $ \supp G_t\subset\me^{-t\mathbf{Z}}K $.  If
        $ x=\me^{-t\mathbf{Z}}y $, $ y\in K $, then
        \begin{equation*}
            x^T\mathbf{S}_tx
            =\int_{-t}^0
            (\me^{s\mathbf{Z}}y)^T\mathbf{C}(\me^{s\mathbf{Z}}y)\,\md s,
        \end{equation*}
        so \eqref{eq-uniform-backward-cubic-diffusion} controls the Gaussian
        factor in \eqref{eq-definition-nonsemisimple-G-t}.

        For completeness, we control the finitely many derivatives required
        by \eqref{eq-global-schwartz-trace-norm-estimate}.  From
        \eqref{eq-polynomial-growth-purely-critical-nonsemisimple},
        \begin{equation*}
            \norm{\mathbf{S}_t}\le C(1+t)^{2m+1},
            \qquad
            \abs{\mathbf{b}_t}\le C(1+t)^{m+1},
        \end{equation*}
        and $ \abs{x}\le C_K(1+t)^m $ on
        $ \me^{-t\mathbf{Z}}K $.  For every multi-index $ \gamma $,
        \begin{equation*}
            \partial^\gamma
            \exp{-\frac12x^T\mathbf{S}_tx+\mi\mathbf{b}_t^Tx}
            =P_{\gamma,t}(x)
            \exp{-\frac12x^T\mathbf{S}_tx+\mi\mathbf{b}_t^Tx},
        \end{equation*}
        where $ P_{\gamma,t} $ is a polynomial in $ x $ whose coefficients
        are polynomial expressions in the entries of $ \mathbf{S}_t $ and
        $ \mathbf{b}_t $.  The chain and Leibniz rules, together with the
        preceding bounds, therefore show that, for every fixed pair
        $ \alpha,\beta $,
        \begin{equation*}
            \sup_{x\in\R^{2d}}
            \abs{x^\alpha\partial^\beta G_t(x)}
            \le P_{\alpha,\beta}(t)\me^{-c_0t^3/2}
        \end{equation*}
        for $ t\ge t_0 $ and a polynomial $ P_{\alpha,\beta} $. Applying
        \eqref{eq-global-schwartz-trace-norm-estimate} to $ G_t $, and using
        injectivity of the Fourier--Weyl transform, proves
        \eqref{eq-nonsemisimple-superexponential-decay} for $ t\ge t_0 $. Increasing the polynomial bound covers $ 1\le t\le t_0 $.
    \end{proof}

    The preceding lemma permits Laplace integrals beyond the usual
    resolvent half-plane.  For $ \lambda\in\C $ and
    $ f\in\cC_c^\infty(\cO) $, define
    \begin{equation} \label{eq-definition-nonsemisimple-laplace-vector}
        X_{\lambda,f}:=
        \int_0^\infty
        \me^{-\lambda t}\cT_{*t}(A_f)\,\md t.
    \end{equation}
    The integral converges in trace norm by
    \eqref{eq-nonsemisimple-superexponential-decay}.  For $ h>0 $, the
    semigroup property and a change of variables give
    \begin{equation*}
        \cT_{*h}(X_{\lambda,f})
        =\me^{\lambda h}
        \left(
        X_{\lambda,f}
        -\int_0^h\me^{-\lambda t}\cT_{*t}(A_f)\,\md t
        \right).
    \end{equation*}
    Taking the trace-norm difference quotient at zero yields
    \begin{equation} \label{eq-nonsemisimple-laplace-domain-identity}
        X_{\lambda,f}\in\dom\cL_*,
        \qquad
        (\lambda-\cL_*)X_{\lambda,f}=A_f.
    \end{equation}

    \begin{proof}[Proof of Theorem
        \ref{theorem-spectrum-purely-critical-nonsemisimple-drift}]
        We first construct approximate eigenvectors.  Fix
        $ \lambda=-a+\mi b $, where $ a>0 $.  Choose $ y_0\in \cO $ with
        $ \tau(y_0)=0 $; this is possible by \eqref{eq-nonsemisimple-clock-translation}.  Choose $ \delta>0 $ and
        $ g\in\cC_c^\infty(\cO) $ such that
        \begin{equation*}
            g\ge0,
            \qquad
            \supp g\subset\{x\in \cO:\abs{\tau(x)}<\delta\},
        \end{equation*}
        and $ s\mapsto g(\me^{s\mathbf{Z}}y_0) $ is not identically zero.
        Set
        \begin{equation} \label{eq-definition-nonsemisimple-test-function}
            f(x):=\me^{\mi b\tau(x)}g(x).
        \end{equation}

        For large $ T $, put
        \begin{equation} \label{eq-nonsemisimple-rescaling}
            \epsilon_T:=(1+T)^{-(m+1)},
            \qquad
            f_T(x):=f(x/\epsilon_T),
            \qquad
            x_T:=\epsilon_T\me^{-T\mathbf{Z}}y_0.
        \end{equation}
        If $ r_0 $ is the integer in
        \eqref{eq-global-schwartz-trace-norm-estimate}, then scaling gives
        \begin{equation} \label{eq-polynomial-bound-rescaled-test-operator}
            \norm{A_{f_T}}_1
            \le C\epsilon_T^{-r_0}
            \le C(1+T)^{(m+1)r_0}.
        \end{equation}
        Let $ X_T:=X_{\lambda,f_T} $.  By
        \eqref{eq-nonsemisimple-laplace-domain-identity},
        \begin{equation} \label{eq-residual-nonsemisimple-laplace-vector}
            (\lambda-\cL_*)X_T=A_{f_T}.
        \end{equation}

        We obtain a lower bound from the characteristic function.  Since
        \begin{equation*}
            f_T(\me^{t\mathbf{Z}}x_T)
            =f(\me^{(t-T)\mathbf{Z}}y_0),
        \end{equation*}
        the support condition on $ g $, together with
        \eqref{eq-nonsemisimple-clock-translation}, restricts the integral to
        $ t=T+s $, $ \abs{s}<\delta $.  Define
        \begin{equation*}
            q(r):=(\me^{r\mathbf{Z}}y_0)^T\mathbf{C}
            (\me^{r\mathbf{Z}}y_0),
            \qquad
            \ell(r):=\boldsymbol{\zeta}^T\me^{r\mathbf{Z}}y_0.
        \end{equation*}
        Using \eqref{eq-definition-nonsemisimple-test-function} and
        $ \tau(\me^{s\mathbf{Z}}y_0)=s $, we find
        \begin{align}
            \chi_{X_T}(\iota^{-1}(x_T))
            ={}&\me^{aT-\mi bT}
            \int_{-\delta}^{\delta}\me^{as}
            \exp{-\frac{\epsilon_T^2}{2}
            \int_{-T}^s q(r)\,\md r}
            \nonumber\\
            &\qquad\qquad\times
            \exp{\mi\epsilon_T\int_{-T}^s\ell(r)\,\md r}
            g(\me^{s\mathbf{Z}}y_0)\,\md s.
            \label{eq-characteristic-nonsemisimple-laplace-vector}
        \end{align}
        Uniformly for $ \abs{s}\le\delta $,
        \begin{equation*}
            0\le
            \epsilon_T^2\int_{-T}^s q(r)\,\md r
            \le C\epsilon_T^2(1+T)^{2m+1}
            \le\frac{C}{1+T}\longrightarrow0.
        \end{equation*}
        The possibly nonvanishing quantity
        $ \epsilon_T\int_{-T}^0\ell(r)\,\md r $ is a common phase, whereas
        \begin{equation*}
            \sup_{\abs{s}\le\delta}
            \abs{\epsilon_T\int_0^s\ell(r)\,\md r}
            \le C\epsilon_T\longrightarrow0.
        \end{equation*}
        After removing the common phase from
        \eqref{eq-characteristic-nonsemisimple-laplace-vector}, the integral
        converges to
        \begin{equation*}
            \int_{-\delta}^{\delta}
            \me^{as}g(\me^{s\mathbf{Z}}y_0)\,\md s>0.
        \end{equation*}
        Consequently, for some $ c_1>0 $ and all sufficiently large $ T $,
        \begin{equation} \label{eq-exponential-lower-bound-nonsemisimple-X-T}
            \norm{X_T}_1
            \ge\abs{\chi_{X_T}(\iota^{-1}(x_T))}
            \ge c_1\me^{aT}.
        \end{equation}
        For $ Y_T:=X_T/\norm{X_T}_1 $, equations
        \eqref{eq-polynomial-bound-rescaled-test-operator},
        \eqref{eq-residual-nonsemisimple-laplace-vector}, and
        \eqref{eq-exponential-lower-bound-nonsemisimple-X-T} give
        \begin{equation*}
            \norm{Y_T}_1=1,
            \qquad
            \norm{(\lambda-\cL_*)Y_T}_1
            \le C(1+T)^{(m+1)r_0}\me^{-aT}\longrightarrow0.
        \end{equation*}
        Thus every $ \lambda $ in the open left half-plane belongs to
        $ \sigma_{\mathrm{ap}}(\cL_*) $.  Since the approximate point
        spectrum is closed and $ \cL_* $ generates a contraction semigroup,
        \begin{equation*}
            \sigma_{\mathrm{ap}}(\cL_*)
            =\sigma(\cL_*)
            =\{\lambda\in\C:\Re\lambda\le0\}.
        \end{equation*}

        It remains to exclude eigenvalues.  Suppose that
        $ 0\ne Y\in\dom\cL_* $ and $ \cL_*Y=\lambda Y $.  Fix
        $ x\in \cO $.  The forward version of the unit-interval argument used
        in \eqref{eq-uniform-backward-cubic-diffusion}, now using
        \begin{equation*}
            \abs{u\me^{k\mathbf{Z}}x}
            =\abs{ux+k vx},
        \end{equation*}
        gives $ c_x>0 $ and $ t_x\ge1 $ such that
        \begin{equation} \label{eq-forward-cubic-diffusion-nonsemisimple}
            \int_0^t
            (\me^{s\mathbf{Z}}x)^T\mathbf{C}(\me^{s\mathbf{Z}}x)\,\md s
            \ge c_xt^3,
            \qquad t\ge t_x.
        \end{equation}
        The eigenvector identity and
        \eqref{eq-evolution-predual-semigroup-characteristic-function} imply
        \begin{align*}
            &\abs{\chi_Y(\iota^{-1}(\me^{t\mathbf{Z}}x))}\\
            &\qquad
            =\exp{\Re\lambda\,t
            +\frac12\int_0^t
            (\me^{s\mathbf{Z}}x)^T\mathbf{C}(\me^{s\mathbf{Z}}x)\,\md s}
            \abs{\chi_Y(\iota^{-1}(x))}.
        \end{align*}
        The left-hand side is bounded by $ \norm{Y}_1 $, while the
        exponential factor tends to infinity by
        \eqref{eq-forward-cubic-diffusion-nonsemisimple}.  Hence
        $ \chi_Y(\iota^{-1}(x))=0 $ for every $ x\in \cO $.  The set $ \cO $ is
        dense and $ \chi_Y $ is continuous, so $ \chi_Y=0 $.  Injectivity of
        the Fourier--Weyl transform gives $ Y=0 $, a contradiction.  Thus
        $ \sigma_{\mathrm p}(\cL_*)=\varnothing $.
    \end{proof}

    \section{Proofs for the reducible semigroups}
    \label{appendix-reducible-semigroups}

    We begin with a general observation about unitary conjugation on the
    trace-class space.

    \begin{lemma}
        \label{lemma-unitary-conjugation-trace-class}
        Let $ H $ be self-adjoint on a separable Hilbert space $ \mathsf{h} $.  Then
        \begin{equation*}
            \cT_{*t}(Y):=\me^{-\mi tH}Y\me^{\mi tH},
            \qquad t\in\R,\quad Y\in\cB_1(\mathsf{h}),
        \end{equation*}
        is a strongly continuous group of isometries.  If $ \cL_* $ denotes its
        generator, then
        \begin{equation}
            \label{eq-generator-rank-one-unitary-conjugation}
            \ketbra{u}{v}\in\dom\cL_* ,\qquad
            \cL_*\bigl(\ketbra{u}{v}\bigr)
            =-\mi\bigl(\ketbra{Hu}{v}-\ketbra{u}{Hv}\bigr),
            \qquad u,v\in\dom H.
        \end{equation}
        Moreover, $ \sigma(\cL_*)\subset\mi\R $.
    \end{lemma}

    \begin{proof}
        Unitary invariance of the trace norm gives
        $ \norm{\cT_{*t}(Y)}_1=\norm{Y}_1 $.  Strong continuity follows first for
        rank-one operators from the strong continuity of $ \me^{-\mi tH} $ and the
        identity
        \begin{equation*}
            \norm{\ketbra{u_1}{v_1}-\ketbra{u_2}{v_2}}_1
            \le \norm{u_1-u_2}\norm{v_1}
            +\norm{u_2}\norm{v_1-v_2},
        \end{equation*}
        and then for every trace-class operator by finite-rank approximation.
        Differentiation in trace norm gives
        \eqref{eq-generator-rank-one-unitary-conjugation}.  Finally, the usual
        Laplace formulas for a bounded group give the resolvent for
        $ \Re\lambda\ne0 $:
        \begin{equation*}
            (\lambda-\cL_*)^{-1}Y
            =\int_0^\infty \me^{-\lambda t}\cT_{*t}(Y)\,\md t,
            \qquad \Re\lambda>0,
        \end{equation*}
        and
        \begin{equation*}
            (\lambda-\cL_*)^{-1}Y
            =-\int_{-\infty}^0 \me^{-\lambda t}\cT_{*t}(Y)\,\md t,
            \qquad \Re\lambda<0.
        \end{equation*}
        Hence $ \sigma(\cL_*)\subset\mi\R $.
    \end{proof}

    \begin{lemma}
        \label{lemma-spectrum-free-particle-conjugation}
        Let $ H=p^2/2 $ be the free-particle Hamiltonian in the Schr\"odinger
        representation of the one-mode canonical commutation relations.  The generator $ \cL_* $ of
        \begin{equation*}
            \cT_{*t}(Y)
            :=\me^{-\mi tH}Y\me^{\mi tH}
        \end{equation*}
        satisfies
        \begin{equation*}
            \sigma(\cL_*)
            =\sigma_{\mathrm{ap}}(\cL_*)
            =\mi\R,
            \qquad
            \sigma_{\mathrm p}(\cL_*)=\varnothing.
        \end{equation*}
    \end{lemma}

    \begin{proof}
        Lemma \ref{lemma-unitary-conjugation-trace-class} gives
        $ \sigma(\cL_*)\subset\mi\R $.  Since
        $ \sigma(H)=[0,\infty) $ and a self-adjoint operator has approximate
        eigenvectors at every point of its spectrum, for arbitrary
        $ \alpha,\beta\ge0 $ there are unit vectors $ u_n,v_n\in\dom H $ such that
        \begin{equation*}
            \norm{(H-\alpha)u_n}\longrightarrow0,
            \qquad
            \norm{(H-\beta)v_n}\longrightarrow0.
        \end{equation*}
        Set $ Y_n:=\ketbra{u_n}{v_n} $.  Then $ \norm{Y_n}_1=1 $, and
        \eqref{eq-generator-rank-one-unitary-conjugation} gives
        \begin{equation*}
            \begin{split}
                \norm{\bigl(\cL_*
                +\mi(\alpha-\beta)\bigr)Y_n}_1
                &\le \norm{(H-\alpha)u_n}
                +\norm{(H-\beta)v_n}\longrightarrow0.
            \end{split}
        \end{equation*}
        Because $ \{\alpha-\beta:\alpha,\beta\ge0\}=\R $, this proves
        $ \mi\R\subset\sigma_{\mathrm{ap}}(\cL_*) $.

        For completeness, the point spectrum is empty.  Under the Fourier transform,
        $ H $ is multiplication by $ p^2/2 $.  A trace-class operator is
        Hilbert--Schmidt and hence has an $ L^2(\R^2) $ kernel $ K(p,q) $.  If
        $ \cL_*Y=\mi\theta Y $, then the corresponding
        Hilbert--Schmidt generator identity gives
        \begin{equation*}
            \left(-\frac{\mi}{2}(p^2-q^2)-\mi\theta\right)K(p,q)=0
            \quad\text{for almost every }(p,q)\in\R^2.
        \end{equation*}
        The level set $ p^2-q^2=-2\theta $ has two-dimensional Lebesgue measure zero,
        so $ K=0 $ almost everywhere and $ Y=0 $.
    \end{proof}

    \begin{proof}[Proof of Proposition
        \ref{proposition-same-drift-counterexamples}]
        For each of the two drift matrices in Table
        \ref{tab-same-drift-counterexamples},
        \begin{equation}
            \label{eq-counterexample-hamiltonian-drifts}
            \mathbf{Z}^T\mathbf{J}+\mathbf{J}\mathbf{Z}=0.
        \end{equation}
        Thus $ \mathbf{C}_{\mathbf{Z}}=\mathbf{C} $.  Both choices
        $ \mathbf{C}=0 $ and $ \mathbf{C}=c\1_2 $ satisfy
        \eqref{eq-quantum-diffusion-matrix} and therefore define Gaussian QMSs.  When $ \mathbf{C}=0 $, the whole complexified
        phase space is a $ \mathbf{Z} $-invariant subspace of
        $ \ker\mathbf{C}_{\mathbf{Z}} $, so the QMS is reducible.  When
        $ \mathbf{C}=c\1_2 $, the quantum diffusion matrix is strictly positive, and
        Proposition \ref{proposition-irreducibility} shows that the QMS is
        irreducible.

        If $ \mathbf{Z}=0 $ and $ \mathbf{C}=0 $, the QMS is the identity semigroup and
        its predual generator is the zero operator on all of
        $ \cB_1(\mathsf{h}) $; hence its spectrum is $ \{0\} $.  For
        $ \mathbf{C}=c\1_2 $, Proposition \ref{proposition-zero-drift-shape} gives
        $ \sigma(\cL_*)=(-\infty,0] $.

        For the second pair, let
        \begin{equation*}
            \mathbf{Z}=\begin{pmatrix}0&-1\\0&0\end{pmatrix}.
        \end{equation*}
        Writing $ z=x+\mi y $ and recalling
        $ q=(a+a^\dagger)/\sqrt2 $ and $ p=(a-a^\dagger)/(\mi\sqrt2) $, our Weyl
        convention gives
        $ W(z)=\exp{\mi\sqrt2(yq-xp)} $.  The free-particle Heisenberg evolution
        generated by $ H=p^2/2 $ satisfies $ q(t)=q+tp $ and $ p(t)=p $, and therefore
        \begin{equation*}
            \cT_t(W(z))=W(\me^{tZ}z),
            \qquad
            \me^{t\mathbf{Z}}=\begin{pmatrix}1&-t\\0&1\end{pmatrix}.
        \end{equation*}
        For $ \mathbf{C}=0 $, Lemma
        \ref{lemma-spectrum-free-particle-conjugation} gives
        $ \sigma(\cL_*)=\mi\R $.  For $ \mathbf{C}=c\1_2 $, the QMS is irreducible and
        the drift is purely critical and nonsemisimple.  Theorem
        \ref{theorem-spectrum-purely-critical-nonsemisimple-drift} gives
        \begin{equation*}
            \sigma(\cL_*)=\{\lambda\in\C:\Re\lambda\le0\}.
        \end{equation*}
        This proves all entries of Table \ref{tab-same-drift-counterexamples}.
    \end{proof}

    \begin{samepage}
    \section*{Use of artificial intelligence tools}

    During the preparation of this manuscript, the authors used ChatGPT to assist with language editing, organization of the exposition, and exploratory discussion of mathematical arguments. All mathematical statements, proofs, references, and conclusions were independently checked, revised, and approved by the authors, who take full responsibility for the content of the manuscript.
    \end{samepage}

\end{document}